%% file: ccsfp228-xiao.tex
\documentclass{sig-alternate-2013}
\usepackage{amsmath}
\usepackage{amsfonts}
\usepackage{amssymb}
\usepackage{multirow}
\usepackage{hyperref}
\usepackage{algorithm}
\usepackage{algpseudocode}
\usepackage{graphicx}
\usepackage{url}
\usepackage[usenames]{color}
\usepackage{colortbl}
\usepackage{subcaption}
\usepackage{enumerate}

\usepackage[
  pass,
]{geometry}

\newfont{\mycrnotice}{ptmr8t at 7pt}
\newfont{\myconfname}{ptmri8t at 7pt}
\let\crnotice\mycrnotice%
\let\confname\myconfname%

\newcommand{\RNum}[1]{\uppercase\expandafter{\romannumeral #1\relax}}
\makeatletter
\algnewcommand{\LineComment}[1]{\Statex \hskip\ALG@thistlm \(\triangleright\) #1}
\algdef{SE}[DOWHILE]{Do}{doWhile}{\algorithmicdo}[1]{\algorithmicwhile\ #1}
\makeatother
\begin{document}

\graphicspath{{figures/}}
\title{Protecting Locations with Differential Privacy under Temporal Correlations}

\numberofauthors{2}
\author{
\alignauthor
Yonghui Xiao\\
\affaddr{Dept. of Math and Computer Science}\\
       \affaddr{Emory University}
       \email{yonghui.xiao@emory.edu}\\
\alignauthor
Li Xiong\\
\affaddr{Dept. of Math and Computer Science}\\
       \affaddr{Emory University}
        \email{lxiong@emory.edu}
}

\maketitle
\begin{abstract}
Concerns on location privacy frequently arise with the rapid development of GPS enabled devices and location-based applications.
While spatial transformation techniques such as location perturbation or generalization have been studied extensively, most techniques rely on syntactic privacy models without rigorous privacy guarantee.
Many of them only consider static scenarios or perturb the location at
single timestamps without considering temporal correlations of a moving user's locations, and hence are vulnerable to various inference attacks.
While differential privacy has been accepted  as a standard for privacy protection, applying differential privacy in location based applications presents new challenges, as the protection needs to be enforced on the fly for a single user and needs to incorporate  temporal  correlations between a user's locations.

In this paper, we propose a systematic solution to preserve location privacy with rigorous privacy guarantee.
First, we propose a new definition, ``$\delta$-location set'' based differential privacy, to account for the temporal correlations in location data.
Second, we show that the well known $\ell_1$-norm sensitivity fails to capture the geometric sensitivity in multidimensional space and propose a new notion, sensitivity hull, based on which the error of differential privacy is bounded.
Third, to obtain the optimal utility we present a planar isotropic mechanism (PIM) for location perturbation, which is the first mechanism achieving the lower bound of differential privacy.
Experiments on real-world datasets also demonstrate that PIM significantly outperforms baseline approaches in data utility.
\end{abstract}
\category{C.2.0}{Computer-Communication Networks}{General}[Security and protection]
\category{K.4.1}{Computers and Society}{Public Policy Issues}[Privacy]

\keywords{Location privacy; Location-based services; Differential privacy; Sensitivity hull; Planar isotropic mechanism}

\newtheorem{theorem}{Theorem}[section]
\newtheorem{lemma}{Lemma}[section]
\newtheorem{definition}{Definition}[section]
\newtheorem{corollary}{Corollary}[section]
\newtheorem{observation}{Observation}[section]
\newtheorem{proposition}{Proposition}[section]
\newtheorem{example}{Example}[section]
\newtheorem{fact}{Fact}[section]
\input{data/locPriv01_v2}
\input{data/locPriv02}

\input{data/relatedwork_Li}

\vspace{3mm}
\begin{small}
\bibliographystyle{abbrv}
\bibliography{../exportlist,ref/privacy,ref/location_privacy,ref/crowdsourcing}
\end{small}

\input{data/locPriv_appendix}

\end{document}

%% file: data/locPriv01_v2.tex
\section{Introduction}
Technology and usage advances in smartphones with localization capabilities
have provided tremendous opportunities for location based applications.
Location-based services (LBS) \cite{junglas2008location,dey2010location} range from searching
points of interest to location-based games and location-based commerce.
Location-based social networks allow users to
share locations with friends, to find friends, and to provide
recommendations about points of interest based on their locations.

One major concern of location based applications is location
privacy \cite{beresford2003location}. To use these applications, users
have to provide their locations to the respective service providers or
other third parties.  This location disclosure raises important privacy concerns since
digital traces of users' whereabouts can expose them to attacks ranging
from unwanted location based spams/scams to blackmail or even physical danger.

\vspace{2mm}\noindent{\bf Gaps in Existing Works and New Challenges.}
Many location privacy protection mechanisms have been proposed during the last
decade \cite{krumm2009survey, ghinita2013privacy} in the setting of LBS or continual location sharing. In such setting, a user sends her location to untrusted service providers or other parties in order to
obtain some services (e.g. to find the nearest restaurant).
One solution is Private Information Retrieval (PIR)
technique, based on cryptography instead of revealing individual locations (e.g. \cite{papadopoulos2010nearest}). However, such technique tends to be computationally expensive and
not practical in addition to requiring different query plans to be designed for
different query types.

Most solutions proposed in the literature are based on location obfuscation which transforms the exact location of a user to an area
(location generalization) or a perturbed location (location perturbation) (e.g. \cite{gedik2008protecting,geo-indistinguishability-CCS13}). Unfortunately, most spatial transformation techniques proposed
so far rely on
syntactic privacy models such as k-anonymity, or ad-hoc uncertainty models, and
do not provide rigorous privacy.
Many of them only consider static scenarios or perturb the location at
single timestamps without considering the temporal correlations of a moving
user's locations, and hence are vulnerable to various inference attacks. 
Consider the following examples.
\begin{enumerate}[I]
\item
Suppose a user moved from school to the cafeteria (where ``$\star$'' is) in Figure \ref{Figure-map1} (left). Three perturbed locations were released by selecting a point probabilistically in each of the three circles (by some spatial cloaking methods). Even though the individual locations were seemingly protected at each timestamp, considering them together with road constraints or the user's moving pattern will enable an adversary to accurately figure out the user is in the cafeteria, resulting in privacy breach.
\item
Suppose a user's location ``$\star$'' is protected in a circle as shown in Figure \ref{Figure-map1} (right).
If by estimation based on previous locations
 the user can only be in the five places at current timestamp as shown in the figure, then the obfuscated location actually exposes the true location. Thus technically, the radius of the circle (in location obfuscation) should be subject to temporal correlations.
\end{enumerate}
While such temporal correlations can be commonly modeled by Markov chain \cite{Quantifying-location-privacy-SP2011,MaskIt-SIGMOD12,Liao-learning-pattern}, and few works have considered such Markov models \cite{Quantifying-location-privacy-SP2011, MaskIt-SIGMOD12}, it remains a challenge to provide rigorous privacy protection under temporal correlations for continual location sharing.

\begin{figure}
\centering
\includegraphics[height=3cm]{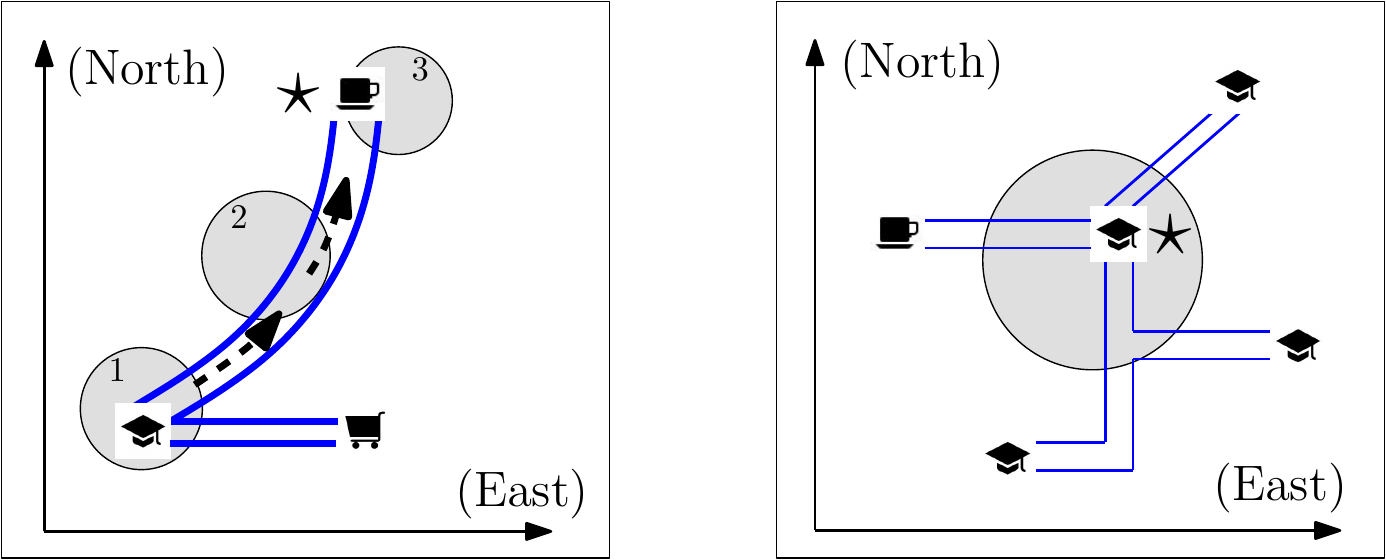}
\caption{{\small Examples of privacy breach caused by temporal correlations of user locations}}
\label{Figure-map1}
\end{figure}

Differential privacy \cite{dwork2006differential} has been accepted  as a standard
for privacy preservation. It was originally proposed to protect aggregated statistics of a dataset by bounding the knowledge gain of an adversary whether a user opts in or out of a dataset.
Applying differential privacy for location
protection is still at an early stage. In particular, several works (e.g. \cite{dp-trajectory-KDD12,
DBLP:conf/icde/QardajiYL13, DBLP:conf/dbsec/FanXS13})
have applied differential privacy on location or trajectory data but in a {\em data
publishing} or {\em data aggregation} setting. In this setting, a trusted data publisher with access to
a set of location snapshots or user trajectories publishes an {\em
aggregate} or synthetic view of the original data while guaranteeing user-level differential privacy, i.e. protecting the presence of a user's
location or entire trajectory in the aggregated data.

There are several challenges in applying differential privacy in the new setting of continual location sharing.  First,
standard differential privacy only protects {\em user-level} privacy (whether a user opts in or out of a dataset); while in our setting,
the protection needs to be enforced on the fly for {\em a single user}.
 Second, as shown in Figure \ref{Figure-map1}, temporal correlations
  based on road networks or the user's moving patterns
  exist and the privacy guarantee needs to account for such correlations.
 Finally, there is no effective location release mechanism with differential privacy under such model.


\vspace{2mm}\noindent{\bf Contributions.}
In this paper, we propose a systematic solution to preserve location privacy with differential privacy guarantee. As shown in Figure \ref{fig_overview}, we consider a moving user with sensitive location stream who needs to share her locations to an untrusted location-based application host or other parties.
A user's true locations are only known by the user.
The ``sanitized'' locations released by the privacy mechanisms are observable to the service providers, as well
as adversaries.
To enable private location sharing, we address (and take advantage of) the temporal correlations, which can not be concealed from adversaries and hence are assumed to be public. Our contributions are summarized as follows.

\begin{figure}
\vspace{-0.65cm}
    \includegraphics[width=0.5\textwidth]{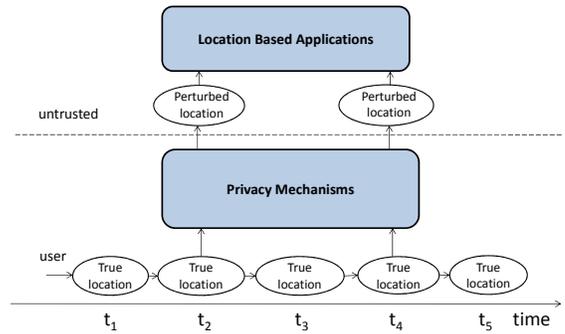}
  \vspace{-2.2cm}
    \caption{{\small Problem setting}}
    \label{fig_overview}
\end{figure}

First, we propose {``$\delta$-location set'' based differential privacy} to protect the true location at every timestamp.
The ``neighboring databases'' in standard differential privacy are any two databases under one operation: adding or removing a record (or a user).
However, this is not applicable in a variety of settings \cite{kifer2011no,Brodening-PET13}, which leads to new and extended notions such as  $\delta$-neighborhood \cite{Fang-neighborhood-CCS14} or event-level \cite{Dwork-continual-STOC10} differential privacy. In our problem, location changes between two consecutive timestamps are determined by temporal correlations modeled through a Markov chain \cite{Quantifying-location-privacy-SP2011, MaskIt-SIGMOD12}. Accordingly we propose a ``$\delta$-location set'' to include all probable locations (where the user might appear).
%
%
%
Intuitively, to protect the true location, we only need to ``hide'' it in the $\delta$-location set in which any pairs of locations are not distinguishable.

Second,
we show that the well known $\ell_1$-norm sensitivity in standard differential privacy fails to capture the geometric sensitivity in multidimensional space. Thus we propose a new notion, sensitivity hull, to capture the geometric meaning of sensitivity.  We also prove that the lower bound of error is determined by the sensitivity hull.

Third, we present an efficient location perturbation mechanism, called planar isotropic  mechanism (PIM), to achieve $\delta$-location set based differential privacy.
\begin{enumerate}[I]
\item
To our knowledge, PIM is the first optimal mechanism that can achieve the lower bound of differential privacy\footnote{The state-of-art differentially private mechanisms \cite{Geometry-Hardt-STOC10,Bhaskara-bound-STOC12} for linear queries can be $O(log(d))$ approximately optimal where $d$ is the number of dimensions. }. The novelty is that in two-dimensional space we efficiently transform the sensitivity hull to its isotropic position such that the optimality is guaranteed.
%
\item
We also implement PIM on real-world datasets, showing that it
preserves location utility for location based queries and
significantly outperforms the baseline Laplace mechanism (LM).
\end{enumerate}


%% file: data/locPriv02.tex
\section{Preliminaries}
We denote scalar variables by normal letters, vectors by bold lowercase letters, and matrices by bold capital letters.
We use $||\cdot||_p$ to denote the $\ell_p$ norm,
$\textbf{x}[i]$ to denote the $i$th element of $\textbf{x}$,
$\mathbb{E}()$ to denote the expectation,
$\textbf{x}^T$ to denote the transpose of vector $\textbf{x}$.
Table \ref{tbl-symbols} summarizes some important symbols for convenience.
\begin{table}
\centering
\begin{tabular}{|c|c|}
\hline
$\textbf{s}_i$& a cell in a partitioned map, {\small  $ i=1,2,\cdots,m$ }\\\hline
$\textbf{u}$, $\textbf{x}$& location in state and map coordinates\\\hline
$\textbf{u}^*,\ \textbf{x}^*$ & true location of the user\\\hline
$\textbf{z}$&the released location in map coordinate\\\hline
$\textbf{p}_t^-$& prior probability (vector) at timestamp $t$\\\hline
$\textbf{p}_t^+$& posterior probability (vector) at timestamp $t$\\\hline
$\Delta\textbf{X}$ & $\delta$-location set\\\hline
$K$ & sensitivity hull \\\hline
\end{tabular}
\caption{Denotation}
\label{tbl-symbols}
\end{table}
\subsection{Two Coordinate Systems}
We use two coordinate systems, state coordinate and map coordinate, to represent a location for the Markov model and map model respectively. Denote $\mathcal{S}$ the domain of space.
If we partition $\mathcal{S}$ into the finest granularity, denoted by ``cell'',
then $\mathcal{S}=\{\textbf{s}_1,\textbf{s}_2,\cdots, \textbf{s}_m\}$ where each $\textbf{s}_i$ is a unit vector with the $i$th element being $1$ and other $m-1$ elements being $0$.
Each cell can represent a state (location) of a user.
On the other hand,
If we view the space
as a map with longitude and latitude, then a $2\times 1$ vector can be used to represent a  user's location $\textbf{x}$ with two components $\textbf{x}[1]$ and $\textbf{x}[2]$.
Figure \ref{Figure-map} shows an example
using these two coordinate systems.
If a user is in $\textbf{s}_7$, the state coordinate and map coordinate are shown as follows. Note that the two coordinate systems can be transformed to each other. We skip how to transform them and treat $\textbf{u}$ and $\textbf{x}$ interchangeable.
\begin{align*}
\textbf{u}=\textbf{s}_7=\left[
\begin{array}{cccccccccc}
0&0&0&0&0&0&1&0&\cdots&0\\
\end{array}
\right]\\
\textbf{x}=[2,4]^T\ with\ \textbf{x}[1]=2\ and\ \textbf{x}[2]=4
\end{align*}
As time evolves, the trace of a user can be represented by a series of locations, $\textbf{x}_1,\textbf{x}_2,\cdots,\textbf{x}_t$ in map coordinate or $\textbf{u}_1,\textbf{u}_2,\cdots,\textbf{u}_t$ in state coordinate.

\begin{figure}
\centering
\includegraphics[width=5cm]{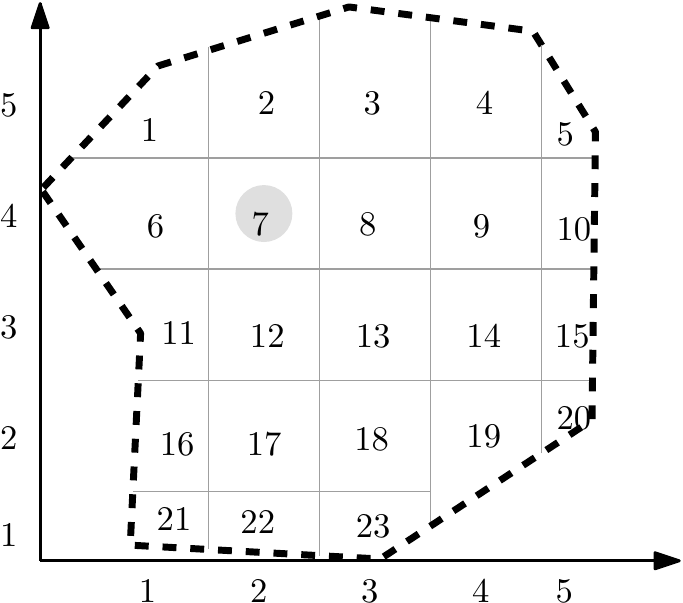}
\caption{{\small Two coordinate systems}}
\label{Figure-map}
\end{figure}


\subsection{Mobility and Inference Model}
Our approach uses Markov chain \cite{Quantifying-location-privacy-SP2011,MaskIt-SIGMOD12, Liao-learning-pattern} to model the temporal correlations between user's locations. Other constraints, such as road network, can also be captured by it. However, we note that Markov model, as well as any mobility models, may have limits in terms of predicability \cite{song2010limits}.  And we will discuss our solution to address these limits later.

In our problem setting, a user's true locations are unobservable, i.e. only known by the user.
The ``sanitized'' locations released by the perturbation mechanism are observable to the service provider, as well as adversaries.
Thus from an adversarial point of view, this process is a Hidden Markov Model (HMM).


At timestamp $t$, we use a vector $\textbf{p}_t$ to denote the probability distribution of a user's location (in each cell). Formally, \begin{align*}
\textbf{p}_t[i]=Pr(\textbf{u}^*_t=\textbf{s}_i)=Pr(\textbf{x}_t^*=the\ coordinate\ of\ \textbf{s}_i)
\end{align*}
where $\textbf{p}_t[i]$ is the $i$th element in $\textbf{p}_t$ and $\textbf{s}_i\in\mathcal{S}$.  In the example of Figure \ref{Figure-map}, if the user is located in cells $\{\textbf{s}_2,\textbf{s}_3,\textbf{s}_7,\textbf{s}_8\}$ with a uniform distribution, the probability vector can be expressed as follows.
\begin{align*}
\textbf{p}=\left[
\begin{array}{ccccccccccc}
0&0.25&0.25&0&0&0&0.25&0.25&0&\cdots&0\\
\end{array}
\right]
\end{align*}


\vspace{2mm}\noindent{\bf Transition Probability.}
We use a matrix $\textbf{M}$ to denote the probabilities that a user moves from one location to another.
Let $m_{ij}$ be the element in $\textbf{M}$ at $i$th row and $j$th column. Then $m_{ij}$ represents the probability that a user moves from cell $i$ to cell $j$. Given probability vector $\textbf{p}_{t-1}$, the probability at timestamp $t$ becomes
$\textbf{p}_{t}=\textbf{p}_{t-1}\textbf{M}$.
We assume the transition matrix $\textbf{M}$ is given in our framework.

\vspace{2mm}\noindent{\bf Emission Probability.}
If given a true location $\textbf{u}_t^*$, a mechanism releases a perturbed location $\textbf{z}_t$, then
the probability $Pr(\textbf{z}_{t}|\textbf{u}^*_{t}=\textbf{s}_i)$ is called ``emission probability'' in HMM.
This probability is determined by the release mechanism and should be transparent to adversaries.

\vspace{2mm}\noindent{\bf Inference and Evolution.}
At timestamp $t$, we use $\textbf{p}_t^-$ and $\textbf{p}_t^+$  to denote the prior and posterior probabilities of a user's location before and after observing the released $\textbf{z}_t$ respectively.
The prior probability can be derived by the posterior probability at  previous timestamp $t-1$
 and the Markov transition matrix as
$
\textbf{p}_{t}^-=\textbf{p}_{t-1}^+\textbf{M}
$.
Given $\textbf{z}_t$,
the posterior probability can be computed using Bayesian inference as follows. For each cell $\textbf{s}_i$:
\begin{align}
\textbf{p}_{t}^+[i]=Pr(\textbf{u}^*_{t}=\textbf{s}_i|\textbf{z}_{t})=
\frac{Pr(\textbf{z}_{t}|\textbf{u}^*_{t}=\textbf{s}_i)\textbf{p}_{t}^-[i]}{\mathop\sum\limits_{j}Pr(\textbf{z}_{t}|\textbf{u}^*_{t}=\textbf{s}_j)\textbf{p}_{t}^-[j]}
\label{eqn-posterior}
\end{align}

%

The inference at each timestamp can be efficiently computed by forward-backward algorithm in HMM, which will be incorporated in our framework.
%

\subsection{Differential Privacy and Laplace Mechanism}
\begin{definition}[Differential Privacy]
\label{def-stadard-dp}
A mechanism $\mathcal{A}$ satisfies $\epsilon$-differential privacy if for any output $\textbf{z}$ and any
\underline{neighboring databases} $\textbf{x}_1$ and $\textbf{x}_2$ where $\textbf{x}_2$ can be obtained from $\textbf{x}_1$ \underline{by either adding or removing one record}\footnote{This is the definition of unbounded differential privacy \cite{kifer2011no}. Bounded neighboring databases can be obtained by changing the value of exactly one record.}, the following holds
\begin{align*}
\frac{Pr(\mathcal{A}(\textbf{x}_1)=\textbf{z})}{Pr(\mathcal{A}(\textbf{x}_2)=\textbf{z})}\leq e^{\epsilon}
\end{align*}
\end{definition}

Laplace mechanism \cite{Dwork-calibrating} is commonly used in the literature to achieve differential privacy. It is built on the  $\ell_1$-norm sensitivity, defined as follows.
\begin{definition}[$\ell_1$-norm Sensitivity]
\label{def-standard-sensitivity}
For any query $f(\textbf{x})$: $\textbf{x}\rightarrow \mathbb{R}^d$, $\ell_1$-norm sensitivity is the \underline{maximum $\ell_1$ norm} of $f(\textbf{x}_1)-f(\textbf{x}_2)$ where $\textbf{x}_1$ and $\textbf{x}_2$ are any two instances in \underline{neighboring databases}.
\begin{align*}
S_f=\mathop{max}\limits_{\textbf{x}_1,\textbf{x}_2 \in  \textrm{ neighboring databases}}||f(\textbf{x}_1)-f(\textbf{x}_2)||_1
\end{align*}
where $||\cdot||_1$ denotes $\ell_1$ norm.
\end{definition}
A query can be answered by $f(\textbf{x})+Lap(S_f /\epsilon)$ to achieve $\epsilon$-differential privacy, where $Lap()\in \mathbb{R}^d$ are i.i.d. random noises drawn from Laplace distribution.

\subsection{Utility Metrics}
To measure the utility of the perturbed locations, we follow the analysis of metrics in \cite{Quantifying-location-privacy-SP2011} and adopt the expected distance (called ``correctness'' in \cite{Quantifying-location-privacy-SP2011}) between the true location $\textbf{x}^*$ and the released location $\textbf{z}$ as our utility metric. 
\begin{align}
\label{equation-util}
\textsc{Error}=\sqrt{\mathbb{E}||\textbf{z}-\textbf{x}^*||_2^2}
\end{align}

In addition, we also study the utility of released locations in the context of location based queries such as finding nearest $k$ Points of Interest (POI).  We will use precision and recall as our utility metrics in this context which we will explain later in the experiment section.
%

\subsection{Convex Hull}
Our proposed sensitivity hull is based on the well studied notion of convex hull in computational geometry.  We briefly provide the definition here.
\begin{definition}[Convex Hull]
Given a set of points $\textbf{X}=\{\textbf{x}_1,\textbf{x}_2,\cdots,\textbf{x}_n\}$, the convex hull of $\textbf{X}$ is the smallest convex set that contains $\textbf{X}$.
\end{definition}
Note that a convex hull in two-dimensional space is a polygon (also called ``convex polygon'' or ``bounding polygon'').
%
Because it is well-studied and implementations are also available \cite{O'Rourke:1998:CGC:521378}, we skip the details and only use $Conv(\textbf{X})$ to denote the function of finding the convex hull of $\textbf{X}$.

\section{Privacy Definition}
To apply differential privacy in the new  setting of continual location sharing,
we conduct a rigorous privacy analysis and
  propose $\delta$-location set based differential privacy in this section.


%

\subsection{$\boldsymbol\delta$-Location Set}
The nature of differential privacy is to ``hide'' a true database in ``neighboring databases'' when releasing a noisy answer from the database. In standard differential privacy, neighboring databases are obtained by either adding or removing a record (or a user) in a database.
 However, this is not applicable in our problem.
  Thus we propose a new notion, $\delta$-location set, to hide the true location at every timestamp.


\vspace{2mm}\noindent{\bf Motivations.}
We first discuss the intuitions that motivates our definition.

First, because the Markov model is assumed to be public, adversaries can make inference using previously released locations.
Thus we, as data custodians in a privacy mechanism, can also track the temporal inference at every timestamp.
At any timestamp, say $t$, a prior probability of the user's current location can be derived, denoted by $\textbf{p}_{t}^-$ as follows.
\begin{align*}
\textbf{p}_t^-[i]=Pr(\textbf{u}_t^*=\textbf{s}_i|\textbf{z}_{t-1},\cdots,\textbf{z}_1)
\end{align*}
Similar to hiding a database in its neighboring databases, we can
 hide the user's true location in possible locations (where $\textbf{p}_t^-[i]>0$). On the other hand, hiding the true location in any impossible locations (where $\textbf{p}_t^-[i]=0$) is a lost cause because the adversary already knows the user cannot be there.

Second, a potential shortcoming of Markov model is that the probability distribution may converge to a stationary distribution after a long time (e.g. an ergodic Markov chain). Intuitively, a user's possible locations can eventually cover the entire map given enough time. Hiding a location in a large area may yield a significantly perturbed location that is not useful at all.

According to \cite{nature-gonzalez08}, moving patterns of human have a ``high degree'' of temporal and spatial regularity. Hence if people tend to go to a number of highly frequented locations, our privacy notion should also emphasize protecting the more probable locations in Markov model.

\vspace{2mm}\noindent{\bf $\boldsymbol\delta$-Location Set.}
%
With above motivations, we define $\delta$-location set at any timestamp $t$, denoted as $\Delta\textbf{X}_t$.
Essentially, $\delta$-location set reflects
a set of probable locations the user might appear
(by leaving out the locations of small probabilities).

\begin{definition}[$\delta$-Location Set]
\label{def-indist-set}
Let $\textbf{p}_{t}^{-}$ be the prior probability of a user's location at timestamp $t$.
$\delta$-location set is a set containing minimum number of locations that have prior probability sum no less than $1-\delta$.
\begin{align*}
\Delta\textbf{X}_t = min \{\textbf{s}_i| \mathop\sum\limits_{\textbf{s}_i}\textbf{p}_t^-[i] \geq 1-\delta\}
\end{align*}
\end{definition}
For example, if $\textbf{p}_t^-=[0.3,0.4,0.05,0.2,0.03,0.02]$ corresponding to $[\textbf{s}_1,\textbf{s}_2,\textbf{s}_3,\textbf{s}_4,\textbf{s}_5,\textbf{s}_6]$, then $\Delta\textbf{X}=\{\textbf{s}_2,\textbf{s}_1,\textbf{s}_4\}$ when $\delta=0.1$; $\Delta\textbf{X}=\{\textbf{s}_2,\textbf{s}_1,\textbf{s}_4,\textbf{s}_3\}$ when $\delta=0.05$.

Note that if $\delta=0$ the location set contains all possible locations. Thus $0$-location set preserves the strongest privacy.

\vspace{2mm}\noindent{\bf Drift.}
Because $\delta$-location set represents the most probable locations,
a drawback is that the true location may be filtered out with a small probability (technically, $Pr(\textbf{x}^*\notin \Delta\textbf{X})=\delta$). Same situation may also occur if the Markov model is not accurate enough in practice due to its limit in predicability, as we mentioned earlier. Therefore, we denote this phenomenon as ``drift'' and handle it with the following surrogate approach.

\vspace{2mm}\noindent{\bf Surrogate.}
When a drift happens, we use a surrogate location in $\Delta\textbf{X}$ as if it is the ``true'' location in the release mechanism.

\begin{definition}[Surrogate]
A surrogate $\tilde{\textbf{x}}$ is the \mbox{cell} in $\Delta\textbf{X}$ with the shortest distance to the true location $\textbf{x}^*$.
\begin{align*}
\tilde{\textbf{x}}=\mathop{argmin}\limits_{\textbf{s}\in\Delta\textbf{X}} dist(\textbf{s},\textbf{x}^*)
\end{align*}
where function $dist()$ denotes the distance between two cells.
\end{definition}
Note that the surrogate approach does not leak any information of the true location, explained as follows. If $\textbf{x}^*\in \Delta\textbf{X}$, then $\textbf{x}^*$ is protected in $\Delta\textbf{X}$; if not, $\tilde{\textbf{x}}$ is protected in $\Delta\textbf{X}$. Using surrogate does not reveal whether $\textbf{x}^*$ is in $\Delta\textbf{X}$ or not. Because in any location release mechanisms $\textbf{x}^*$ is treated as a black box (oblivious to adversaries), replacing $\textbf{x}^*$ with $\tilde{\textbf{x}}$ is also a black box. We formally prove the privacy guarantee in Theorem \ref{theo-alg-frame}.



In some cases, a surrogate may be far from the true location. Then the released location may not be useful. Therefore,
we also measure the distance between released location and true location in our experiment to reflect the long-term effect of surrogate.

%

\subsection{Differential Privacy on $\boldsymbol\delta$-Location Set}
We define  differential privacy based on $\delta$-location set,
with the intuition that the released location $\textbf{z}_t$ will not help an adversary to differentiate any instances in the $\delta$-location set.
\begin{definition}[Differential Privacy]
\label{def-danamic-dp}
At any timestamp $t$,
a randomized mechanism $\mathcal{A}$ satisfies
$\epsilon$-differential privacy on $\delta$-location set $\Delta\textbf{X}_t$ if, for any output $\textbf{z}_t$
 and any two locations $\textbf{x}_1$ and $\textbf{x}_2$ in $\Delta\textbf{X}_t$,
the following holds:
\begin{equation}
\label{eqn-locpriv}
\frac{Pr(\mathcal{A}(\textbf{x}_1)=\textbf{z}_t)}{Pr(\mathcal{A}(\textbf{x}_2)=\textbf{z}_t)}\leq e^{\epsilon}
\end{equation}
\end{definition}

Above definition  guarantees the true location is always protected in $\delta$-location set at every timestamp.
In another word, the released location $\textbf{z}_t$ is differentially private at timestamp $t$ for continual location sharing under temporal correlations. For other application settings, like protecting the trace or trajectory of a user, we defer the investigation to future works.

%

\subsection{Adversarial Knowledge}
In reality, there may be a variety of adversaries with \mbox{all} kinds of prior knowledge.
Accordingly, we prove that for the problem of continual location sharing
differential privacy is equivalent to adversarial privacy, first studied in \cite{Rastogi-adversarial-privacy}.
\begin{definition}[Adversarial Privacy]
\label{def-advpriv}
A mechanism is $\epsilon$-adversarially private if for any location $\textbf{s}_i\in\mathcal{S}$, any output $\textbf{z}$ and any adversaries knowing the true location is in $\Delta\textbf{X}$, the following holds:
\begin{align}
\label{eqn-adv-priv}
\frac{Pr(\textbf{u}^*_t=\textbf{s}_i|\textbf{z}_t)}{Pr(\textbf{u}^*_t=\textbf{s}_i)}\leq e^{\epsilon}
\end{align}
where $Pr(\textbf{u}^*_t=\textbf{s}_i)$ and $Pr(\textbf{u}^*_t=\textbf{s}_i|\textbf{z}_t)$ are the prior and posterior probabilities of any adversaries.
\end{definition}
We can show Definition \ref{def-danamic-dp} is equivalent to adversarial privacy for continual location sharing, which can be derived from the PTLM property \cite{Rastogi-adversarial-privacy}.
\begin{lemma}
For the problem of continual location sharing, the following properties hold:
\begin{enumerate}[I]
\item
There exists only one true location at a timestamp:
\begin{align*}
Pr(\textbf{u}^*=\textbf{s}\cap \textbf{u}^*=\textbf{s}')=0, \textrm{ for any locations } \textbf{s},\textbf{s}'\in\mathcal{S}
\end{align*}
\item
Let $\textbf{S}\subseteq\mathcal{S}$ be an area and $Pr(\textbf{S})$ be the probability that the user is in $\textbf{S}$. For any two areas $\textbf{S}$ and $\textbf{S}'$,
\begin{align*}
Pr(\textbf{S})Pr(\textbf{S}')\geq Pr(\textbf{S}\cap \textbf{S}') Pr(\textbf{S}\cup \textbf{S}')
\end{align*}
\end{enumerate}
\end{lemma}

\begin{theorem}
\label{theo-equav}
For the problem of continual location sharing,
Definition \ref{def-danamic-dp} is equivalent to Definition \ref{def-advpriv}.
\end{theorem}

Definition \ref{def-advpriv} limits the information gain for adversaries knowing the condition $\textbf{x}_t^*\in\Delta\textbf{X}$. If $\textbf{x}_t^*\notin\Delta\textbf{X}$, our framework reveals no extra information due to the surrogate approach.
Thus adversarial knowledge can be bounded, discussed as follows.

\vspace{2mm}
\noindent{\bf Standard Adversary.} For adversaries who have exactly the same Markov model and keep tracking all the released locations, their knowledge is also the same as our model (with location inference in Section \ref{sec-loc-infr}). In this case,
differential privacy and adversarial privacy are guaranteed, and we know exactly the adversarial knowledge, which in fact can be controlled by adjusting $\epsilon$.

\vspace{2mm}
\noindent{\bf Weak Adversary.} For adversaries who have little knowledge about the user, the released locations may help them obtain more information. With enough time to evolve, they may converge to standard adversaries eventually. But their adversarial knowledge will not exceed standard adversaries.

\vspace{2mm}
\noindent{\bf Strong Adversary.} For adversaries who have additional information, the released location from  differential privacy may not be very helpful. Specifically, a strong adversary with auxiliary information may have more accurate prior knowledge. However, if the adversary cannot identify the true location so that $Pr(\textbf{u}^*=\textbf{s}_i)=1$ for any $\textbf{s}_i\in\mathcal{S}$, Definition \ref{def-advpriv} is always satisfied.
%
On the other hand, if an ``omnipotent'' adversary already knows the true location, then
no mechanism can actually protect location privacy.
%

\subsection{Comparison with Other Definitions}
\label{sec-comparison}
\noindent{\bf Differential Privacy.}
Since the concept of neighboring databases is not generally applicable (as discussed earlier), induced neighborhood \cite{kifer2011no}, metric based neighborhood \cite{Brodening-PET13} and $\delta$-neighborhood \cite{Fang-neighborhood-CCS14} were proposed. The general idea is that the neighborhood can be formulated by some constraints of data or distance (metric) functions instead of adding or removing a record. However, applying these neighborhood based differential privacy is not feasible in our model because there is only one sole tuple (location) at each timestamp without any ``neighbors''. Hence we define $\delta$-location set to extend the notion of ``neighborhood''.

\vspace{2mm}\noindent{\bf Geo-indistinguishability.}
Another closely related definition is the Geo-indistinguishability \cite{geo-indistinguishability-CCS13}, which protects a user's location within a radius (circle) with a ``generalized differential privacy'' guarantee.
In other words, the neighborhood is defined with Euclidian distance.
Nevertheless, such spatial perturbation technique may not be reasonable in reality. For example, as shown in Figure \ref{Figure-map1}, the ``generalized differential privacy'' can still be breached given the road network constraint or user's moving pattern (which is represented by Markov model).
Thus location privacy must be protected under temporal correlations.

\vspace{2mm}\noindent{\bf Blowfish privacy.}
Our privacy definition shares the same insight as the unconstrained Blowfish privacy framework \cite{Blowfish-SIGMOD14} in statistical
data release context, which uses secret pairs and privacy policy to build a subset of possible
database instances as ``neighbors''. We show that $\delta$-location set based differential privacy can be instantiated as a special case of unconstrained Blowfish privacy at each timestamp.
\begin{theorem}
Let $\mathcal{S}$ be the domain of all possible locations. Let $G$ be a complete graph where each node denotes a location in $\mathcal{S}$. Let $\Delta\textbf{X}$ be a condition such that $\textbf{x}^*\in\Delta\textbf{X}$. At each timestamp, Definition \ref{def-danamic-dp} is equivalent to $\{\epsilon,\{\mathcal{S},G,\Delta\textbf{X}\}\}$-Blowfish privacy.
\end{theorem}
%

\subsection{Discussion}
\noindent{\bf Learning Markov Model.}
Existing methods such as the knowledge construction module in \cite{Quantifying-location-privacy-SP2011} or EM method in HMM can be used to acquire the transition matrix $\textbf{M}$, which will not be discussed in this paper. However, depending on the power of adversaries, two typical $\textbf{M}$ can be learned.
\begin{enumerate}[I]
\item
Popular $\textbf{M}$ can be learned from public transit data.
\item
Personal $\textbf{M}$ can be derived
with personal transit data\footnote{For example, mobile apps, like Google Now, may have a user's location history to derive the user's moving pattern.}.
\end{enumerate}
No matter which $\textbf{M}$ is adopted in our framework,
the adversarial knowledge is always bounded, as discussed before.
 However, the usefulness of released locations may vary for different adversaries. We also compare the two models in our experiments.

%

\vspace{2mm}
\noindent{\bf When Markov Model is not Accurate.}
When the location data does not exactly follow a Markov model or the learned Markov model is not accurate enough, our framework still guarantees differential privacy, but may generate more error. For example, when a user starts to go to a new place for the first time, which has not been reflected in the learned Markov model, a drift happens. In this case, our mechanism still works because we can handle the drift case. Nevertheless, the utility may be downgraded, especially when the new place is far from the ``probable'' locations.

%
\vspace{2mm}\noindent{\bf Composibility.}
Since we only need to release one perturbed location at a timestamp, the sequential composition \cite{McSherry-PINQ} is not applicable. {Otherwise, for multiple releases at a timestamp the composition of $\epsilon$ holds.}
%
On the other hand, given a series of perturbed locations $\{\textbf{z}_1,\textbf{z}_2,\cdots,\textbf{z}_t\}$ released from timestamp $1$ to $t$, a new problem is how to protect and measure the overall privacy guarantee of the entire trace.  We defer this to future work.

\section{Sensitivity Hull}
\label{sec-sensitivityhull}
The notion of sensitivity indicates the differences between any two query answers from two instances in neighboring databases. However, in multidimensional space, we show that $\ell_1$-norm sensitivity (in Definition \ref{def-standard-sensitivity}) fails to capture the exact sensitivity. Thus we propose a new notion, sensitivity hull. Note that sensitivity hull is an independent notion from the context of location privacy and can be plugged in any data-independent perturbation mechanisms.
\subsection{Sensitivity Hull}
\label{sec-shull}
To derive the meaning of sensitivity, let us consider the following example in traditional setting of differential privacy.
\begin{example}
\label{example-DPCHull3}
Assume we have an employee table $T$ with attributes gender and income. Then we answer the following query workload $f$:
\begin{flalign*}
&f_1:\ Select\ count(*)\ from\ T\ where\ gender=``female"\\
&f_2:\ Select\ count(*)\ from\ T\ where\ income>50000
\end{flalign*}
\end{example}
Let $\textbf{x}_1$ and $\textbf{x}_2$ be neighboring databases so that  $\textbf{x}_1$ is equal to $\textbf{x}_2$ adding or removing {\em a random user}.  Suppose $f(\textbf{x}_2)=[10,20]^T$. Then the possible answers for $f(\textbf{x}_1)$ could be
one of the following columns, from which $\Delta f$ can be derived.
\begin{flalign*}
\hspace{.5cm}
&f(\textbf{x}_1)=
\small
\left[
\begin{array}{ccccccc}
11 & 10 & 10 & 11 & 9 & 9 & 10\\
21 & 21 & 20 & 20 & 20& 19 & 19
\end{array}
\right]&
\\
&\Delta f
=f(\textbf{x}_1)-f(\textbf{x}_2)
\small
=\left[
\begin{array}{ccccccc}
1 & 0 & 0 & 1 & -1 & -1 & 0\\
1 & 1 & 0 & 0 & 0 & -1 & -1
\end{array}
\right]&
\\
&S_f=max||\Delta f||_1=2 \ \ (\ell_1 \textrm{-norm sensitivity})&
\end{flalign*}

\begin{figure}
\begin{subfigure}{0.48\textwidth}
\centering
\includegraphics[width=5cm]{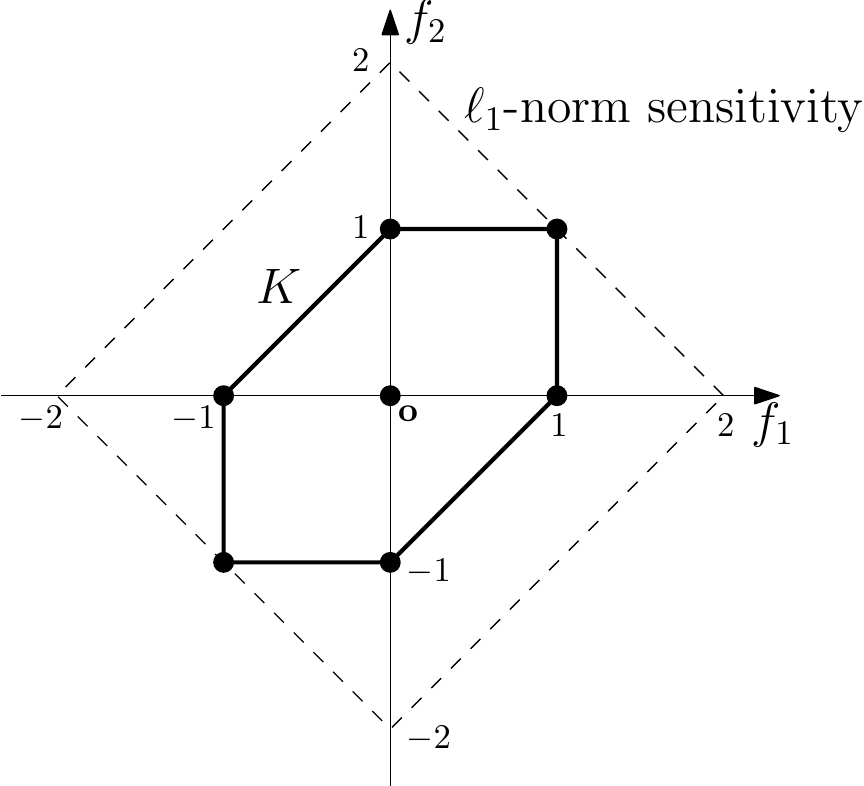}
\label{Figure-DPCHull}
\end{subfigure}
\caption{{\small Sensitivity hull of Example \ref{example-DPCHull3}. Solid lines denote the sensitivity hull $K$; dashed lines are the $\ell_1$-norm sensitivity. }}
\label{Figure-DPCHull3}
\end{figure}


In Figure \ref{Figure-DPCHull3}, the dashed lines form the set of $||\Delta f||_1=2$
because the $\ell_1$-norm sensitivity is $2$. However, $\Delta f$ only consists of all the ``$\bullet$'' points. It is obvious that the $\ell_1$-norm sensitivity exaggerates the ``real sensitivity''.
To capture the geometric representation of $\Delta f$ in multidimensional space, we define sensitivity hull (the solid lines in Figure \ref{Figure-DPCHull3}) as follows.
\begin{definition}[Sensitivity Hull]
\label{def-shull1}
The sensitivity \\hull of a query $f$ is the convex hull of $\Delta f$ where $\Delta f$ is the set of $f(\textbf{x}_1)-f(\textbf{x}_2)$ for any pair $\textbf{x}_1$ and $\textbf{x}_2$ in $\delta$-location set $\Delta\textbf{X}$.
\begin{flalign*}
\hspace{2.4cm}
K&=Conv\left(\Delta f\right)&
\\
\Delta f&=\mathop\cup\limits_{\textbf{x}_1,\textbf{x}_2\in \Delta\textbf{X}}
\left( f({\textbf{x}_1})-f({\textbf{x}_2}) \right)&
\end{flalign*}
\end{definition}
\begin{theorem}
A sensitivity hull $K$ is centrally symmetric: if $\textbf{v}\in K$ then $-\textbf{v}\in K$.
\end{theorem}
\begin{theorem}
If data $\textbf{x}$ is in discrete domain, then for any $f: \textbf{x}\rightarrow \mathbb{R}^d$, the sensitivity hull of $f$ is a polytope in $\mathbb{R}^d$.
\end{theorem}

\subsection{Error Bound of Differential Privacy}
We extend the error bound of differential privacy in database context \cite{Geometry-Hardt-STOC10} to our location setting using sensitivity hull.

\begin{lemma}
\label{lemma-query-to-K}
Suppose $\textbf{F}:\mathbb{R}^N\rightarrow\mathbb{R}^d$ is an aggregate function.
When neighboring databases are obtained by changing an attribute, the sensitivity hull $K$ of $\textbf{F}$ is a polytope $\textbf{F}\textbf{B}_1^N$ where $\textbf{B}_1^N$ is the $N$-dimensional unit $\ell_1$ ball.
\end{lemma}
\begin{proof}
Because the true data $\textbf{x}\in \mathbb{R}^N$, any change (by changing an attribute) of $\textbf{x}^N$ will result in a $\pm 1$ operation in one of the $N$ dimensions. Thus the convex hull for \textbf{x} changes is $\textbf{B}_1^N$. The changes in $\textbf{x}$ will cause  variations in output domain by $\textbf{F}\textbf{B}_1^N$. Because linear transformations keep the convexity, then $K=\textbf{F}\textbf{B}_1^N$.
\end{proof}
\begin{lemma}[\cite{Geometry-Hardt-STOC10}]\footnote{To check the correctness, we refer readers to the latest version of \cite{Geometry-Hardt-STOC10} at \url{http://mrtz.org/papers/HT10geometry.pdf}.}
\label{lemma-opt-error}
To answer a linear function $\textbf{F}:\mathbb{R}^N\rightarrow\mathbb{R}^d$ under Definition \ref{def-stadard-dp}, every $\epsilon$-differentially private mechanism must have
\begin{align*}
\textsc{Error}\geq \Omega \left( \frac{d}{\epsilon}\left(\frac{\textsc{Vol}(K)}{\textsc{Vol}(B_2^d)}\right)^{1/d} \right)
\end{align*}
where $\textsc{Vol}(\cdot)$ is the Volume and $B_2^d$ is the unit $\ell_2$ ball.
\end{lemma}

Then we can derive the lower bound of dynamic differential privacy as follows.
\begin{theorem}[Lower Bound]
\label{theorem-lowerbound}
Let $K$ be the sensitivity hull of $\delta$-location set  $\Delta\textbf{X}$.
To satisfy Definition \ref{def-danamic-dp}, every mechanism must have
\begin{align*}
\textsc{Error}\geq \Omega \left( \frac{1}{\epsilon}\sqrt{\textsc{Area}(K)} \right)
\end{align*}
where $\textsc{Area}(K)$ is the area of  $K$.
\end{theorem}
\begin{proof}
Because $\textbf{x}$ is in discrete domain (from Markov states), $K$ is a polygon (two-dimensional polytope).
The number of vertices must be an  even number $2N$ because $K$ is symmetric. There must be a matrix $\textbf{F}\in\mathbb{R}^{2\times N}$ that satisfies $K=\textbf{F}\textbf{B}_1^N$ where $\textbf{B}_1^N$ is the $N$-dimensional unit $\ell_1$ ball.
Therefore, answering $\textbf{F}$ is the same problem as locating a point from $K$, which is our geometric location problem.
In our setting, the output domain is two-dimensional. The volume also becomes area of $K$. Then we obtain above lower bound.
\end{proof}

\section{Location Release Algorithm}
\subsection{Framework}
%
%
The framework of our proposed location release algorithm is shown in Algorithm \ref{alg-framework}.
At each timestamp, say $t$, we compute the prior probability vector $\textbf{p}_t^-$.
If the location needs to be released,
 we construct a $\delta$-location set $\Delta\textbf{X}_t$. Then if the true location $\textbf{x}^*$ is excluded in $\Delta\textbf{X}_t$ (a drift), we use surrogate to replace $\textbf{x}^*$. Next a differentially private mechanism (like Algorithm \ref{alg-KNorm} which will be presented next) can be adopted to release a perturbed location $\textbf{z}_t$. In the meantime, the released $\textbf{z}_t$ will also be used to update the posterior probability $\textbf{p}_t^+$ (in the equation below) by Equation (\ref{eqn-posterior}), which subsequently will be used to compute the prior probability for the next timestamp $t+1$. Then at timestamp $t+1$, the above process is repeated.
\begin{align*}
\textbf{p}_t^+[i]=Pr(\textbf{u}_t^*=\textbf{s}_i|\textbf{z}_t,\textbf{z}_{t-1},\cdots,\textbf{z}_1)
\end{align*}

%

\begin{algorithm}[htb]
\caption{Framework}
\begin{algorithmic}[1]
\Require{
$\epsilon_t$, $\delta$, $\textbf{M}$, $\textbf{p}_{t-1}^+$, $\textbf{x}_t^*$
}
\State{$\textbf{p}_{t}^-\gets\textbf{p}_{t-1}^+\textbf{M}$;}
\Comment{{\tt \scriptsize Markov transition}}
\If{location needs to be released}
\State{Construct $\Delta \textbf{X}_t$;}
\Comment{{\tt \scriptsize $\delta$-location set}}
\If{$\textbf{x}_t^*\notin\Delta\textbf{X}_t$}
\Comment{{\tt \scriptsize a drift}}
\State{$\textbf{x}_t^*\gets$ surrogate;}
\EndIf

\State{$\textbf{z}_t\gets$\Call{Algorithm \ref{alg-KNorm}}{$\epsilon_t$, $\Delta\textbf{X}_t$, $\textbf{x}_t^*$};}
\label{line-releasing}
\Comment{{\tt \scriptsize release $\textbf{z}_t$}}
\State{Derive posterior probability $\textbf{p}_t^+$ by Equation (\ref{eqn-posterior});}
\label{line-bayesian}
\EndIf
\\\Return{\Call{Algorithm \ref{alg-framework}}{$\epsilon_{t+1}$, $\delta$, $\textbf{M}$, $\textbf{p}_{t}^+$, $\textbf{x}_{t+1}^*$}};\LineComment{{\tt \scriptsize go to next timestamp}}

\end{algorithmic}
\label{alg-framework}
\end{algorithm}
\begin{theorem}
\label{theo-alg-frame}
At any timestamp $t$, Algorithm \ref{alg-framework} is $\epsilon_t$-differentially private on $0$-location set.
\end{theorem}
\begin{proof}
It is equivalent to prove adversarial privacy on $0$-location set, which includes all possible locations.
If $\textbf{x}_t^*\in\Delta\textbf{X}_t$, then $\textbf{z}_t$ is generated by $\textbf{x}_t^*$. By Theorem \ref{theo-priv-KNorm}, $\textbf{z}_t$ is $\epsilon_t$-differentially private. So
$
\frac{Pr(\textbf{u}^*_t=\textbf{s}_i|\textbf{z}_t)}{Pr(\textbf{u}^*_t=\textbf{s}_i)}\leq e^\epsilon
$.
When $\textbf{x}_t^*\notin\Delta\textbf{X}_t$, then a surrogate $\tilde{\textbf{x}}_t$ replaces $\textbf{x}_t^*$. Then
\begin{align*}
\frac{Pr(\textbf{u}^*_t=\textbf{s}_i|\textbf{z}_t)}{Pr(\textbf{u}^*_t=\textbf{s}_i)}
=\frac{ \mathop\sum_k Pr(\textbf{u}^*_t=\textbf{s}_i | \tilde{\textbf{x}}_t=\textbf{s}_k)Pr(\tilde{\textbf{x}}_t=\textbf{s}_k|\textbf{z}_t)  }{ \mathop\sum_k Pr(\textbf{u}^*_t=\textbf{s}_i |\tilde{\textbf{x}}_t=\textbf{s}_k)Pr(\tilde{\textbf{x}}_t=\textbf{s}_k)}
\leq e^{\epsilon}
\end{align*}
Therefore, by equivalence (Theorem \ref{theo-equav}) Algorithm \ref{alg-framework} is $\epsilon_t$-differentially private on $0$-location set.
\end{proof}
%

\noindent{\bf Laplace Mechanism.}
With the $\ell_1$-norm sensitivity in Definition \ref{def-standard-sensitivity}, Laplace mechanism (LM) can be adopted in Line \ref{line-releasing} of Algorithm \ref{alg-framework}. The problem of this approach is that it will over-perturb a location because $\ell_1$-norm sensitivity could be much larger than the sensitivity hull, as discussed in Section \ref{sec-sensitivityhull}. We use LM with $\delta$-location set as a baseline in our experiment.

\subsection{Planar Isotropic Mechanism}
Because we showed (in Lemma \ref{lemma-query-to-K}) that the sensitivity hull of a query matrix is a polytope (polygon in our two-dimensional location setting),
the state-of-art $K$-norm based mechanism \cite{Geometry-Hardt-STOC10,Bhaskara-bound-STOC12,Nikolov-geometry-STOC13} can be used.
\begin{definition}[K-norm Mechanism \cite{Geometry-Hardt-STOC10}]
Given a linear function $\textbf{F}:\mathbb{R}^N\rightarrow\mathbb{R}^d$ and its sensitivity hull $K$, a mechanism is $K$-norm mechanism if for any output $\textbf{z}$, the following holds:
\begin{align}
\label{eqn-pdf-K-Norm}
Pr(\textbf{z})=\frac{1}{\Gamma(d+1)\textsc{Vol}(K/\epsilon)}exp \left( -\epsilon||\textbf{z}-\textbf{Fx}^*||_K \right)
\end{align}
where $\textbf{Fx}^*$ is the true answer, $||\cdot||_K$ is the (Minkowski) norm of $K$, $\Gamma()$ is Gamma function and $\textsc{Vol}()$ indicates volume.
\end{definition}

However, standard $K$-norm mechanism was designed for high-dimensional structure of sensitivity hull, whereas in our problem a location is only two-dimensional. Thus we can further optimize $K$-norm mechanism to achieve the lower bound of differential privacy. We propose a Planar Isotropic Mechanism (PIM) based on $K$-norm mechanism as follows.

\vspace{2mm}
\noindent{\bf Rationale.}
The rationale of PIM is that in two-dimensional space we efficiently transform the sensitivity hull to its isotropic position\footnote{We refer readers to \cite{Giannopoulos-isotropic-03,Milman-isotropic-89} for a detailed study of isotropic position.}
so that the optimality is guaranteed.

%

\begin{theorem}\cite{Geometry-Hardt-STOC10}
\label{error-K-norm}
If the sensitivity hull $K$ is in $C$-appro-ximately isotropic position, then $K$-norm mechanism has error $O(C)\textsc{LB}(K)$ where $\textsc{LB}(K)$ is the lower bound of differential privacy.
\end{theorem}

From Theorem \ref{error-K-norm}, we know that $K$-norm mechanism would be the optimal solution if the sensitivity hull $K$ is in isotropic position, denoted by $K_I$.
Although in high-dimensional space transforming a convex body to its isotropic position is extremely expensive,
it is feasible in two-dimensional space. To this end,  we need the following corollary
(which can be derived from \cite{Rudelson99randomvectors,Lovasz-2006-annealing}).
%
\begin{corollary}[Isotropic Transformation]
\label{theo-iso-sample}
For any convex body $K$ in $\mathbb{R}^2$, any integer $p\geq 1$, there is an absolute constant $c$ such that if $l\geq 4cp^2$, with probability at least $1-2^{-p}$, $K_I=\textbf{T}K$ is in isotropic position.
\begin{align}
\label{eqn-iso-T}
\textbf{T}=\left(
\frac{1}{l}
\sum_{i=1}^{l}\textbf{y}_i \textbf{y}_i^T
\right)^{-\frac{1}{2}}
\end{align}
where $\textbf{y}_1,\textbf{y}_2,\cdots,\textbf{y}_l$ are independent random points  uniformly distributed in $K$.
\end{corollary}

Therefore, the isotropic transformation of any sensitivity hull $K$ can be fulfilled by sampling, which is a trivial task in two-dimensional space.
For instance, a hit-and-run algorithm \cite{Lovasz-2004-hit-and-run} only takes $O(log^3(1/\delta))$ time where $\delta$ is an error parameter. We skip the sampling details and refer readers to the survey paper of Santosh Vempala \cite{Vempala05geometricrandom} for a complete study.

\begin{figure*}[!ht]
\begin{subfigure}{0.33\textwidth}
\centering
\vspace{0.5cm}
\includegraphics[width=3.7cm]{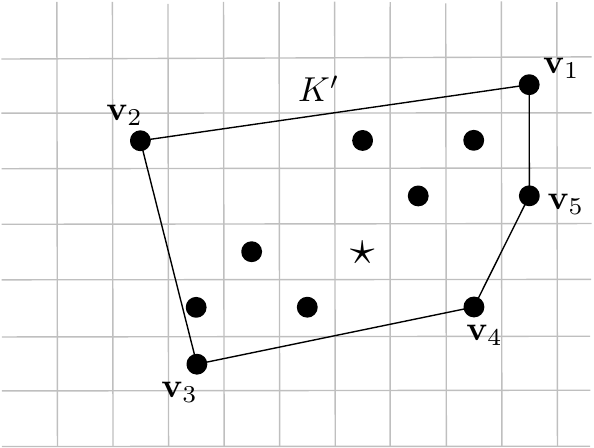}
\vspace{0.53cm}
\caption{}
\label{Figure-KNorm}
\end{subfigure}
\begin{subfigure}{0.33\textwidth}
\centering
\includegraphics[width=5.5cm]{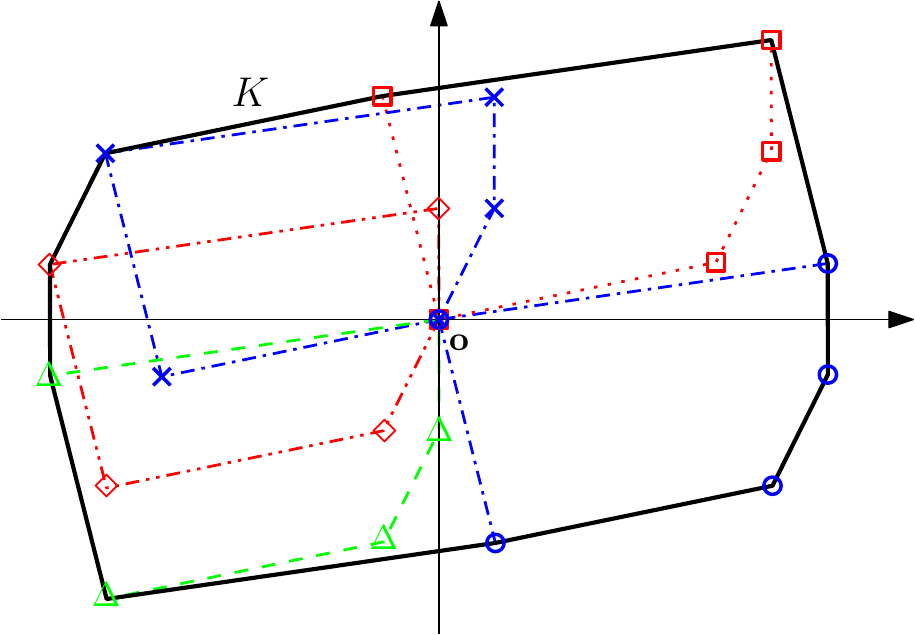}
\caption{}
\label{Figure-KNorm5}
\end{subfigure}
\begin{subfigure}{0.33\textwidth}
\centering
\includegraphics[width=4.5cm]{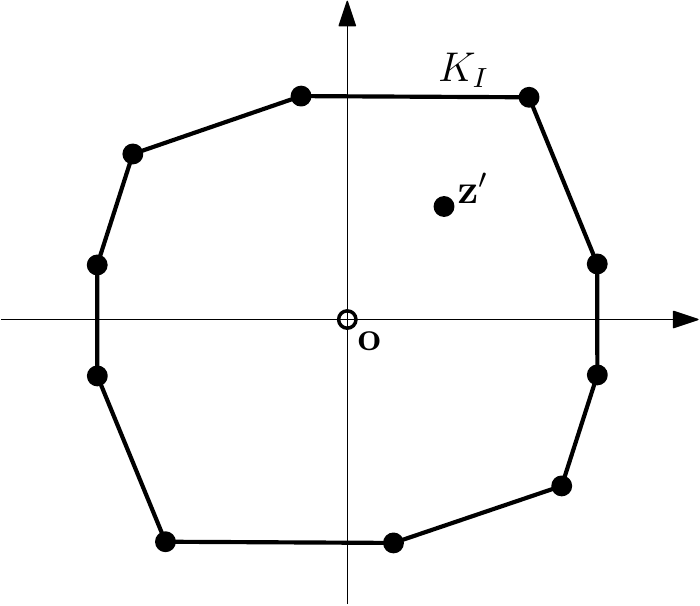}
\caption{}
\label{Figure-KNorm9}
\end{subfigure}
\caption{
{\small (a) Convex hull of $\Delta\textbf{X}$.}
{\small (b) Finding the sensitivity hull $K$.}
{\small (c) Transform $K$ to isotropic position $K_I$. Sample a point $\textbf{z}'$. }
}
\label{Figure-KNorms}
\end{figure*}
\vspace{2mm}\noindent{\bf Algorithm.}
As an overview,
PIM involves the following steps:\\
\indent (1) Compute sensitivity hull $K$ from $\Delta\textbf{X}$;\\
\indent (2) Transform $K$ to isotropic position $K_I$;\\
\indent (3) generating a noise in the space of $K_I$ by $K$-norm mechanism;\\
\indent (4) Transform to the original space.




%
We first describe how to compute sensitivity hull $K$.
Suppose we have a $\delta$-location set $\Delta\textbf{X}$ at a timestamp. We can first derive the convex hull of $\Delta\textbf{X}$, denoted by $K'=Conv(\Delta\textbf{X})$.
For example, in Figure \ref{Figure-KNorm}, the convex hull $K'$ is shown by the black lines given $\delta$-location set as ``$\bullet$'' and ``$\star$'' where ``$\star$'' is the true location.
Denote $\textbf{v}_1,\textbf{v}_2,\cdots,\textbf{v}_h$ the vertices of $K'$.
Then we use a set $\Delta\textbf{V}$ to store $\textbf{v}_i-\textbf{v}_j$ for any $\textbf{v}_i$ and $\textbf{v}_j$ from the vertices  of $K'$ as the equation below.
%
In Figure \ref{Figure-KNorm5}, for instance, the polygon ``${\small \boldsymbol{\bigtriangleup}}\cdots{\small \boldsymbol{\bigtriangleup}}$'' denotes $\textbf{v}_i-\textbf{v}_1$ for all $\textbf{v}_i$.
Then $Conv(\Delta\textbf{V})$ will be the sensitivity hull $K$ of the $\delta$-location set, as shown by the polygon with solid lines in Figure \ref{Figure-KNorm5}.
%
\begin{flalign*}
\hspace{2.4cm}
K&=Conv(\Delta\textbf{V})&
\\
\Delta\textbf{V}&=\mathop\cup\limits_{\textbf{v}_1,\textbf{v}_2\in \textrm{ vertices of }K'}(\textbf{v}_1-\textbf{v}_2)&
\end{flalign*}
%
%


Next we transform $K$ to its isotropic position $K_I$.
We sample $\textbf{y}_1,\textbf{y}_2,\cdots,\textbf{y}_l$ uniformly from $K$. Then a matrix $\textbf{T}$ can be derived by Equation (\ref{eqn-iso-T}). To verify if $\textbf{T}$ is stable, we can derive another $\textbf{T}'$. If the Frobenius norm $||\textbf{T}'-\textbf{T}||_F$ is small enough (e.g. $<10^{-3}$), then we accept $\textbf{T}$. Otherwise we repeat above process with larger $l$.
In the end, $K_I=\textbf{T}K$ is the isotropic position of $K$, as shown in
Figure \ref{Figure-KNorm9}.

Next a point $\textbf{z}'$ can be uniformly sampled from $K_I$. We generate a random variable $r$ from Gamma distribution $\Gamma(3,\epsilon^{-1})$. Let $\textbf{z}'=r\textbf{z}'$.
Then we transform the point $\textbf{z}'$ to the original space by $\textbf{z}'=\textbf{T}^{-1}\textbf{z}'$. The released  location is $\textbf{z}=\textbf{x}^*+\textbf{z}'$.
%
%




Algorithm \ref{alg-KNorm} summarizes the process of PIM. Lines \ref{line-sample-start}$\sim$\ref{line-sample-end} can be iterated until $\textbf{T}$ is stable, whereas
the computational complexity is not affected by the iterations because the number of samples is bounded by a constant (by Corollary \ref{theo-iso-sample}).
\begin{algorithm}[htb]
\caption{Planar Isotropic Mechanism}
\begin{algorithmic}[1]
\Require{
$\epsilon$, $\Delta\textbf{X}$, $\textbf{x}^*$
}
\State{$K'\leftarrow Conv(\Delta\textbf{X})$;}
\label{line-ConvX}
\Comment{{\tt \scriptsize convex hull of $\Delta\textbf{X}$}}
\State{$\Delta\textbf{V}\longleftarrow \mathop\cup\limits_{\textbf{v}_1,\textbf{v}_2\in\textrm{ vertices\ of\ } K'}(\textbf{v}_1-\textbf{v}_2)$};
\State{$K\leftarrow Conv(\Delta\textbf{V})$;}
\label{line-sensitivity-hull}
\Comment{{\tt \scriptsize sensitivity hull}}
\vspace{2mm}
\State{Repeat lines \ref{line-sample-start},\ref{line-sample-end} with larger $l$ if $\textbf{T}$ is not stable:}
\label{line-repeat}
\State{\hspace{\algorithmicindent}Sample $\textbf{y}_1,\textbf{y}_2,\cdots,\textbf{y}_l$ uniformly from $K$;}
\label{line-sample-start}
\State{\hspace{\algorithmicindent}$\textbf{T}\gets \left(
\frac{1}{l}
\sum_{i=1}^{l}\textbf{y}_i \textbf{y}_i^T
\right)^{-\frac{1}{2}}$};
\label{line-sample-end}
\State{$K_I=\textbf{T}K$};
\vspace{2mm}
\Comment{{\tt \scriptsize isotropic transformation}}
\State{Uniformly sample $\textbf{z}'$ from $K_I$};
\State{Sample $r\sim \Gamma(3,\epsilon^{-1})$;}
\\
\Return{$\textbf{z}=\textbf{x}^*+r\textbf{T}^{-1}\textbf{z}'$;}
\Comment{{\tt \scriptsize release $\textbf{z}$}}
\end{algorithmic}
\label{alg-KNorm}
\end{algorithm}


\noindent{\bf Privacy and Performance Analysis.}
We now present the privacy property, complexity, and the error  of  PIM.
\begin{theorem}
\label{theo-priv-KNorm}
Algorithm \ref{alg-KNorm} is $\epsilon$-differentially private on $\delta$-location set $\Delta\textbf{X}$.
\end{theorem}
\begin{proof}
The isotropic transformation is a unique $\mathbb{R}^2\rightarrow \mathbb{R}^2$ mapping. In the space of $K_I$, it is easy to  prove that the probability distribution of $\textbf{z}'$ is equal to
Equation (\ref{eqn-pdf-K-Norm})
in which $K_I$ is in place of $K$. Therefore, Algorithm \ref{alg-KNorm} is $\epsilon$-differentially private.
\end{proof}

\begin{theorem}
\label{theo-time-KNorm}
Algorithm \ref{alg-KNorm} takes $O(nlog(h)+h^2log(h))$ time where $n$ is the size of $\Delta\textbf{X}$ and $h$ is number of vertices on $Conv(\Delta\textbf{X})$.
\end{theorem}
\begin{proof}
Line \ref{line-ConvX} takes time $O(nlog(h))$ where $n$ is the size of $\Delta\textbf{X}$ and $h$ is number of vertices on $Conv(\Delta\textbf{X})$. Line \ref{line-sensitivity-hull} takes time $O(h^2log(h))$. Because the number of samples is bounded by a constant, lines \ref{line-repeat}$\sim$\ref{line-sample-end} need $O(1)$ time. Thus the overall complexity is $O(nlog(h)+h^2log(h))$.
\end{proof}

\begin{theorem}
\label{theo-alg-error}
Algorithm \ref{alg-KNorm} has error $O\left(\frac{1}{\epsilon}\sqrt{\textsc{Area}(K)}\right)$ at most, which means
it achieves the lower bound in Theorem \ref{theorem-lowerbound}.
\end{theorem}
\begin{proof}
From
Equation (\ref{eqn-pdf-K-Norm}),
we can obtain the following error in the isotropic space
\begin{align*}
\textsc{Error}^2=\int_{\textbf{x}} Pr(\textbf{x})||\textbf{x}||_2^2d\textbf{x}
=\frac{12det(\textbf{T}^{-1})}{\epsilon^2 \textsc{Area}(K_I)}\mathop\mathbb{E}\limits_{\textbf{x}\in K_I}||\textbf{x}||_2^2
\end{align*}
where $det(\textbf{T}^{-1})$ is the determinant of $\textbf{T}^{-1}$. By the isotropic property,
\begin{align*}
\textsc{Error}^2=\frac{12 L_{K_I}^2}{\epsilon^2}det(\textbf{T}^{-1})\textsc{Area}(K_I)
\end{align*}
where $L_{K_I}$ is the isotropic constant of $K_I$. After transforming back to the original space, it is easy to show that $det(\textbf{T}^{-1})\textsc{Area}(K_I)$ becomes $\textsc{Area}(K)$. Therefore, Algorithm \ref{alg-KNorm} has error at most $O(\frac{1}{\epsilon}\sqrt{\textsc{Area}(K)})$, which is the lower bound in  Theorem \ref{theorem-lowerbound}. Hence Algorithm \ref{alg-KNorm} is optimal.
\end{proof}

\subsection{Location Inference}
\label{sec-loc-infr}
The inference of line \ref{line-bayesian} in Algorithm \ref{alg-framework} is a general statement because inference methods depend on specific release algorithms. To implement the inference for PIM, we need to transform the location $\textbf{s}_i$ and the released location $\textbf{z}_t$ to the isotropic space of $K_I$. Then in Equation (\ref{eqn-posterior}), the probability $Pr(\textbf{z}_t| \textbf{u}_t^*=\textbf{s}_i)$ can be computed as follows. This completes the whole algorithm.
\begin{flalign*}
\hspace{0.7cm}
Pr(\textbf{z}_t| \textbf{u}_t^*=\textbf{s}_i)
&= \frac{\epsilon^2}{2\textsc{Area}(K_I)}exp(-\epsilon||\textbf{z}_t'-\textbf{s}_i'||_{K_I})&\\
&\textbf{z}_t'=\textbf{T}\textbf{z};\ \textbf{s}_i'=\textbf{T}\textbf{s}_i&
\end{flalign*}

\section{Experimental Evaluation}
\label{sec-experiment}
In this section we present experimental evaluation of our method.
All algorithms were implemented in Matlab on a PC with 2.9 GHz Intel i7 CPU and 8 GB Memory.

\vspace{2mm}\noindent{\bf Datasets.}
We used two real-world datasets.
\begin{enumerate}[I]
\item
Geolife data. Geolife data \cite{GeoLife-Zheng-10} was collected from $182$ users in a period of over three years.
It recorded a wide range of users' outdoor movements, represented by a series of tuples containing latitude, longitude and timestamp.
The trajectories were updated in a high frequency, e.g. every $1\sim 60$ seconds.
We extracted all the trajectories within the $3$rd ring of Beijing to train the Markov model, with the map partitioned into cells of  $0.34\times 0.34\ {km}^2$.
\item
Gowalla data. Gowalla data \cite{KDD-Gowalla-2011} contains $6,442,890$ check-in locations of $196,586$ users over the period of Feb. 2009 to Oct. 2010. We extracted all the check-ins in Los Angeles to train the Markov model, with the map partitioned into cells of $0.89\times 0.89\ {km}^2$. Because check-ins were logged in a relatively low frequency, e.g. every $1\sim 50$ minutes,
we can examine the difference of the results from Gowalla and Geolife.
\end{enumerate}

\vspace{2mm}
\noindent{\bf Metrics.} We used the following metrics in our experiment, including two internal metrics: size of $\Delta\textbf{X}$, drift ratio, and two sets of utility metrics:  distance, precision and recall.
We skip the runtime report because most locations were released within $0.3$ second by PIM.
\begin{enumerate}[I]
\item
Since our privacy definition is based on $\delta$-location set $\Delta\textbf{X}$, we evaluated the size of $\Delta\textbf{X}$ to understand how $\Delta\textbf{X}$ grows or changes.
\item
The definition of $\Delta\textbf{X}$ and the potential limit of Markov model may cause the true location to fall outside $\Delta\textbf{X}$ (drift). Thus we measured the drift ratio computed as the number of timestamps the true location is excluded in $\Delta\textbf{X}$ over total number of timestamps.
\item
We measured the distance between the released location and the true location, which can be considered as a general utility metric independent of specific location based applications.
%
\item
We also run $k$ nearest neighbor ($k$NN) queries using the released locations and report its precision and recall compared to the true $k$NN set using the original location.  Suppose the true $k$NN set is $R$, the returned $k'$NN set (we set $k'\geq k$) is $R'$, precision is defined as $|R \cap R'| / k'$, and recall is defined as $|R \cap R'| /k$.
\end{enumerate}

\begin{figure}[t]
\centering
\begin{subfigure}{0.233\textwidth}
\centering
\includegraphics[width=4.2cm]{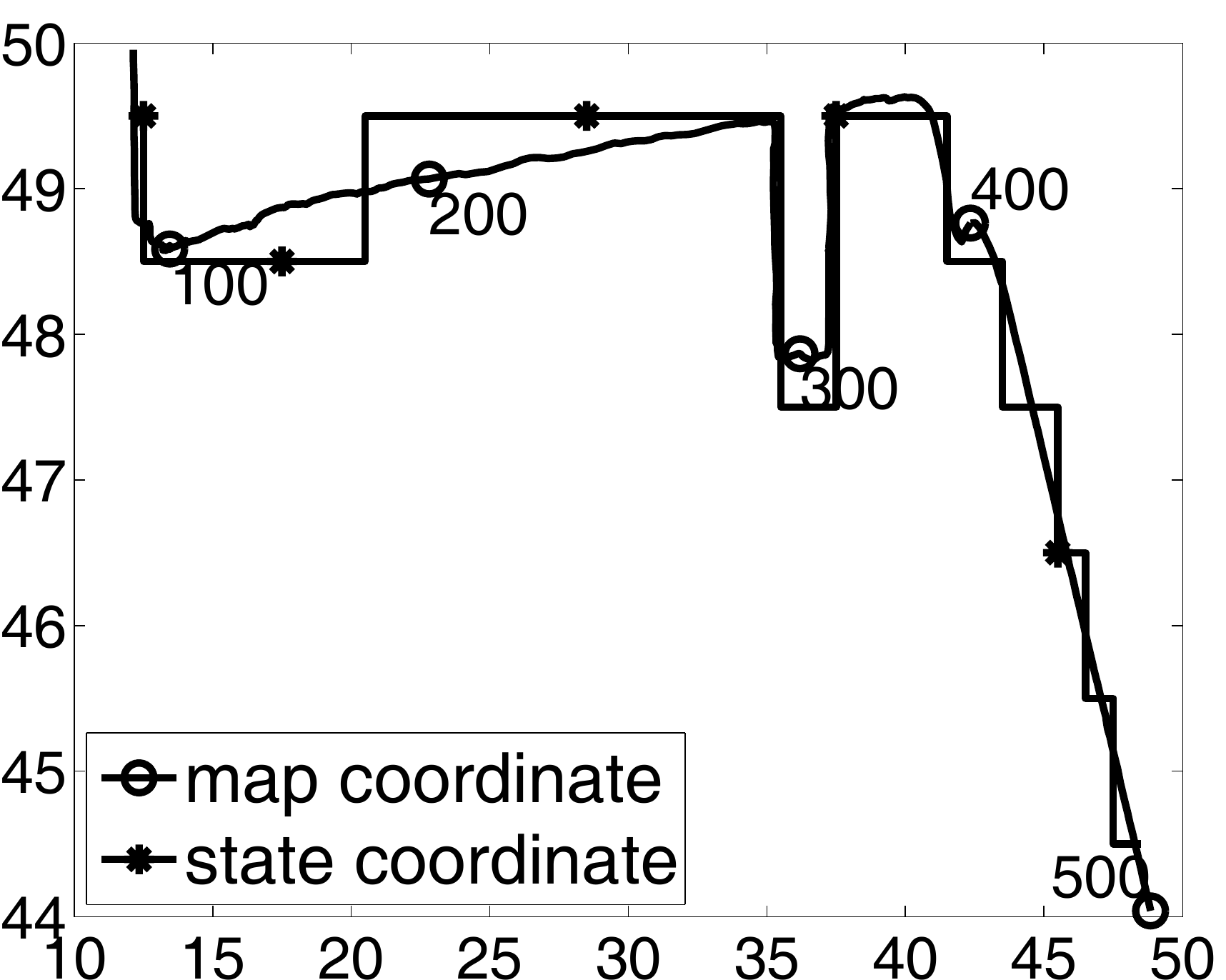}
\caption{{\small True trace }}
\label{Figure-example-trace0}
\end{subfigure}
\begin{subfigure}{0.233\textwidth}
\centering
\includegraphics[width=4.2cm]{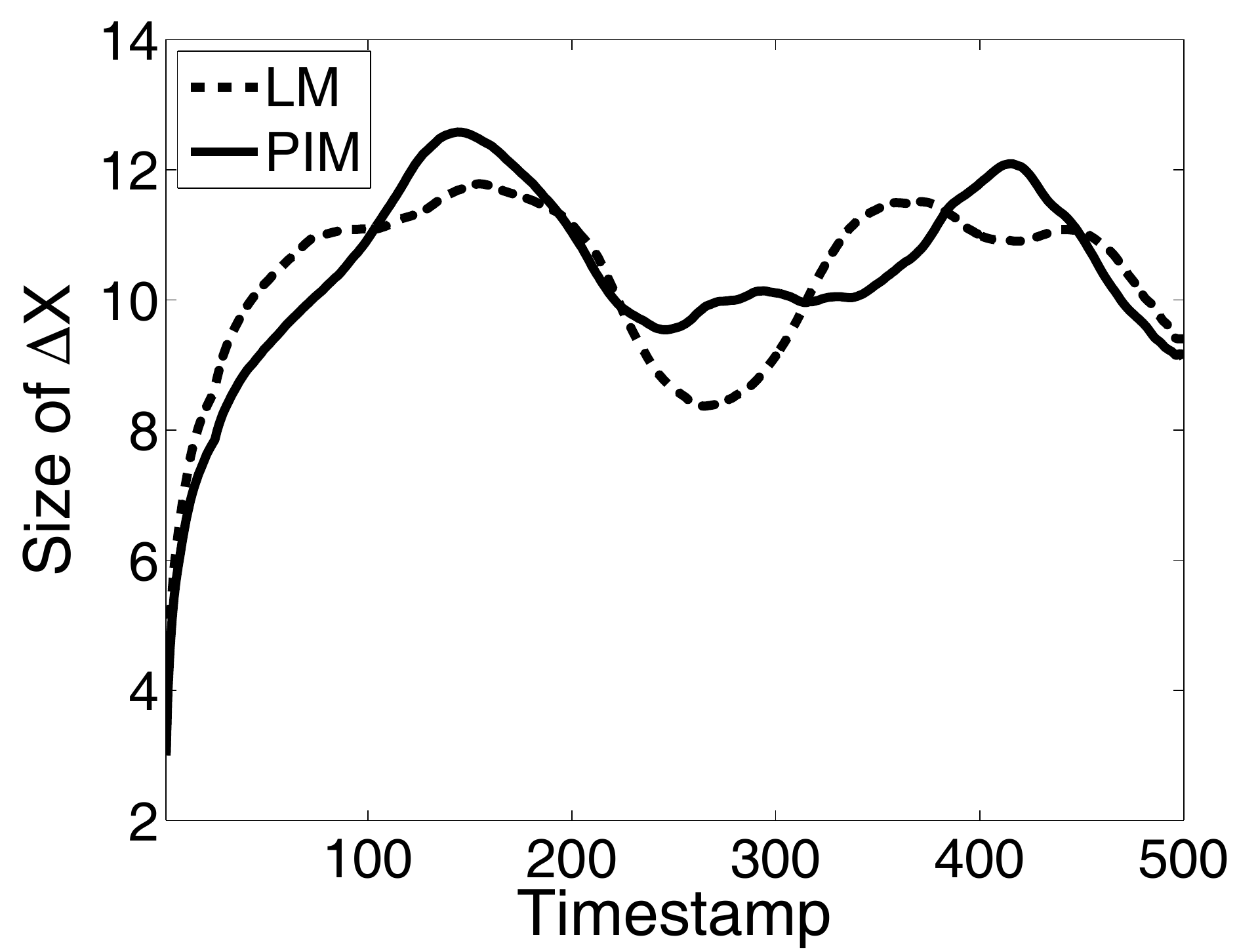}
\caption{{\small Size of $\Delta\textbf{X}$}}
\label{Figure-aTraj-size-001}
\end{subfigure}
\\
\begin{subfigure}{0.233\textwidth}
\centering
\includegraphics[width=4.2cm]{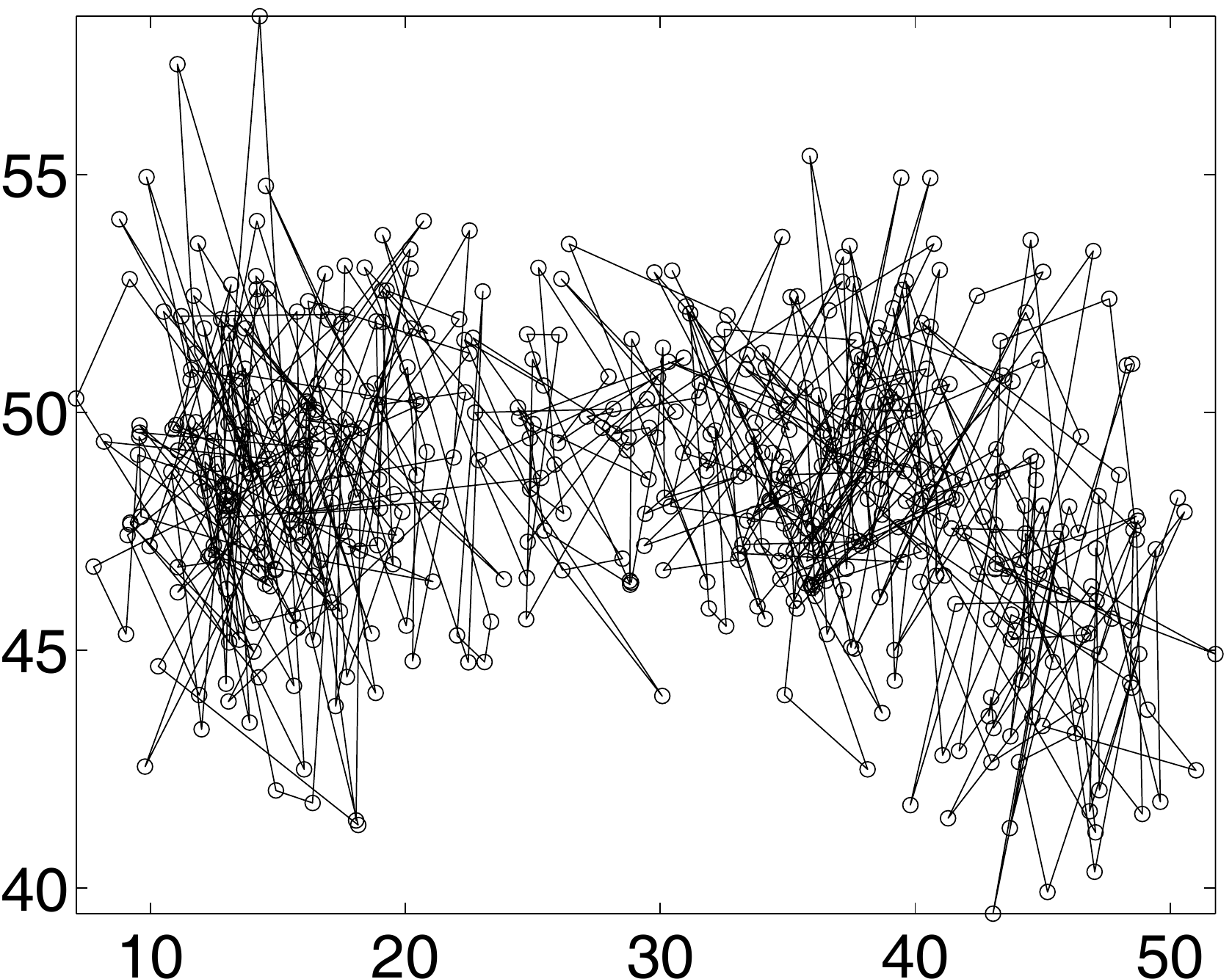}
\caption{{\small LM released trace}}
\label{Figure-aTraj-traj-LM-001}
\end{subfigure}
\begin{subfigure}{0.233\textwidth}
\centering
\includegraphics[width=4.2cm]{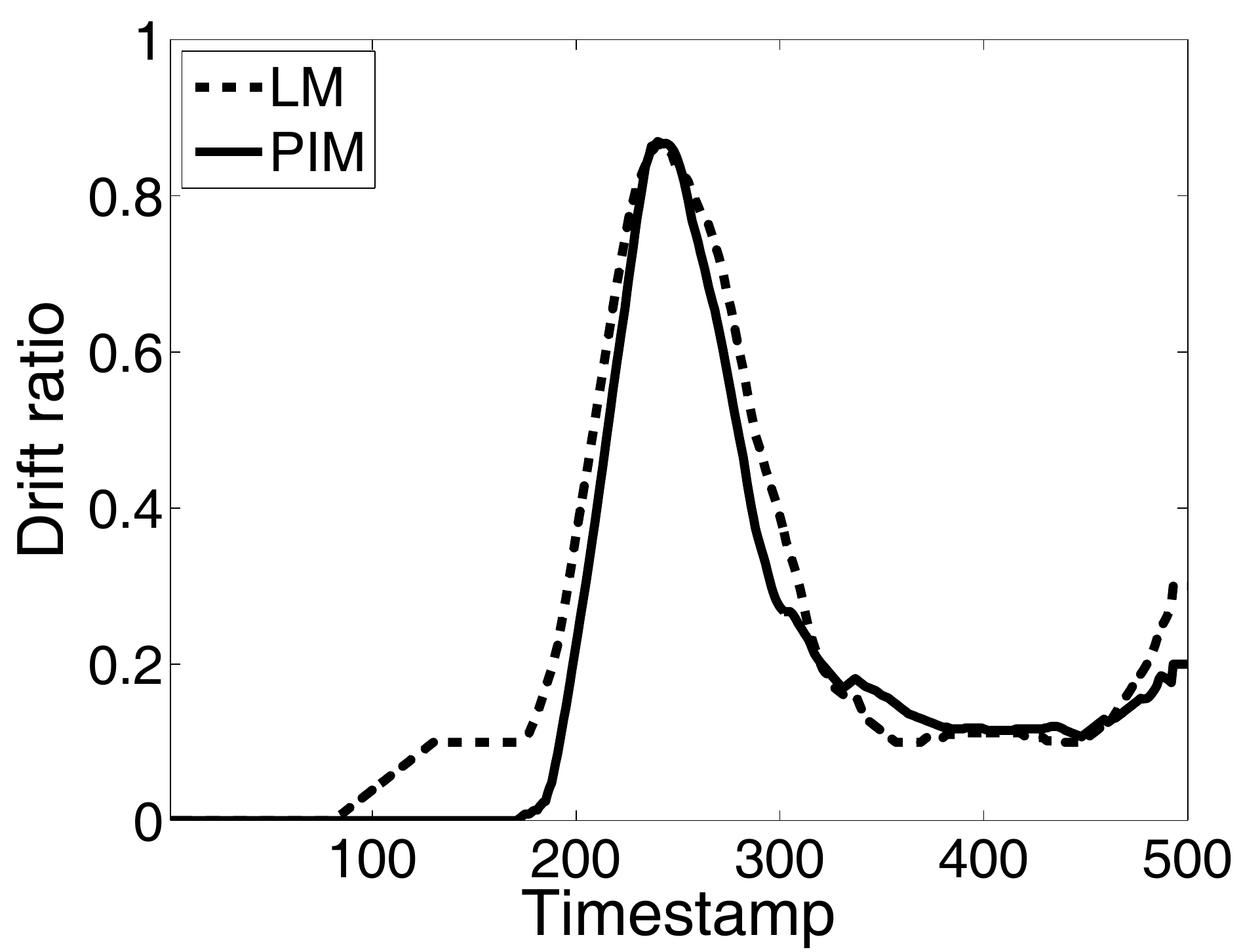}
\caption{{\small Drift ratio}}
\label{Figure-aTraj-drift-001}
\end{subfigure}
\\
\begin{subfigure}{0.233\textwidth}
\centering
\includegraphics[width=4.2cm]{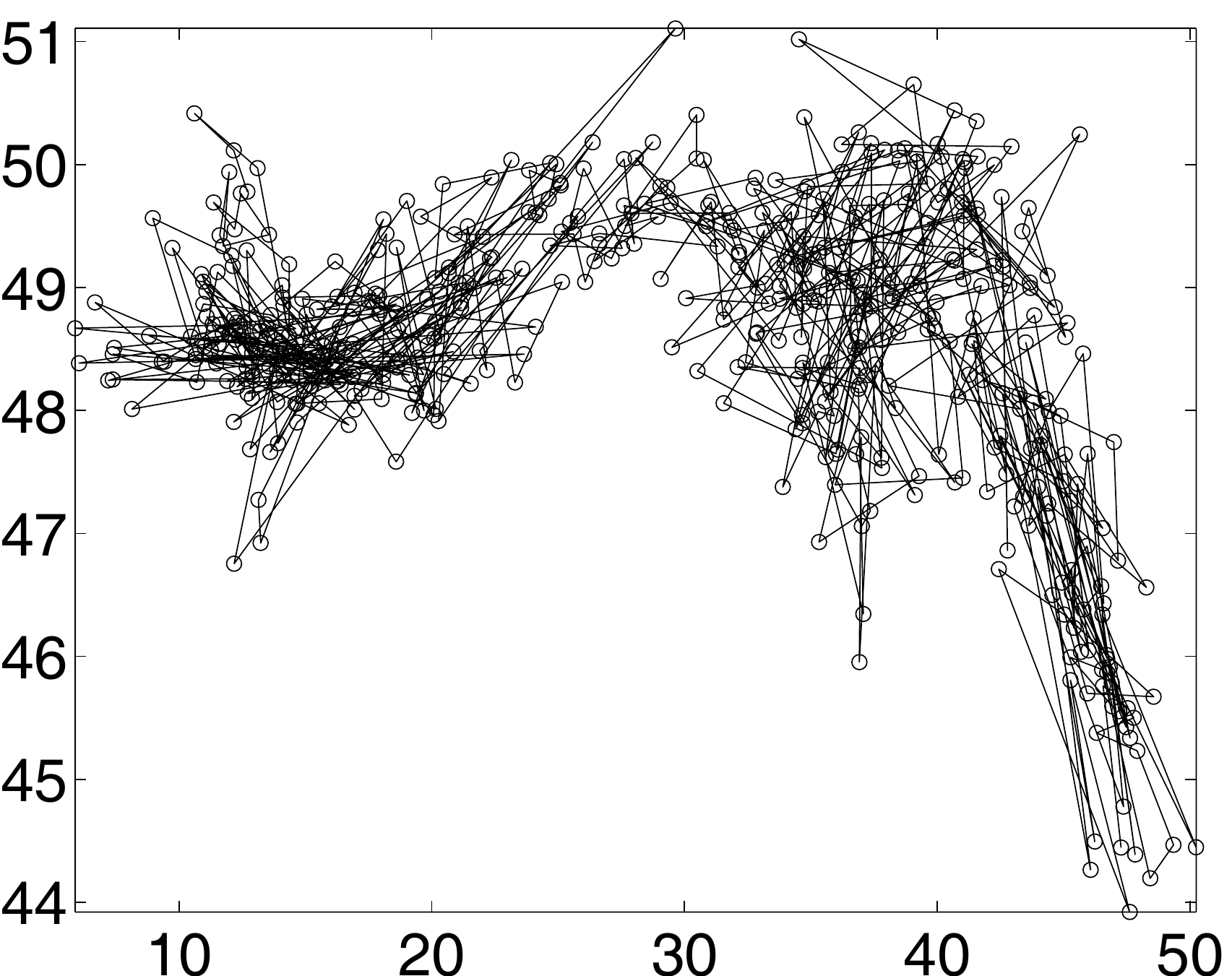}
\caption{{\small PIM released trace}}
\label{Figure-aTraj-traj-IM-001}
\end{subfigure}
\begin{subfigure}{0.233\textwidth}
\centering
\includegraphics[width=4.2cm]{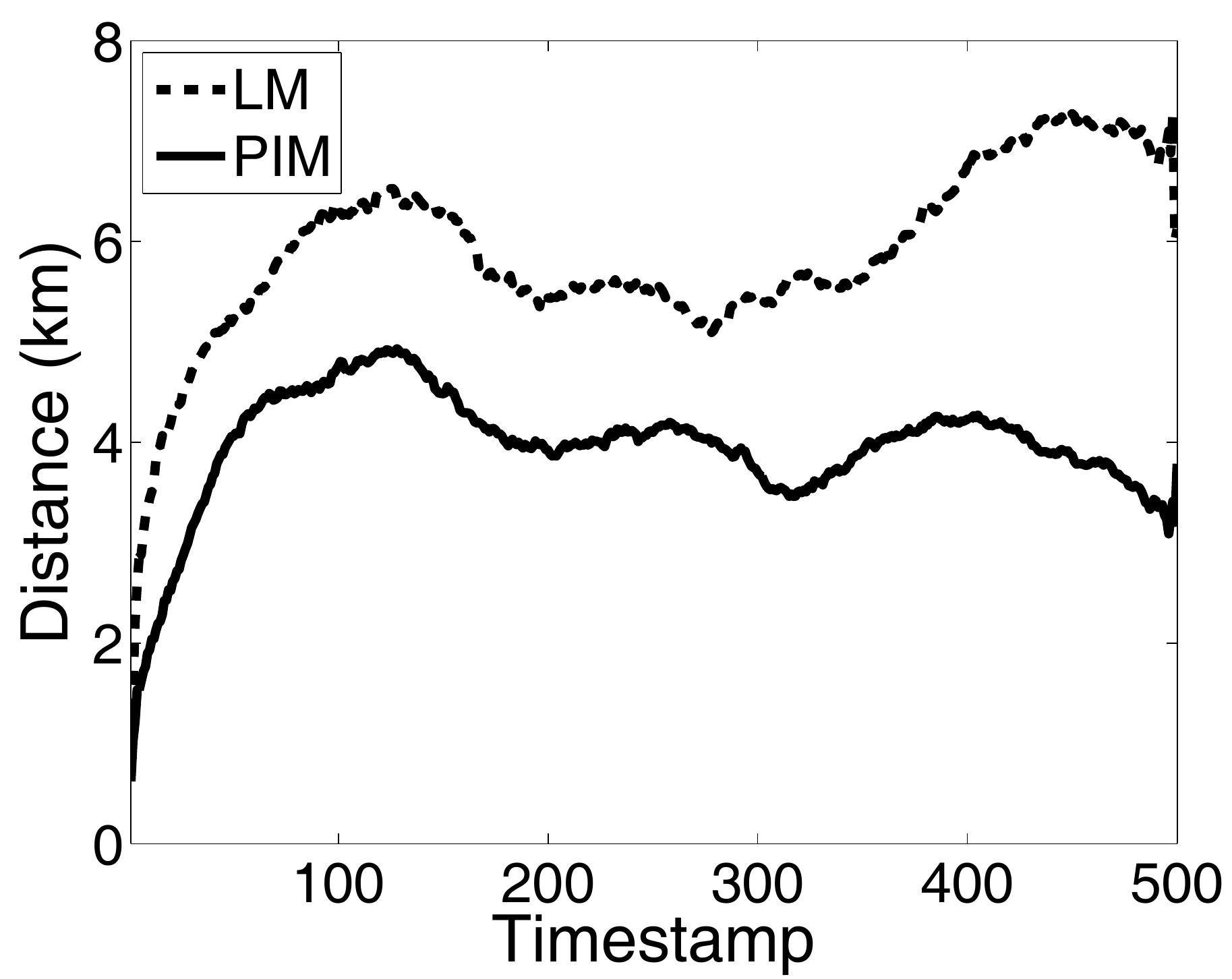}
\caption{{\small Distance}}
\label{Figure-aTraj-dist-001}
\end{subfigure}
\caption{{\small Performance over time:}
{\small (a) The true (original) trace;}
{\small (c)(e) Released traces;}
{\small (b) Size of $\Delta\textbf{X}$ over time;}
{\small (d) Drift ratio over time;}
{\small (f) Distance over time.}
}
\label{Figure-overtime}
\end{figure}

\begin{figure}[!ht]
\begin{subfigure}{0.233\textwidth}
\centering
\includegraphics[width=4.2cm]{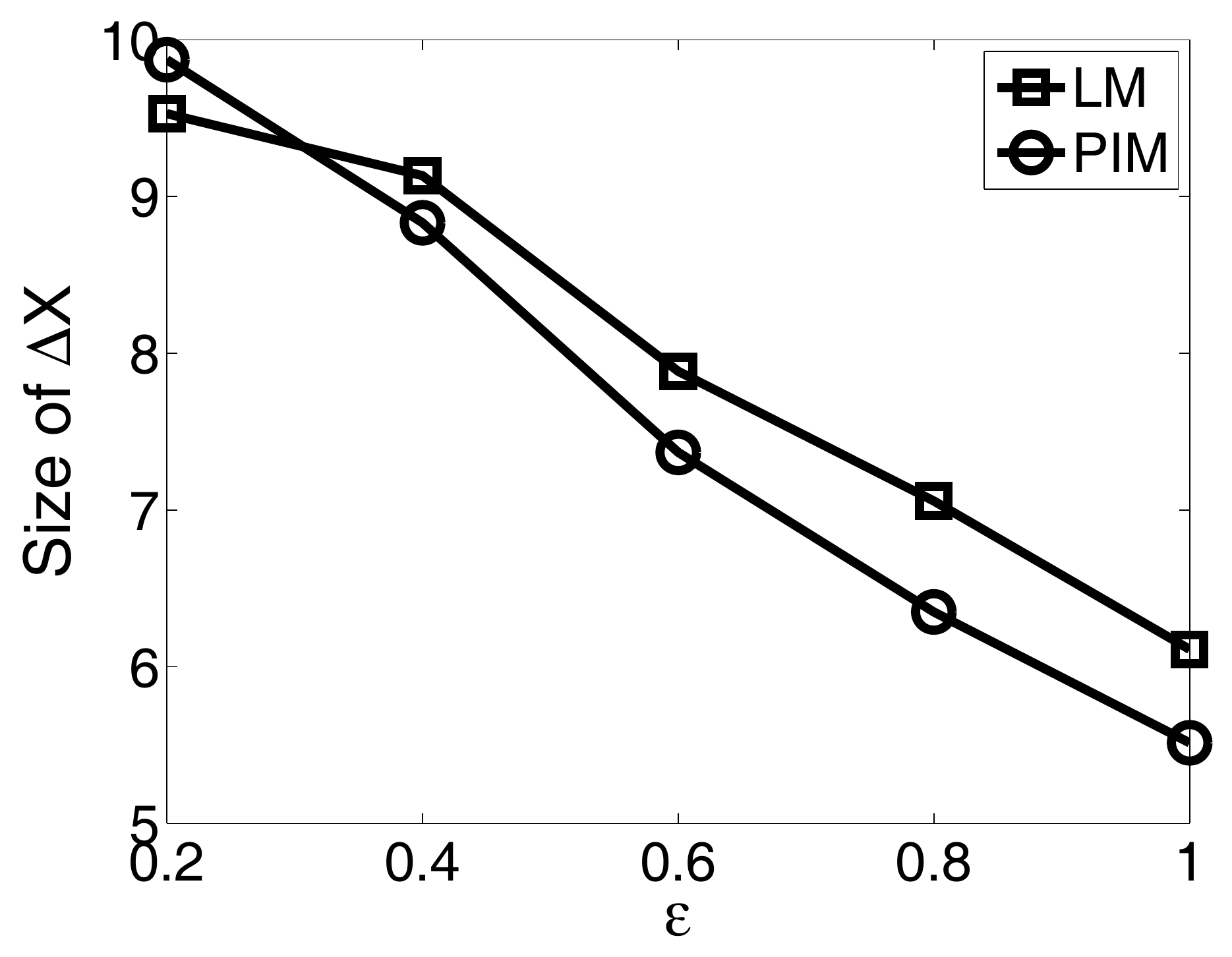}
\caption{{\small Size vs. $\epsilon$}}
\label{Figure-epsilon-size-4}
\end{subfigure}
\begin{subfigure}{0.233\textwidth}
\centering
\includegraphics[width=4.2cm]{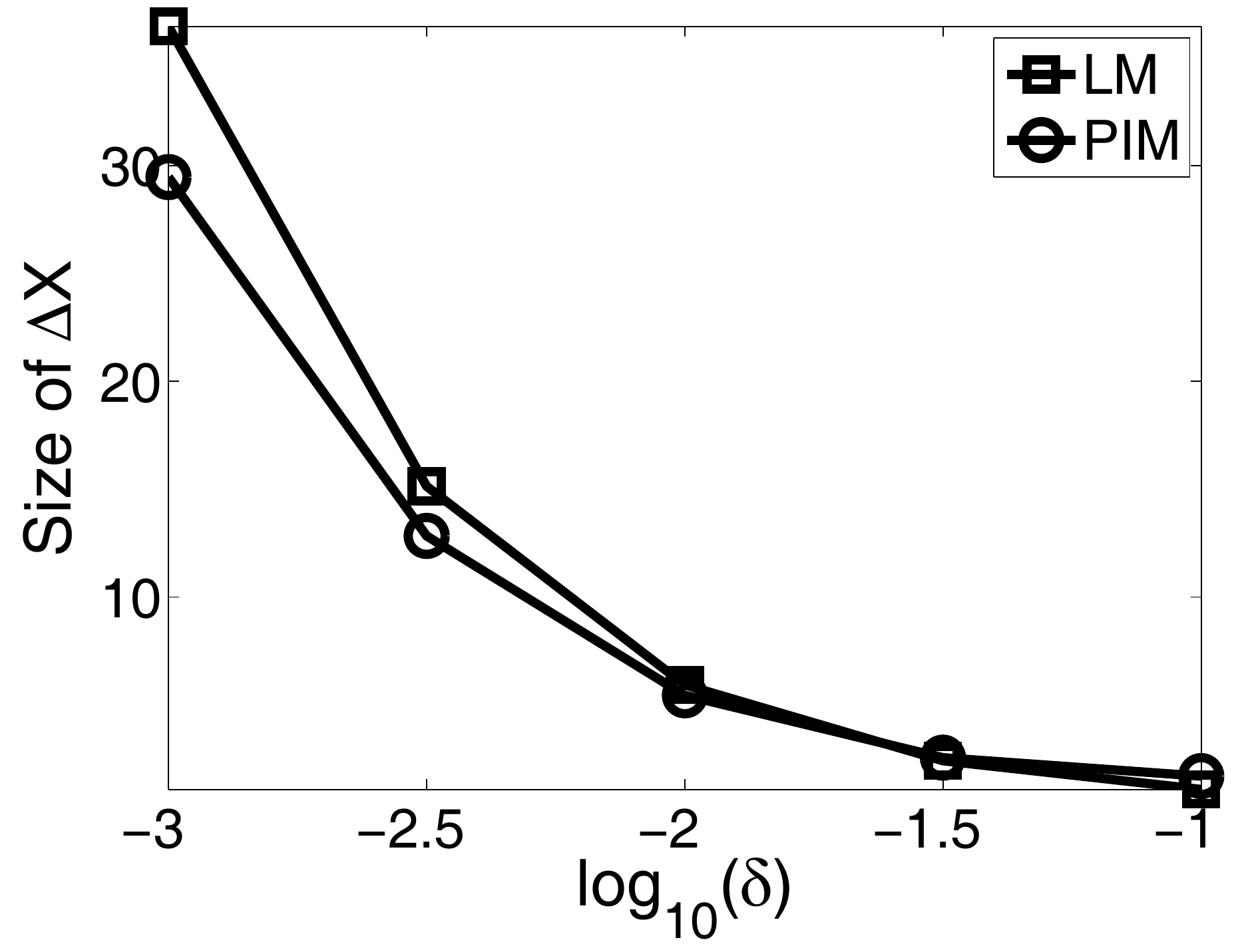}
\caption{{\small Size vs. $\delta$ }}
\label{Figure-delta-size-4}
\end{subfigure}
\begin{subfigure}{0.233\textwidth}
\centering
\includegraphics[width=4.2cm]{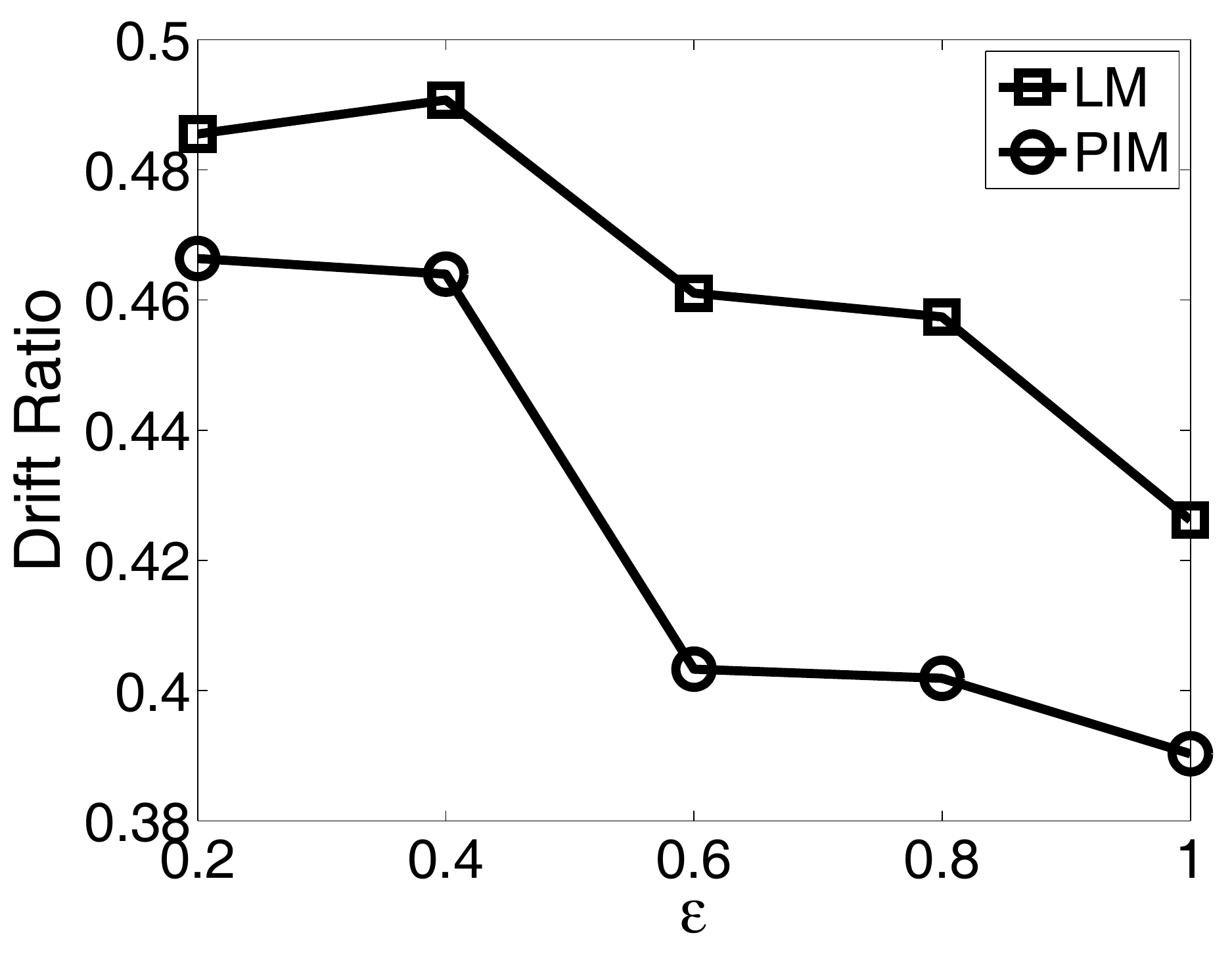}
\caption{{\small Drift Ratio vs. $\epsilon$ }}
\label{Figure-epsilon-drift-4}
\end{subfigure}
\begin{subfigure}{0.233\textwidth}
\centering
\includegraphics[width=4.2cm]{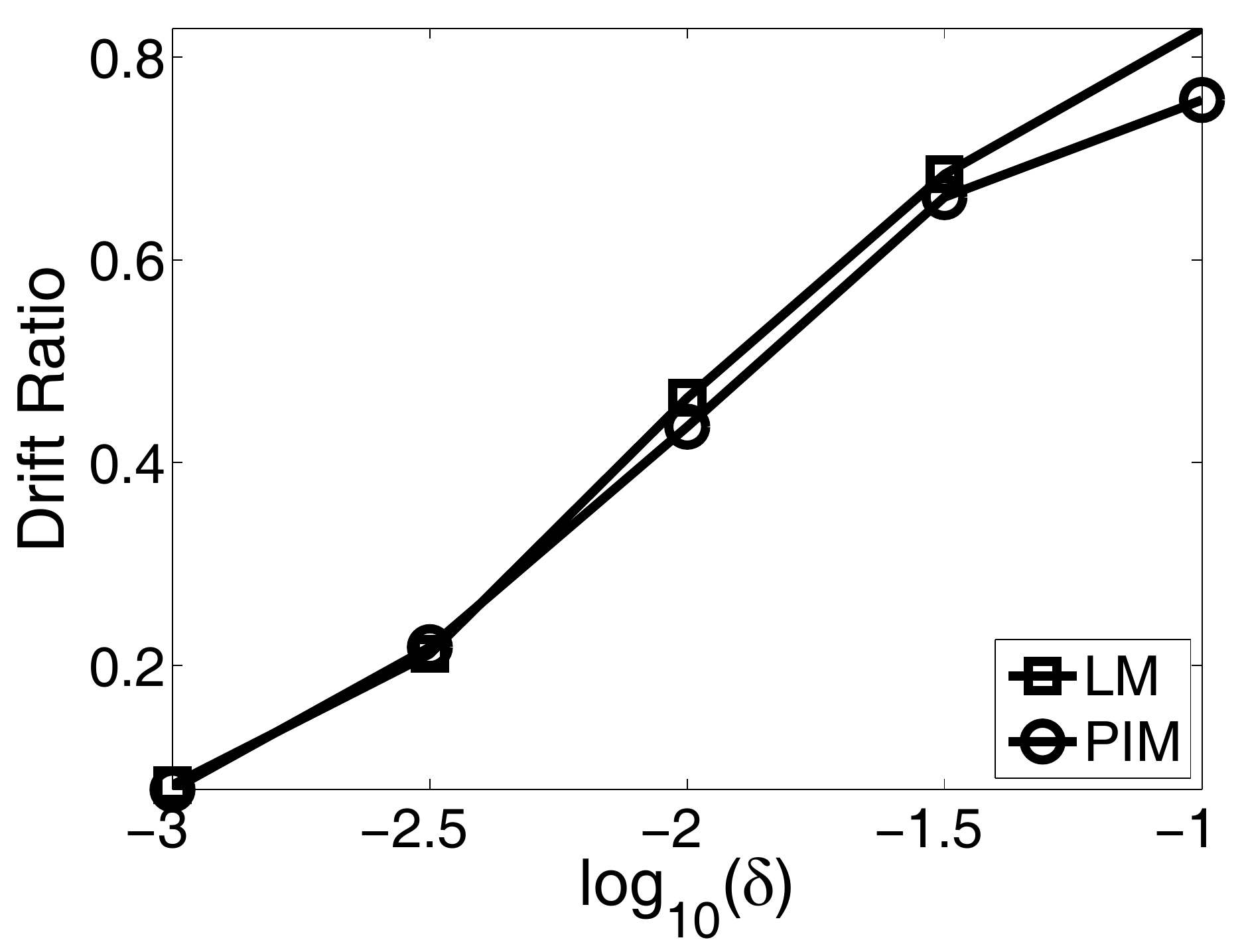}
\caption{{\small Drift Ratio vs. $\delta$ }}
\label{Figure-delta-drift-4}
\end{subfigure}

\begin{subfigure}{0.233\textwidth}
\centering
\includegraphics[width=4.2cm]{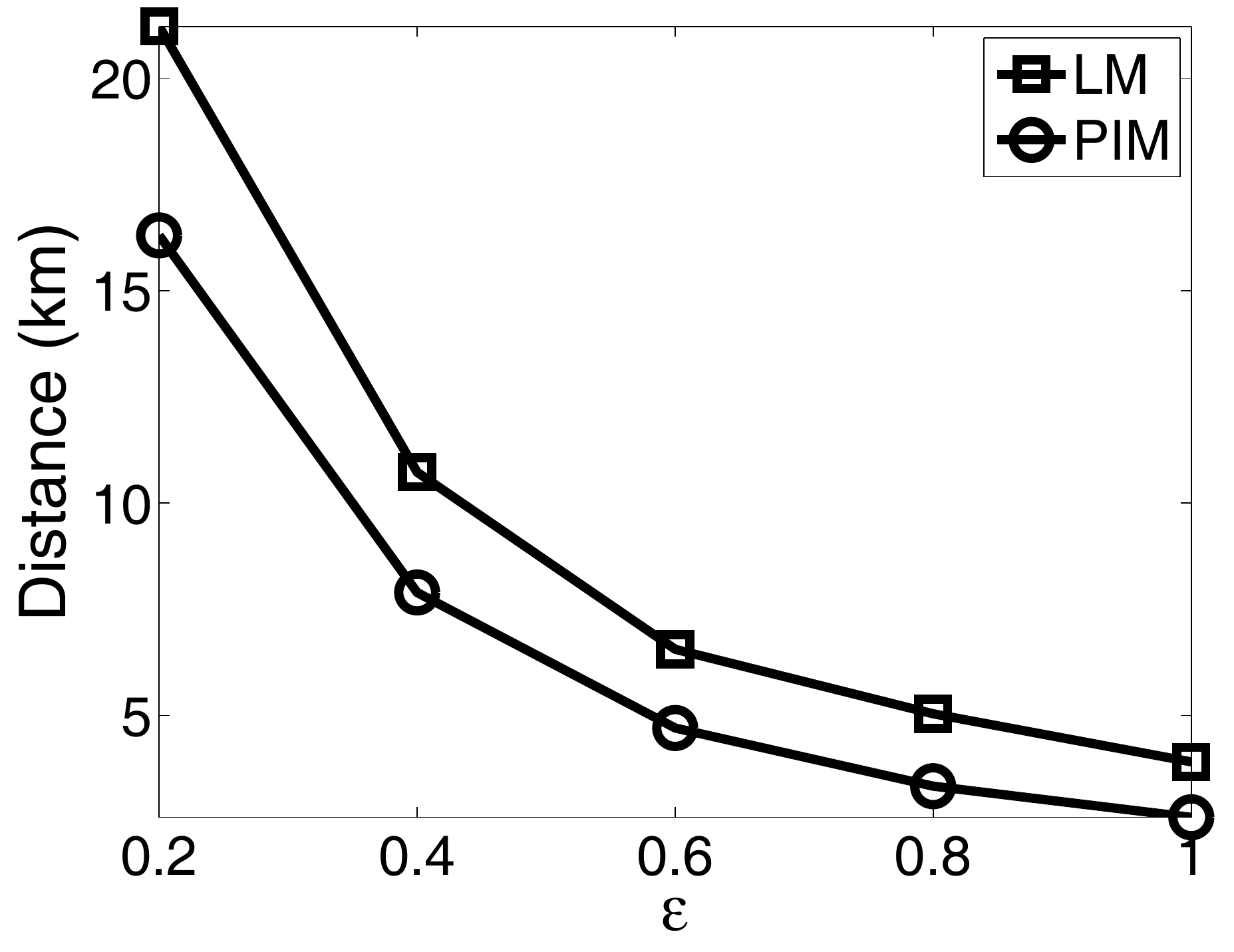}
\caption{{\small Distance vs. $\epsilon$ }}
\label{Figure-epsilon-Var-4}
\end{subfigure}
\begin{subfigure}{0.233\textwidth}
\centering
\includegraphics[width=4.2cm]{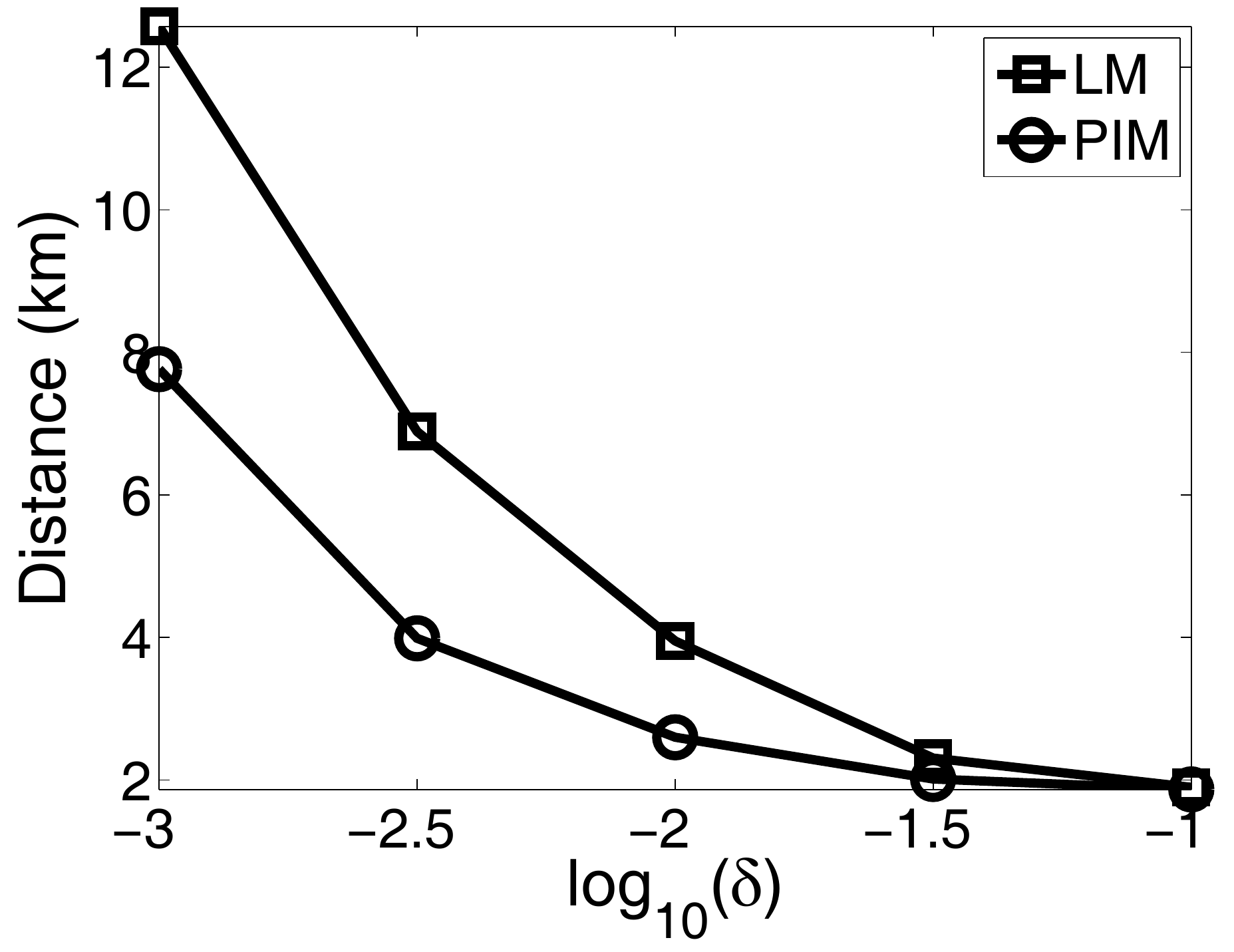}
\caption{{\small Distance vs. $\delta$ }}
\label{Figure-delta-Var-4}
\end{subfigure}
\caption{{\small Impact of parameters on GeoLife data with popular $\textbf{M}$:  }
{\small (a)(b) Impact of $\epsilon$ and $\delta$ on size of $\Delta\textbf{X}$;}
{\small (c)(d) Impact of $\epsilon$ and $\delta$ on drift ratio;}
{\small (e)(f) Impact of $\epsilon$ and $\delta$ on distance.}
}
\label{Figure-impact-GeoLife-popular}
\end{figure}

\begin{figure}[!ht]
\begin{subfigure}{0.233\textwidth}
\centering
\includegraphics[width=4.2cm]{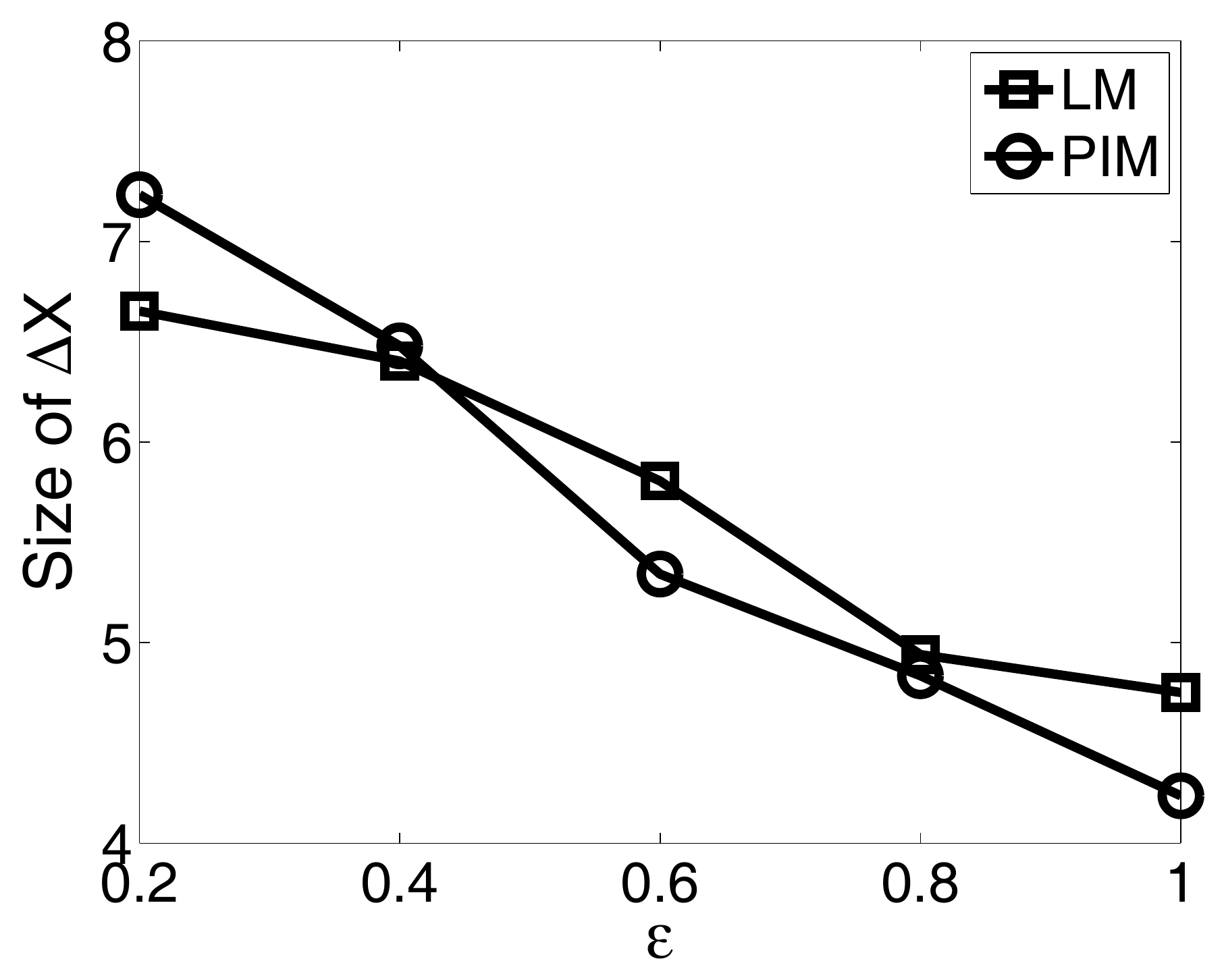}
\caption{{\small Size vs. $\epsilon$ }}
\label{Figure-epsilon-size-41}
\end{subfigure}
\begin{subfigure}{0.233\textwidth}
\centering
\includegraphics[width=4.2cm]{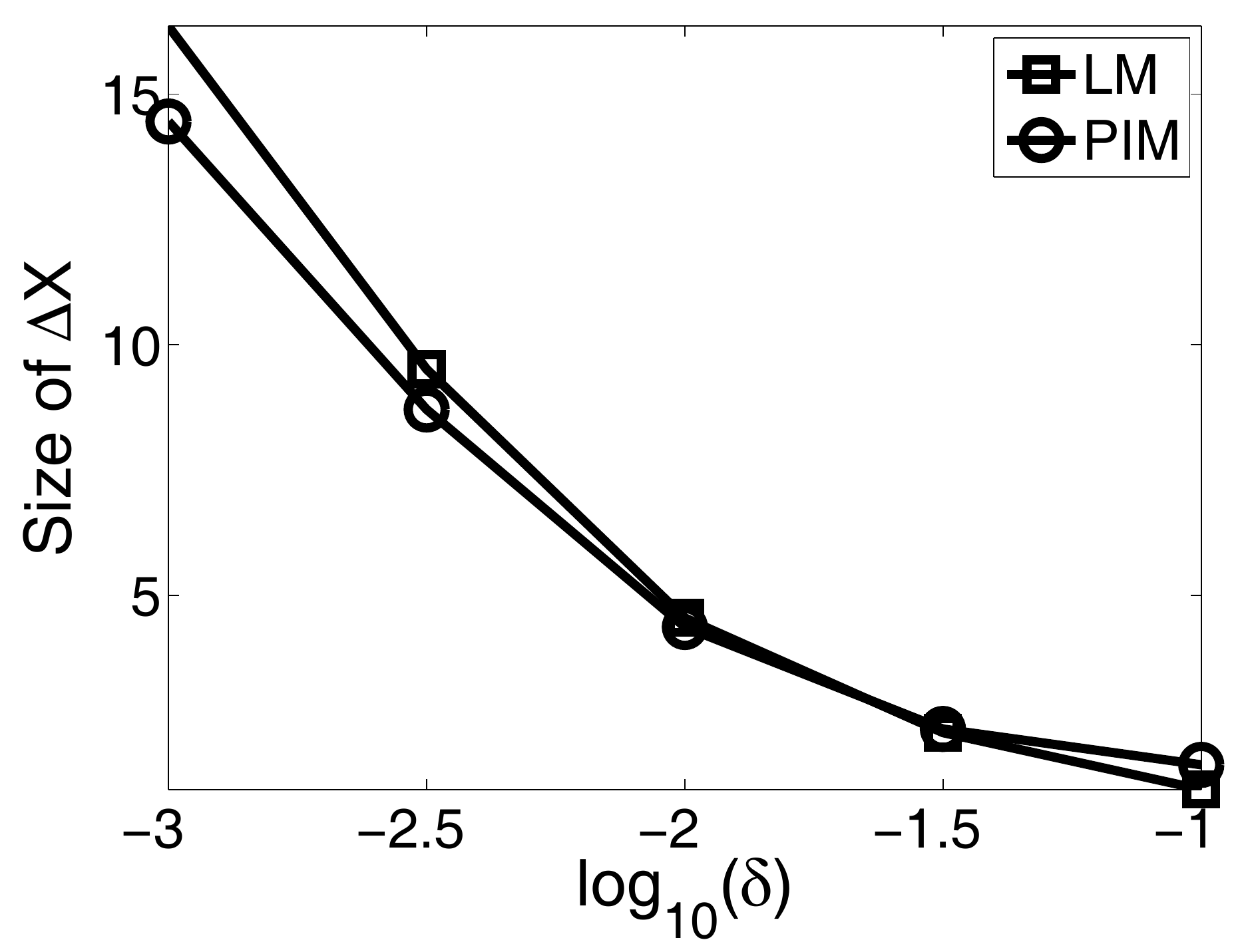}
\caption{{\small Size vs. $\delta$ }}
\label{Figure-delta-size-41}
\end{subfigure}

\begin{subfigure}{0.233\textwidth}
\centering
\includegraphics[width=4.2cm]{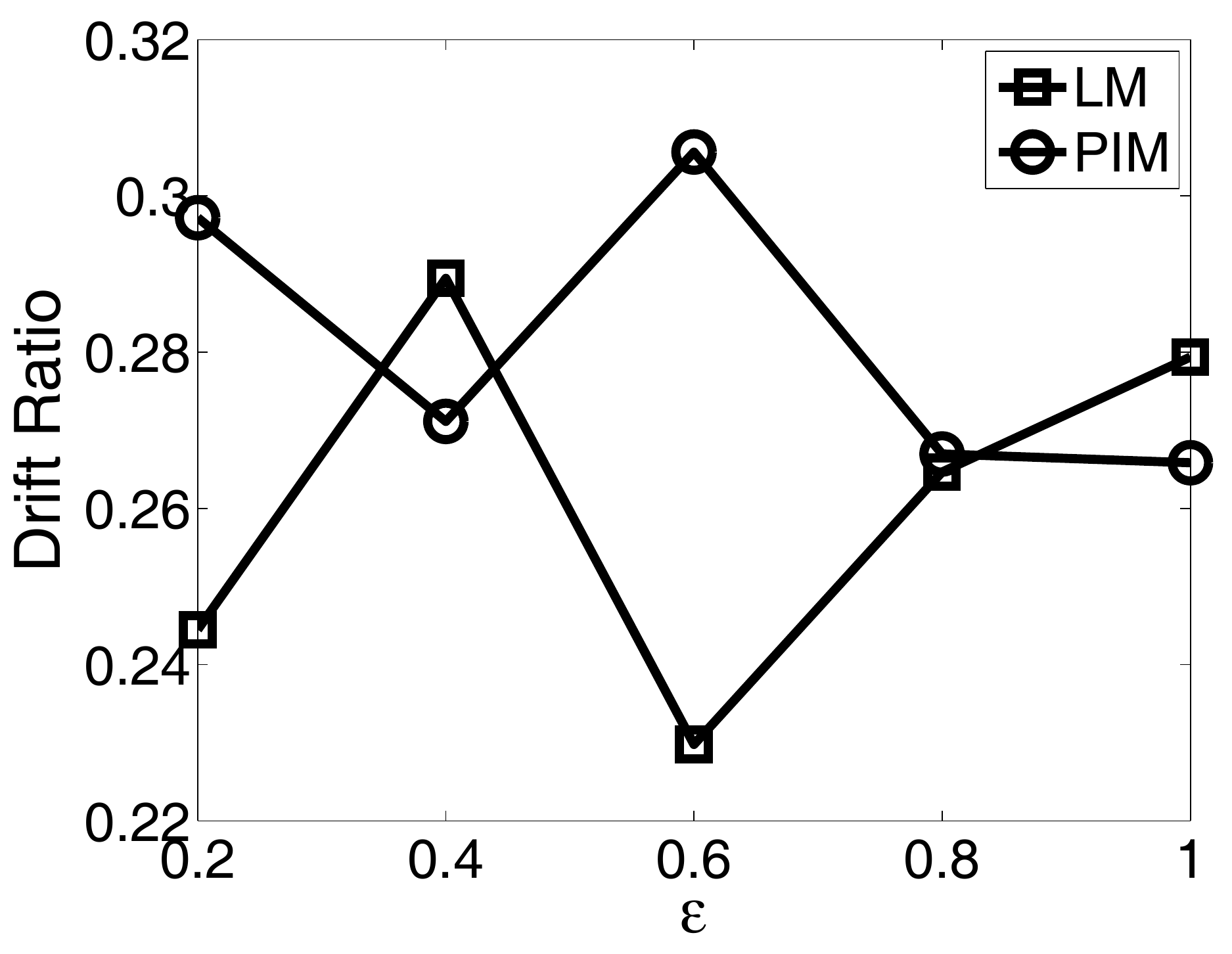}
\caption{{\small Drift Ratio vs. $\epsilon$ }}
\label{Figure-epsilon-drift-41}
\end{subfigure}
\begin{subfigure}{0.233\textwidth}
\centering
\includegraphics[width=4.2cm]{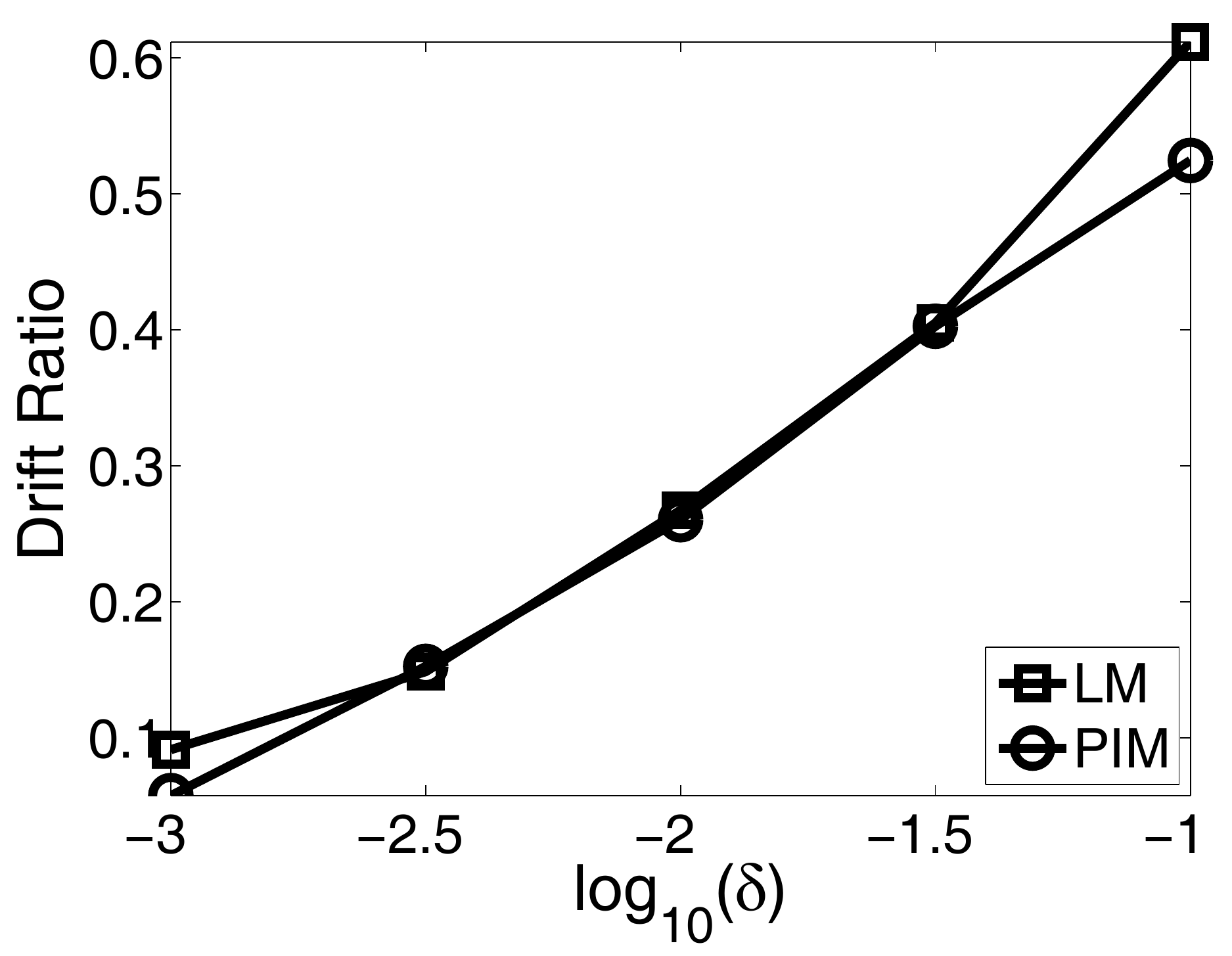}
\caption{{\small Drift Ratio vs. $\delta$ }}
\label{Figure-delta-drift-41}
\end{subfigure}

\begin{subfigure}{0.233\textwidth}
\centering
\includegraphics[width=4.2cm]{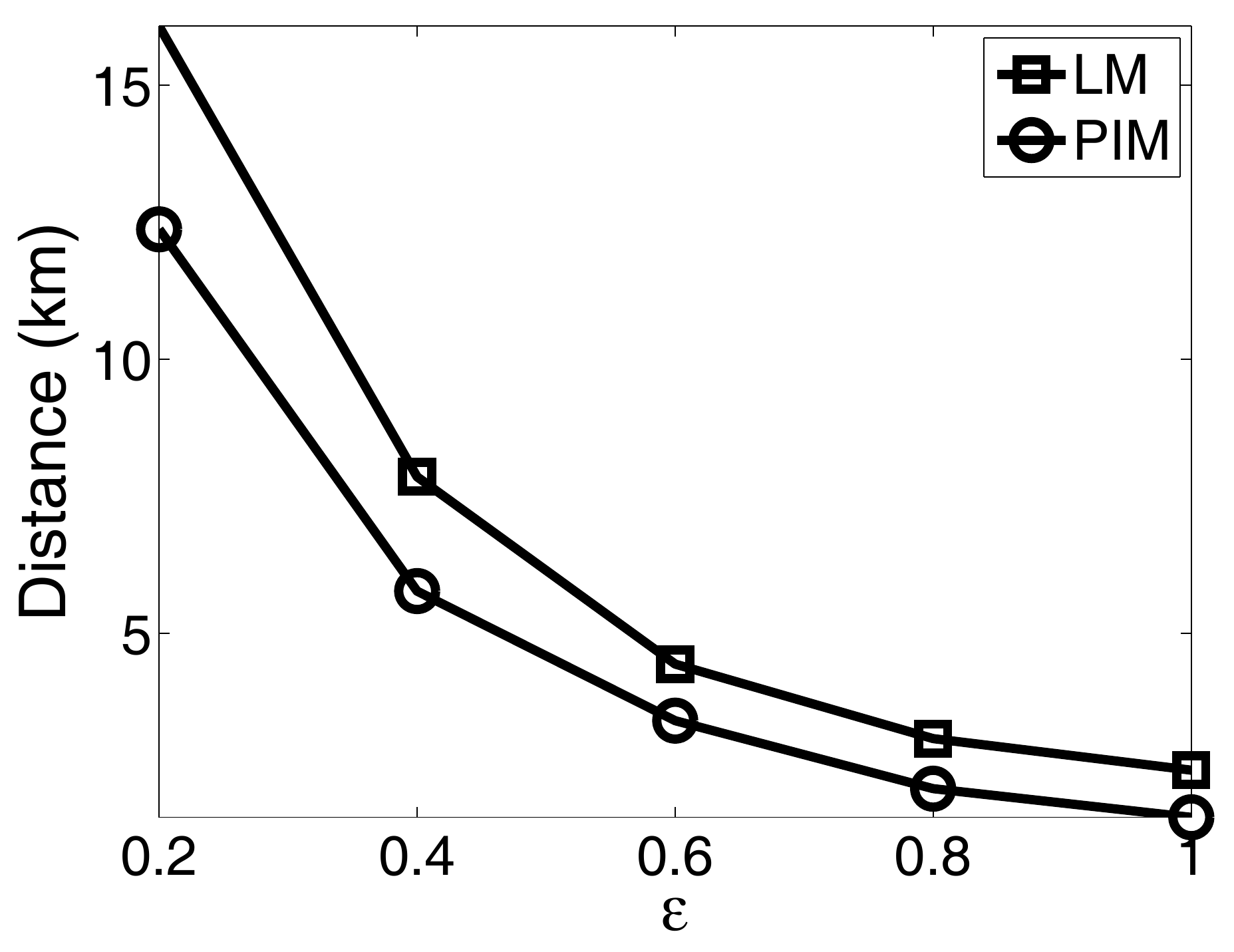}
\caption{{\small Distance vs. $\epsilon$ }}
\label{Figure-epsilon-Var-41}
\end{subfigure}
\begin{subfigure}{0.233\textwidth}
\centering
\includegraphics[width=4.2cm]{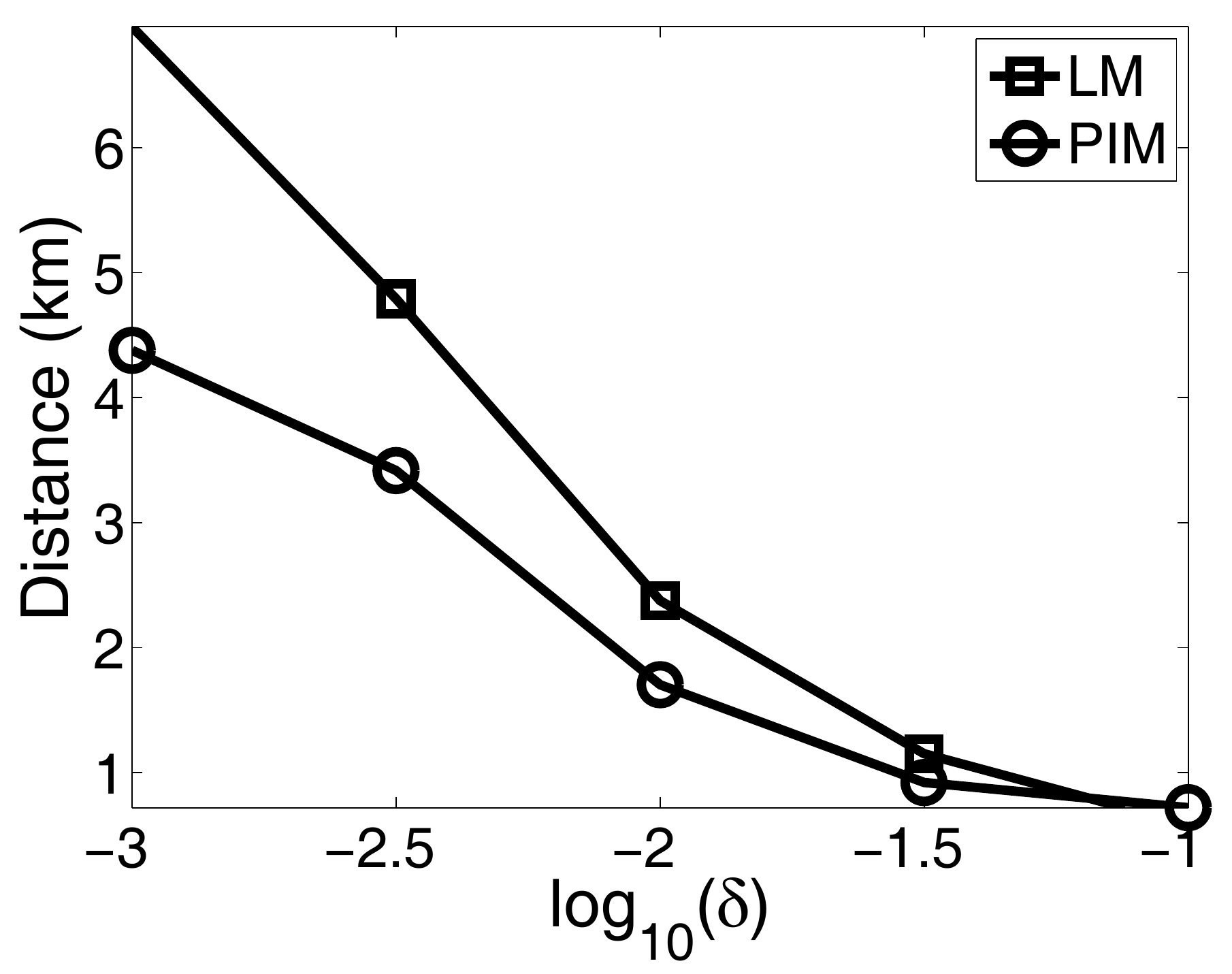}
\caption{{\small Distance vs. $\delta$ }}
\label{Figure-delta-Var-41}
\end{subfigure}
\caption{{\small Impact of parameters on GeoLife data with personal $\textbf{M}$: }
{\small (a)(b) Impact of $\epsilon$ and $\delta$ on size of $\Delta\textbf{X}$;}
{\small (c)(d) Impact of $\epsilon$ and $\delta$ on drift ratio;}
{\small (e)(f) Impact of $\epsilon$ and $\delta$ on distance.}
}
\label{Figure-impact-GeoLife-personal}
\end{figure}

\begin{figure}[!ht]
\begin{subfigure}{0.233\textwidth}
\centering
\includegraphics[width=4.2cm]{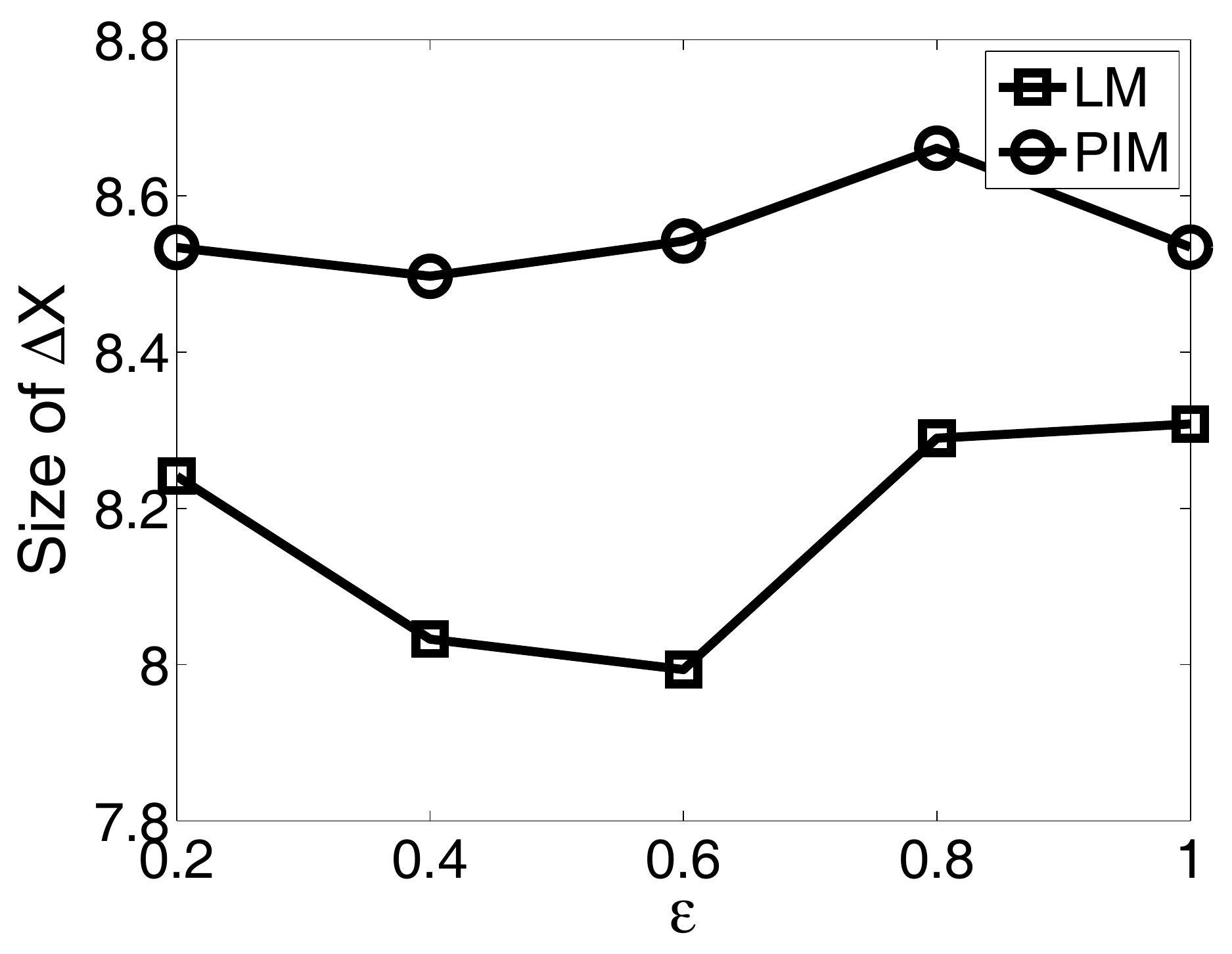}
\caption{{\small Size vs. $\epsilon$ }}
\label{Figure-epsilon-size-5}
\end{subfigure}
\begin{subfigure}{0.233\textwidth}
\centering
\includegraphics[width=4.2cm]{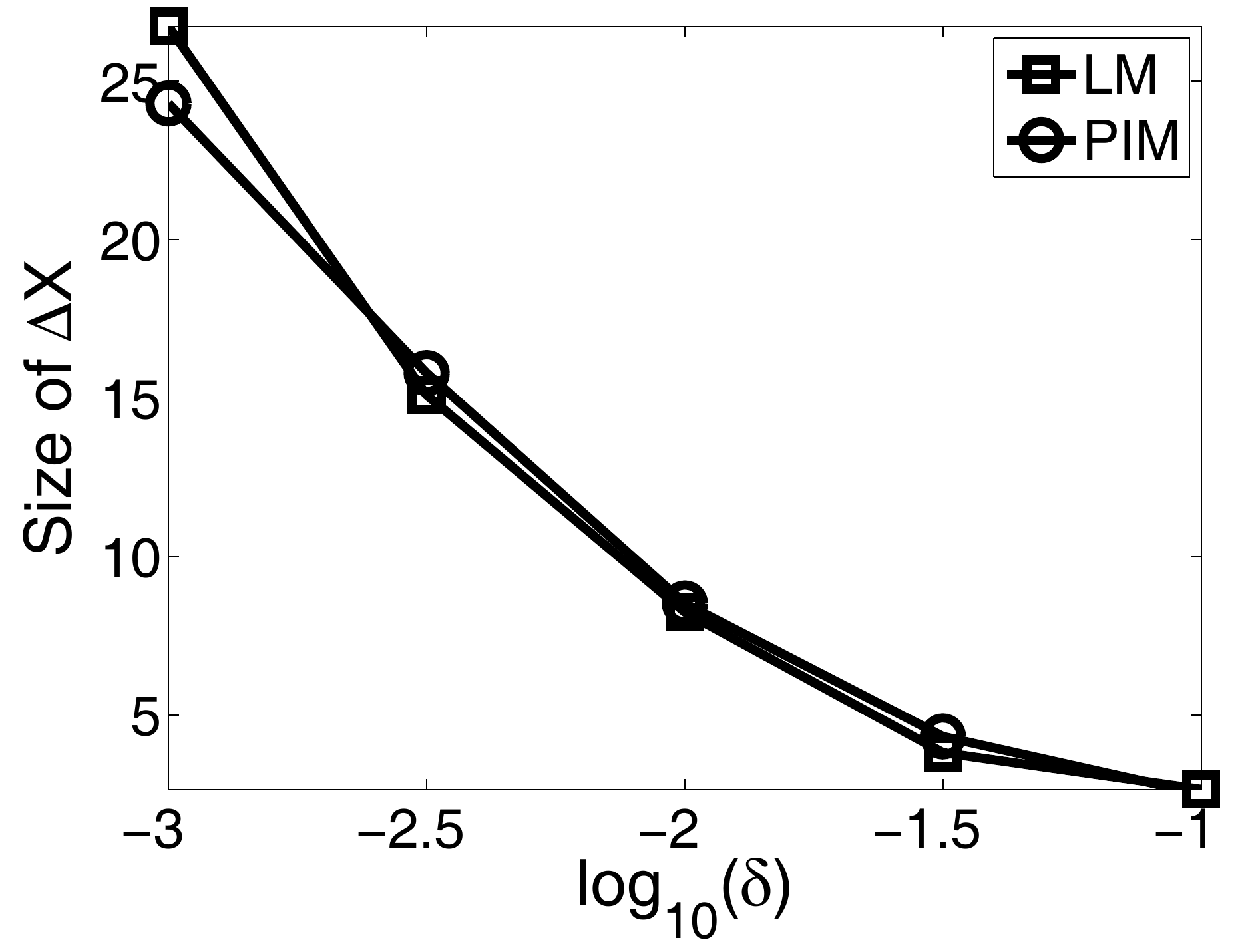}
\caption{{\small Size vs. $\delta$ }}
\label{Figure-delta-size-5}
\end{subfigure}
\begin{subfigure}{0.233\textwidth}
\centering
\includegraphics[width=4.2cm]{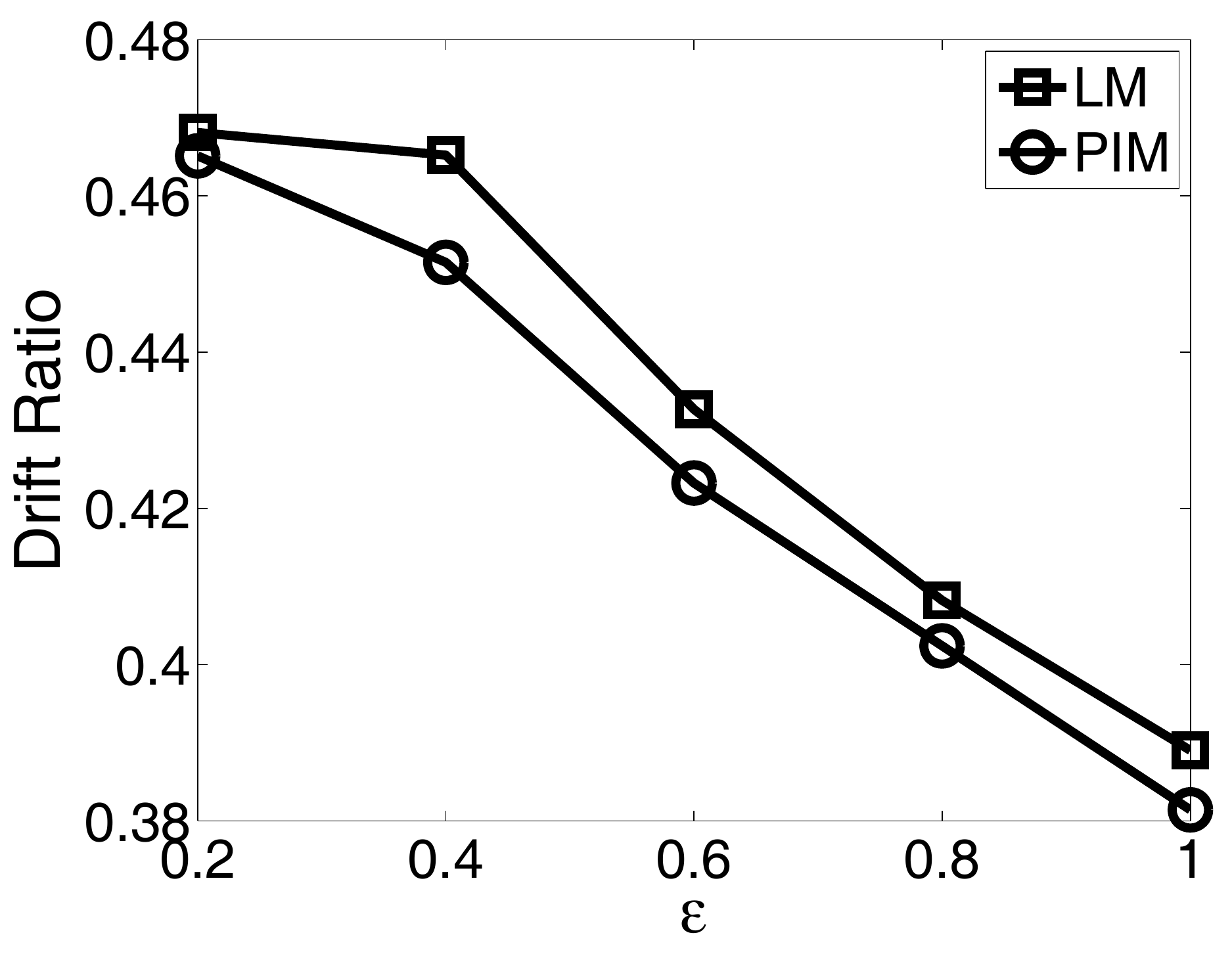}
\caption{{\small Drift Ratio vs. $\epsilon$ }}
\label{Figure-epsilon-drift-5}
\end{subfigure}
\begin{subfigure}{0.233\textwidth}
\centering
\includegraphics[width=4.2cm]{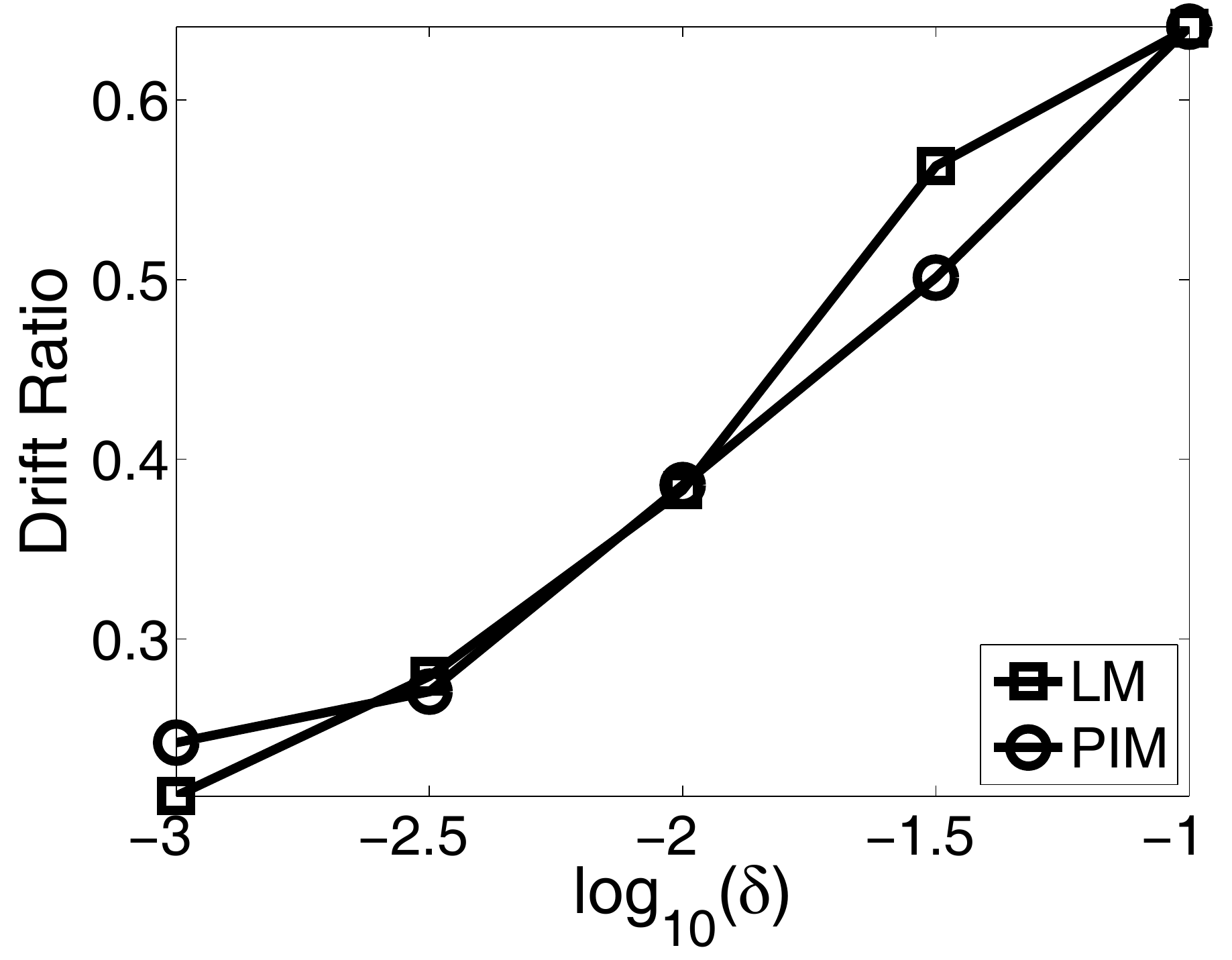}
\caption{{\small Drift Ratio vs. $\epsilon$ }}
\label{Figure-delta-drift-5}
\end{subfigure}
\begin{subfigure}{0.233\textwidth}
\centering
\includegraphics[width=4.2cm]{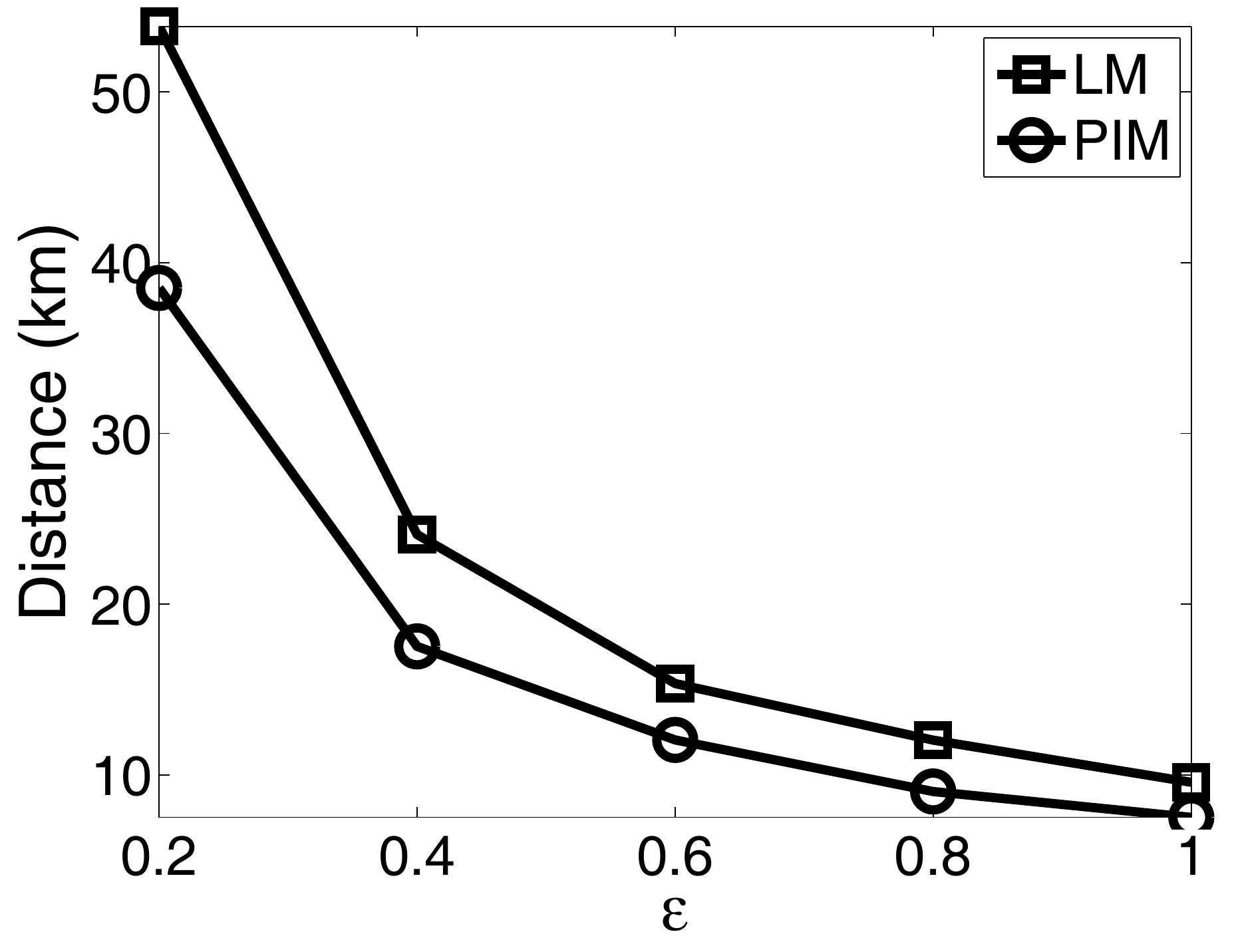}
\caption{{\small Distance vs. $\epsilon$ }}
\label{Figure-epsilon-Var-5}
\end{subfigure}
\begin{subfigure}{0.233\textwidth}
\centering
\includegraphics[width=4.2cm]{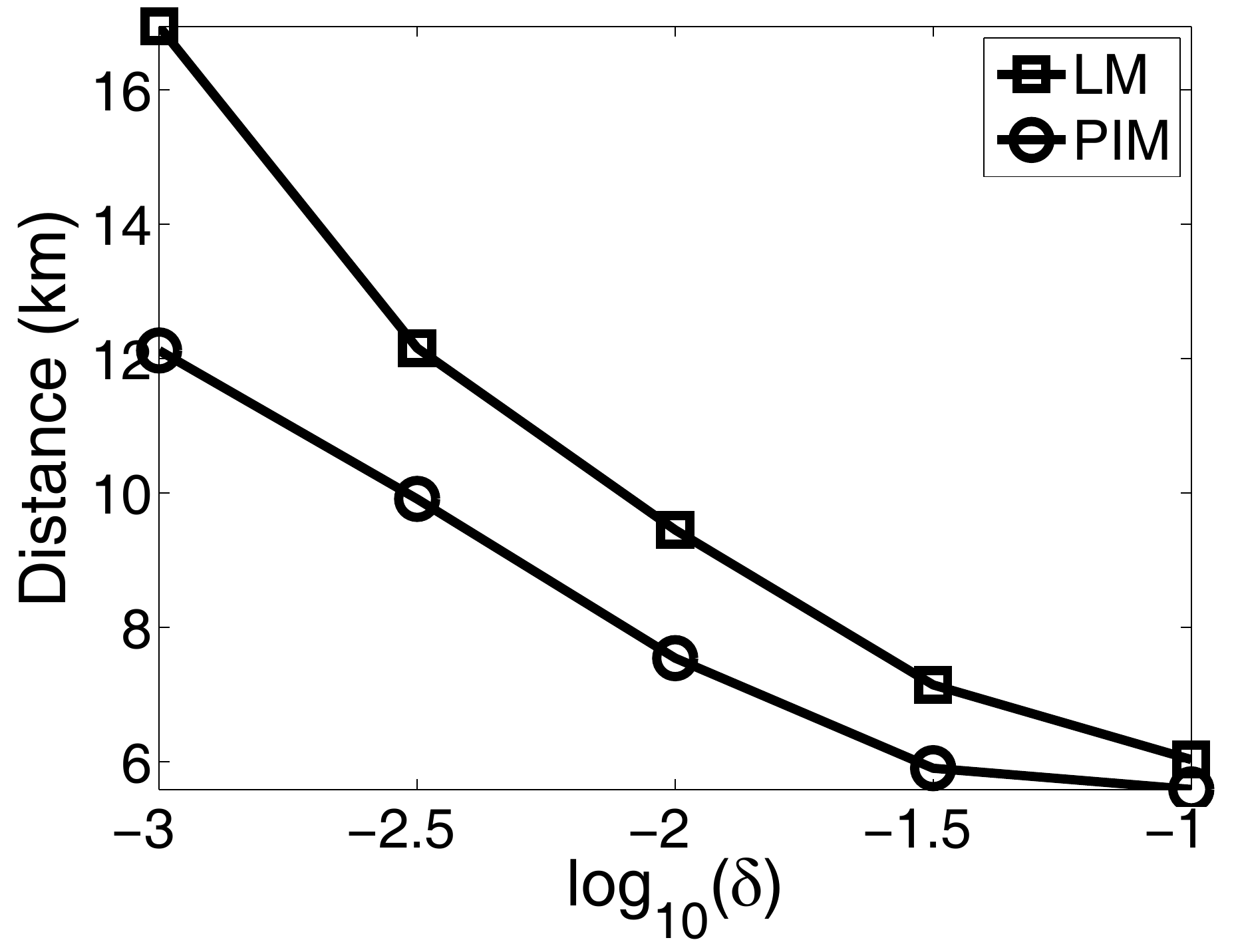}
\caption{{\small Distance vs. $\delta$ }}
\label{Figure-delta-Var-5}
\end{subfigure}
\caption{{\small Impact of parameters on Gowalla data with popular $\textbf{M}$: }
{\small (a)(b) Impact of $\epsilon$ and $\delta$ on size of $\Delta\textbf{X}$;}
{\small (c)(d) Impact of $\epsilon$ and $\delta$ on drift ratio;}
{\small (e)(f) Impact of $\epsilon$ and $\delta$ on distance.}
}
\label{Figure-impact-Gowalla-popular}
\end{figure}

\subsection{Performance Over Time}
In order to show the performance of a release mechanism as a user moves over time, including how $\Delta\textbf{X}$ changes, how often drift happens and how accurate is the perturbed location, we first run a set of experiments for a single test trajectory with popular $\textbf{M}$ learned from all users. We selected a random test trajectory from Geolife dataset consisting of $500$ timestamps.  We tested both PIM and LM at each timestamp with $\epsilon=1$ and $\delta=0.01$. Each method was run $20$ times and the average is reported.
Figure \ref{Figure-example-trace0}  shows the original trajectory in map and state (grid) coordinates; Figures \ref{Figure-aTraj-traj-LM-001} and \ref{Figure-aTraj-traj-IM-001}  show the released (perturbed) locations at each timestamp.
We can see that the released locations of PIM is closer to the true location, compared with LM.
%
%

\vspace{2mm}\noindent{\bf Size of $\Delta\textbf{X}$.}
From Figure \ref{Figure-aTraj-size-001} we see that the size of $\Delta\textbf{X}$ does not increase dramatically, instead it maintains at stable level after a few timestamps. The reason is that by selecting the $\delta$-location set the inference mechanism only boost probabilities of locations in $\Delta\textbf{X}$. Then the probabilities of other locations decay gradually. Thus a stable $\delta$-location set can be maintained.
%

\vspace{2mm}\noindent{\bf Drift Ratio.}
In Figure \ref{Figure-aTraj-drift-001}, the peak of drift ratio happened in timestamp $200\sim 300$. This can be explained by the fact that the true trajectory has a turning corner as in Figure \ref{Figure-example-trace0}, and the transition probability of making this right turn is relatively small in the Markov model.

When a drift happens, we use surrogate for release mechanisms.
Because the surrogate is the nearest cell to the true location in $\Delta\textbf{X}$ and the release mechanism is based on the surrogate, the posterior probability of the surrogate will be boosted. Consequently, in the next timestamp the probability that $\Delta\textbf{X}$ includes the previous true location rises. This ``lagged catch-up'' can be verified by Figures \ref{Figure-aTraj-dist-001}, \ref{Figure-aTraj-traj-LM-001} and \ref{Figure-aTraj-traj-IM-001}.


\vspace{2mm}\noindent{\bf Distance.}
The distance is reported in Figure \ref{Figure-aTraj-dist-001}. We can see that PIM provided more accurate locations than LM for two reasons. First, because PIM is optimal, the posterior probability distribution is more accurate than LM. Second, with such distribution a better  (Bayesian) inference can be obtained, making $\Delta\textbf{X}$ more accurate for the coming timestamp.



\subsection{Impact of Parameters}
Since the performance may vary for different trajectories, we chose $100$ trajectories from $100$ users, each of which has $500$ timestamps, to evaluate the overall performance and the impact of  parameters.
The default values are $\epsilon=1$ and $\delta=0.01$ if not mentioned.
The average performances for both datasets are reported in Figures \ref{Figure-impact-GeoLife-popular} (on GeoLife data with popular $\textbf{M}$), Figure \ref{Figure-impact-GeoLife-personal} (on GeoLife data with personal $\textbf{M}$) and Figure \ref{Figure-impact-Gowalla-popular} (on Gowalla data with popular $\textbf{M}$).

\vspace{2mm}\noindent{\bf Size of $\Delta\textbf{X}$ vs. $\boldsymbol\epsilon$.}
In Figures \ref{Figure-epsilon-size-4} and \ref{Figure-epsilon-size-41} (Geolife data), size of $\Delta\textbf{X}$ shrinks with larger $\epsilon$ because the inference result is enhanced by big $\epsilon$. On the other hand, impact of $\epsilon$ would be negligible in Gowalla data because one-step transition in Markov model has limited predictability (check-ins are not frequent), as in Figure \ref{Figure-epsilon-size-5}.

\vspace{2mm}\noindent{\bf Size of $\Delta\textbf{X}$ vs. $\boldsymbol\delta$.}
Size of $\Delta\textbf{X}$ is mainly determined by $\delta$ as shown in Figures \ref{Figure-delta-size-4}, \ref{Figure-delta-size-41} and \ref{Figure-delta-size-5}.
Note that LM and PIM have similar size of $\Delta\textbf{X}$, meaning the true location is hidden in the similar size of candidates.
When $\delta$ grows, size of $\Delta\textbf{X}$ reduces dramatically because more improbable locations are truncated. However, $\delta$ cannot be too large because it preserves nearly no privacy if size of $\Delta\textbf{X}$ is close to $1$.
Thus we use $\delta=0.01$ by default, which guarantees the sizes of $\Delta\textbf{X}$ are larger than $4$ in the three settings.

\vspace{2mm}\noindent{\bf Drift Ratio vs. $\boldsymbol\epsilon$.}
Figures \ref{Figure-epsilon-drift-4} and \ref{Figure-epsilon-drift-5} show that drift ratio declines with larger $\epsilon$, which is easy to understand because larger $\epsilon$ provides more accurate release. However,
the impact of $\epsilon$ is not obvious in Figure \ref{Figure-epsilon-drift-41}. The reason is that the size of $\Delta\textbf{X}$ is already small as in Figure \ref{Figure-delta-size-41}, hence the increase of $\epsilon$
does not help much in improving the accuracy of the inference.

\vspace{2mm}\noindent{\bf Drift Ratio vs. $\boldsymbol\delta$.}
Figures \ref{Figure-delta-drift-4}, \ref{Figure-delta-drift-41} and \ref{Figure-delta-drift-5} show that drift ratio rises when $\delta$ increases due to reduced $\Delta\textbf{X}$,
and PIM is slightly better than LM. However, due to the phenomenon of ``lagged catch-up'',
we will see next that the accuracy of the released locations was still improved with increasing $\delta$.

\vspace{2mm}\noindent{\bf Distance vs. $\boldsymbol\epsilon$.}
Figures \ref{Figure-epsilon-Var-4}, \ref{Figure-epsilon-Var-41} and \ref{Figure-epsilon-Var-5} show the distance with varying $\epsilon$. We can see that PIM performed better than LM. In Gowalla data, because check-in locations are far away from each other, the distance is larger than Geolife.


\vspace{2mm}\noindent{\bf Distance vs. $\boldsymbol\delta$.}
Because bigger $\delta$ will result in less candidates in $\delta$-location set, the distance declines when $\delta$ increases. Figures \ref{Figure-delta-Var-4}, \ref{Figure-delta-Var-41} and \ref{Figure-delta-Var-4} show that PIM achieves better accuracy than LM. However, from $10^{-1.5}$ to $10^{-1}$, the improvement on distance is very small while privacy guarantee drops significantly as in Figures \ref{Figure-delta-size-4}, \ref{Figure-delta-size-41} and \ref{Figure-delta-size-5}.
Especially in Figure \ref{Figure-delta-size-41}, size of $\Delta\textbf{X}$ is $1$ when $\delta=0.1$.
Therefore, choosing a high value of $\delta$ (like $\delta>0.03$) does not provide the best trade-off of privacy and utility.

\vspace{2mm}\noindent{\bf Impact of Markov model.} Comparing Figures \ref{Figure-impact-GeoLife-popular} on popular $\textbf{M}$ and Figure \ref{Figure-impact-GeoLife-personal} on personal $\textbf{M}$, we can see the impact of different Markov model. With more accurate (personal) model, better utility can be achieved, including smaller size of $\Delta\textbf{X}$, lower drift ratio and less distance. However, the same privacy level ($\epsilon$-differential privacy) is maintained (on different $\Delta\textbf{X}$) regardless of $\textbf{M}$.

\subsection{Utility for Location Based Queries}
To demonstrate the utility of released locations, we also measured the precision and recall of $k$NN queries at each of the 500 timestamps in the 100 trajectories with popular $\textbf{M}$. The average results of $k$NN from original locations and $k'$NN from released locations
are reported in Figure \ref{Figure-KNN} with $\epsilon=1$ and $\delta=0.01$.

In Figures \ref{Figure-KNN-1} and \ref{Figure-KNN-2}, we show the precision and recall with $k=k'$. Note that in this case precision is equal to recall. We can see that when $k$ grows precision and recall also increase because the nearest neighbors have to be found in larger areas. PIM is consistently better than LM.

Next we fixed $k=5$ and varied $k'$. Figures \ref{Figure-KNN-3} and \ref{Figure-KNN-4} show the precision drops when $k'$ rises because of a larger returned set. On the other hand, Figures \ref{Figure-KNN-5} and \ref{Figure-KNN-6} indicate recall increases with large $k'$. Overall, PIM has better precision and recall than LM.
\begin{figure}[!t]
\begin{subfigure}{0.233\textwidth}
\centering
\includegraphics[width=4.2cm]{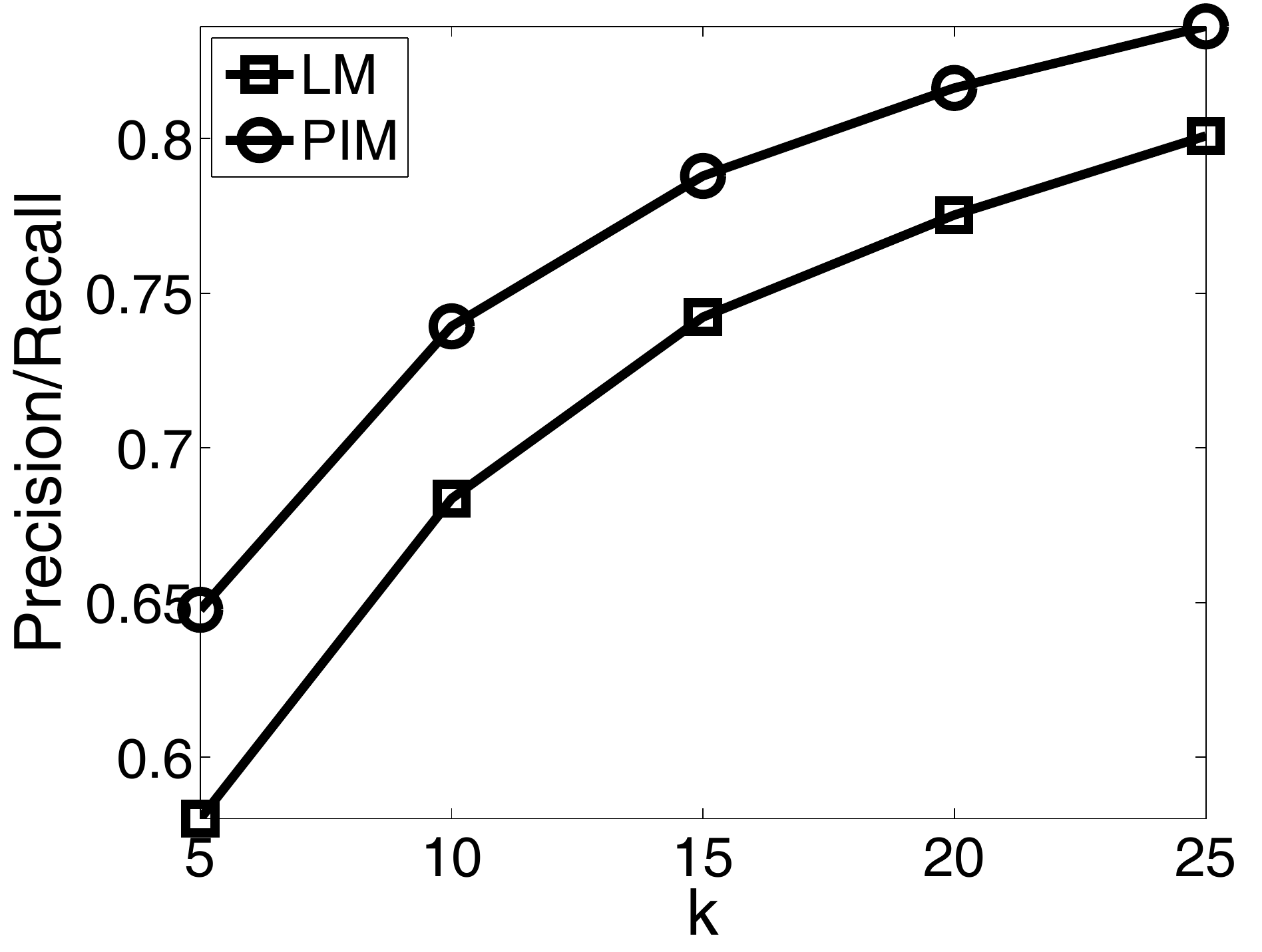}
\caption{{\small Precision/Recall (Geolife)}}
\label{Figure-KNN-1}
\end{subfigure}
\begin{subfigure}{0.233\textwidth}
\centering
\includegraphics[width=4.2cm]{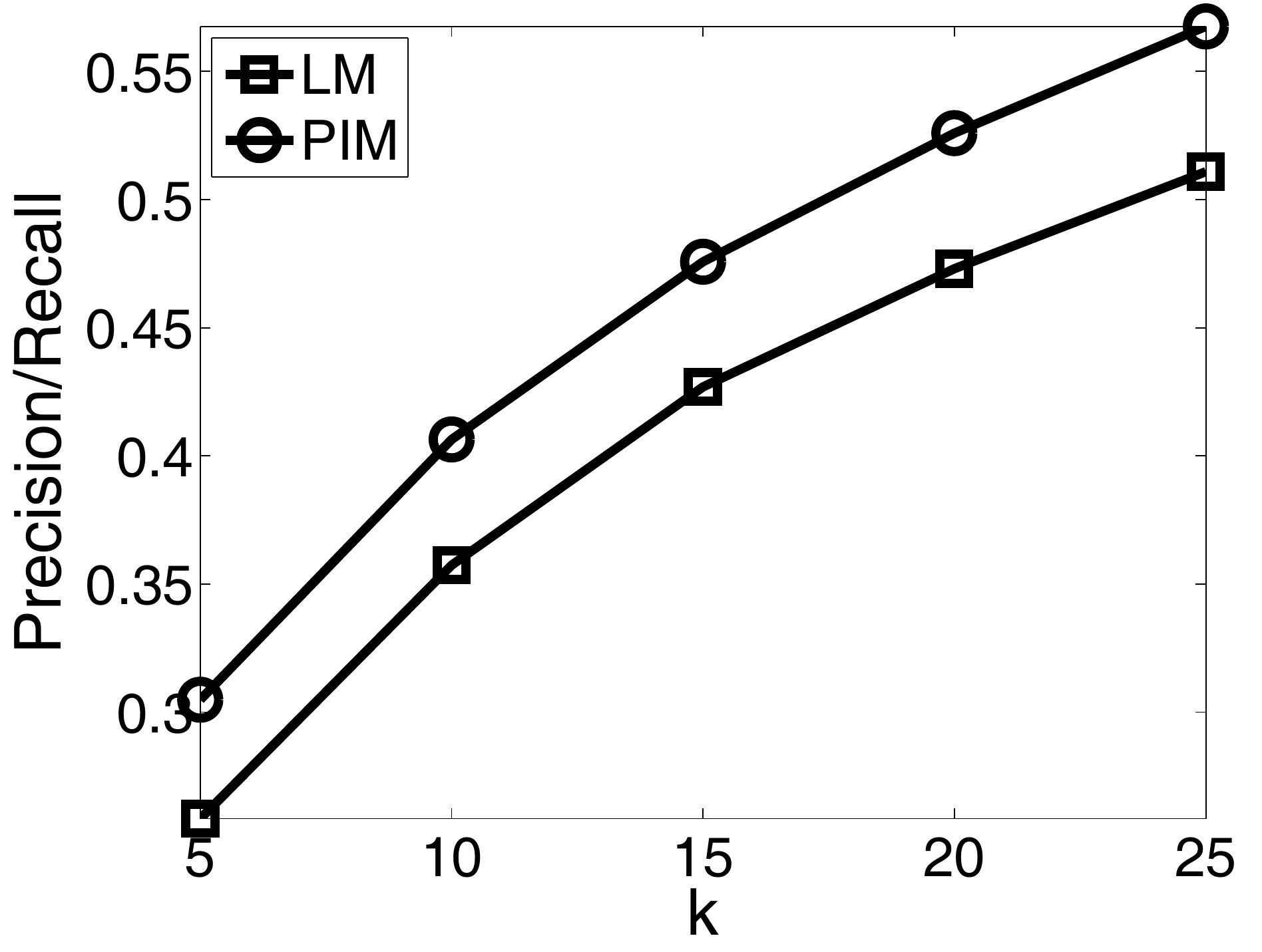}
\caption{{\small Precision/Recall (Gowalla)}}
\label{Figure-KNN-2}
\end{subfigure}
\begin{subfigure}{0.233\textwidth}
\centering
\includegraphics[width=4.2cm]{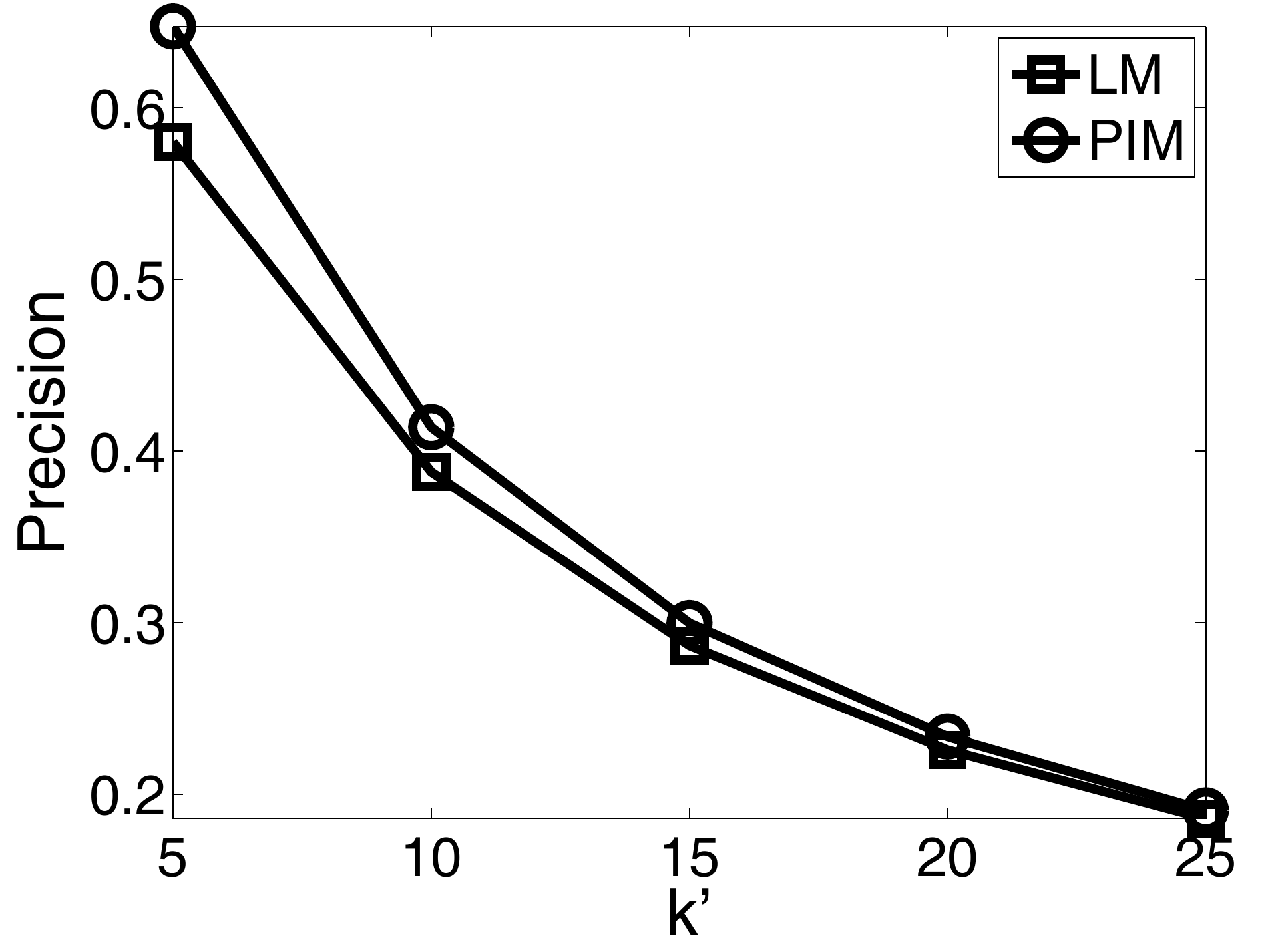}
\caption{{\small Precision (Geolife)}}
\label{Figure-KNN-3}
\end{subfigure}
\begin{subfigure}{0.233\textwidth}
\centering
\includegraphics[width=4.2cm]{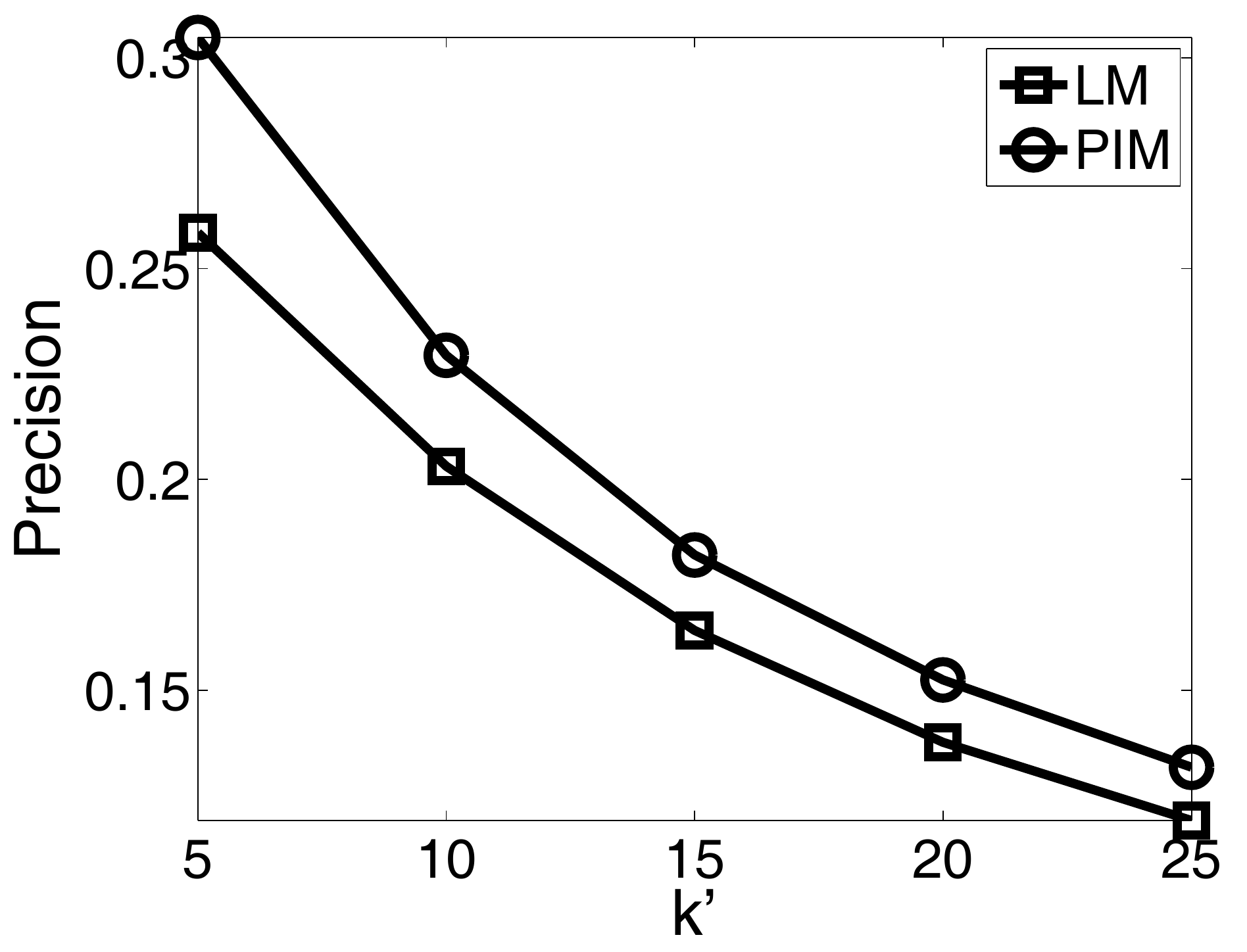}
\caption{{\small Precision (Gowalla)}}
\label{Figure-KNN-4}
\end{subfigure}
\begin{subfigure}{0.233\textwidth}
\centering
\includegraphics[width=4.2cm]{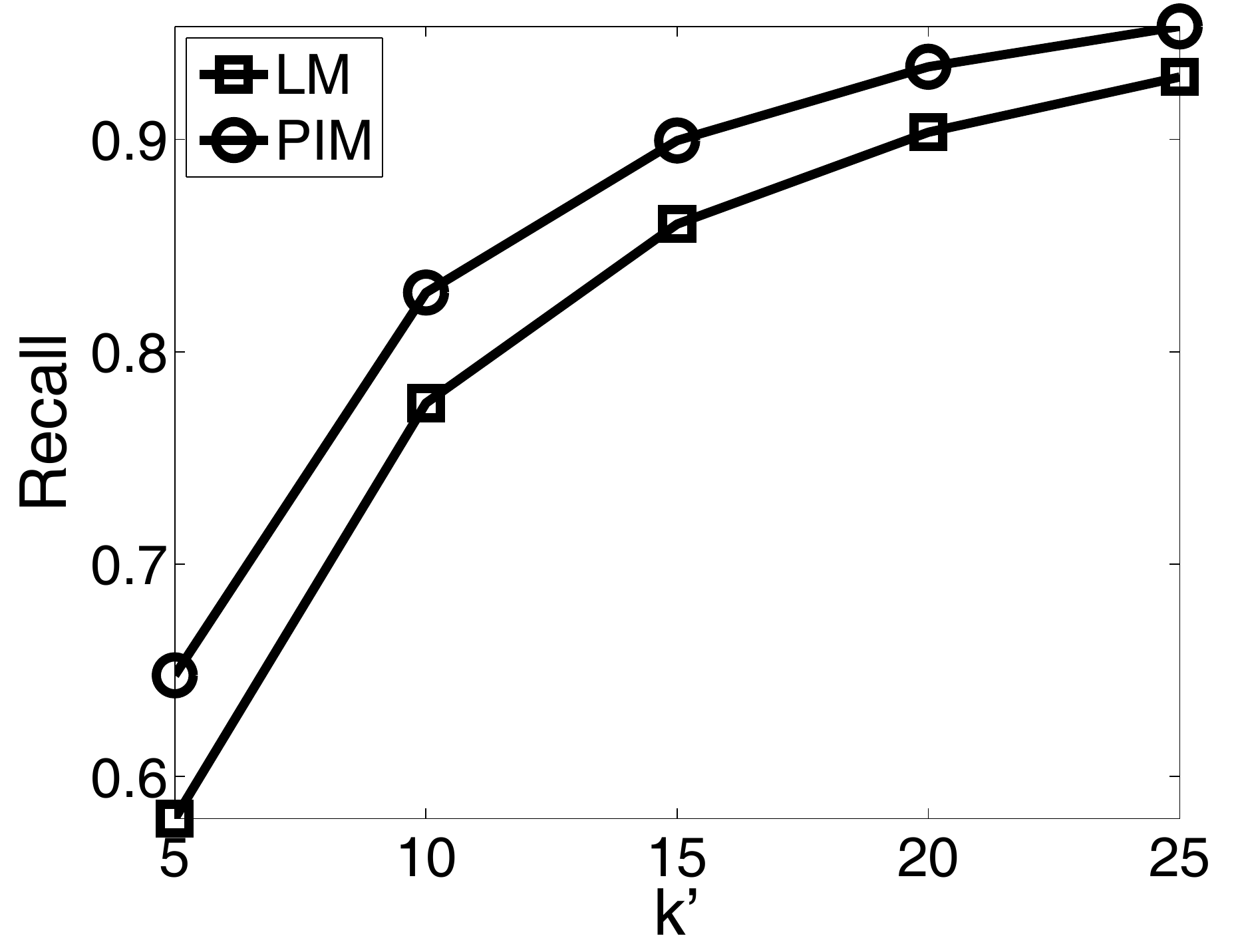}
\caption{{\small Recall (Geolife)}}
\label{Figure-KNN-5}
\end{subfigure}
\begin{subfigure}{0.233\textwidth}
\centering
\includegraphics[width=4.2cm]{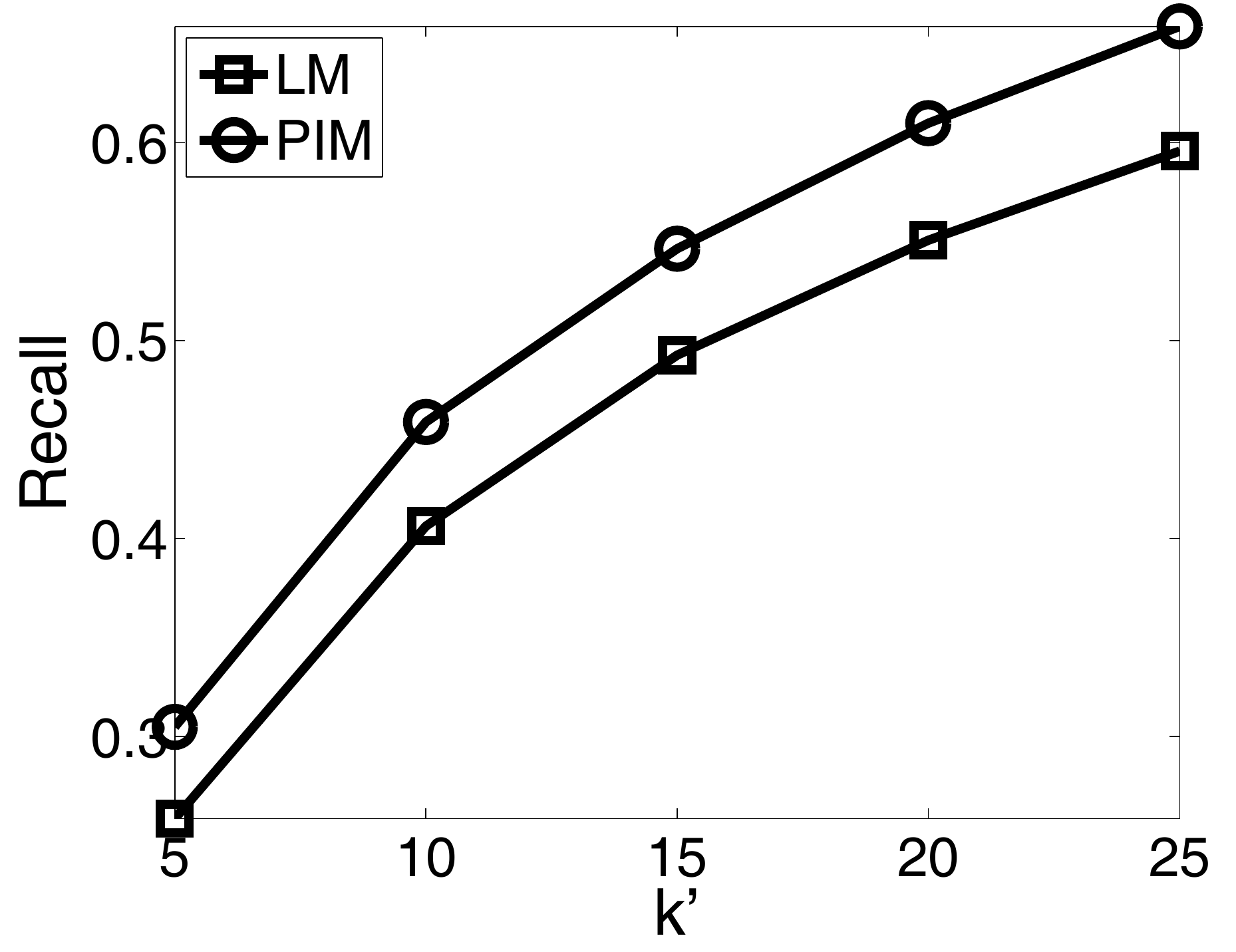}
\caption{{\small Recall (Gowalla)}}
\label{Figure-KNN-6}
\end{subfigure}
\caption{{\small $k$NN results: }
{\small (a)(b) precision and recall under $k=k'$;}
{\small (c)(d) precision vs. $k'$;}
{\small (e)(f) recall vs. $k'$.}
}
\label{Figure-KNN}
\end{figure}

%% file: data/relatedwork_Li.tex
\section{Related Works}
\subsection{Location Privacy}
There is a rich set of literature related to location privacy. A few recent books and surveys
\cite{krumm2009survey, ghinita2013privacy} provide an up-to-date review of Location Privacy Preserving Mechanisms (LPPMs).
%

%
LPPMs
generally use obfuscation methods, such as spatial cloaking, cell merging, location precision reduction or dummy cells, to achieve anonymity based privacy or uncertainty based privacy.  However, anonymity or ad hoc
uncertainty based techniques do not always provide sufficient privacy protection
\cite{
trace-attack2, Quantifying-location-privacy-SP2011}.  Most of them do not consider the temporal correlations between locations and are subject
to various inference attacks.
 The recent work \cite{geo-indistinguishability-CCS13} proposed a notion of geo-indistinguishability which extends differential privacy.  However, a fundamental difference is that neighboring pairs or secrets are defined based on a radius and
it does not consider the temporal correlations of multiple locations.

Several works use Markov models for modeling users'
mobility and inferring user locations or trajectories
\cite{Liao-learning-pattern, putmode_Qiao10}.
\cite{MaskIt-SIGMOD12} proposed an insightful technique with a provable privacy guarantee to filter a user context stream even if the adversaries are powerful and have knowledge about the temporal correlations but it used suppression instead of perturbation.
\cite{Quantifying-location-privacy-SP2011} investigated the question of how to
formally quantify the privacy of existing LPPMs and assumed that an adversary can
model users' mobility using a Markov chain learned from a population.

\subsection{Differential Privacy}
Several variants or generalizations of differential privacy have been studied, as discussed in Section \ref{sec-comparison}.
However, applying differential privacy for location protection has not been investigated in depth.
Several recent works have applied differential privacy to publish {\em aggregate} information from a large volume of location, trajectory or spatiotemporal data (e.g. \cite{dp-trajectory-KDD12,
DBLP:conf/icde/QardajiYL13,
DBLP:conf/dbsec/FanXS13, Copula-haoran-EDBT14}).  Our contribution is to extend differential privacy in a new setting of continual location sharing of only one user whose locations are temporally correlated.
Optimal query answering under differential privacy has been studied recently. Hardt and Talwar \cite{Geometry-Hardt-STOC10} studied the theoretical lower bound for any differentially private mechanisms and proposed $K$-norm mechanism.
Bhaskara et al \cite{Bhaskara-bound-STOC12} studied another $K$-norm based method to project the sensitivity hull onto orthogonal subspaces. Nikolov et al \cite{Nikolov-geometry-STOC13} also improved the efficiency of $K$-norm mechanism by finding the minimal enclosing ellipsoid to release multivariate Gaussian noises. So far the best utility of existing mechanisms can be $log(d)$ approximately optimal. We extended the $K$-norm mechanism to location data by examining the two-dimensional sensitivity hull and designing its isotropic transformation so that optimal utility can be achieved.

\section{Conclusion and Future work}
In this paper we proposed $\delta$-location set based differential privacy to protect a user's true location at every timestamp under temporal correlations.
We generalized the notion of  ``neighboring databases'' to $\delta$-location set for the new setting and extended the well known $\ell_1$-norm sensitivity to sensitivity hull in order to capture the geometric meaning of sensitivity.
 Then with sensitivity hull we derived the lower bound of $\delta$-location set based differential privacy. To achieve the lower bound, we designed the optimal planar isotropic mechanism to release differentially private locations with significantly high efficiency and utility.
%

The framework of  $\delta$-location set based differential privacy can work with any mobility models (besides Markov chain).  As a future work direction, we are interested in instantiating it with different and more advanced mobility models and studying the impact.

\section{Acknowledgement}
This research is supported by NSF under grant No. 1117763 and the AFOSR DDDAS program under grant FA9550-12-1-0240. We thank the anonymous reviewers for their valuable comments that helped improve the final version of this paper.

%% file: data/locPriv_appendix.tex
\newpage
\section{Additional Material}

\subsection{Laplace Mechanism}
%
In the problem of location sharing, the query being asked is: at a timestamp where is the true location ``$\textbf{x}=?$''. We denote the true location $\textbf{x}=[\textbf{x}[1],\textbf{x}[2]]^T$ where $\textbf{x}[1]$ and $\textbf{x}[2]$ are the two coordinates of $\textbf{x}$ on a map.
By Definition \ref{def-standard-sensitivity}, the sensitivity is
\begin{align*}
\Delta f=\mathop{max}\limits_{\textbf{x}_1,\textbf{x}_2\in \Delta\textbf{X}}||\textbf{x}_1-\textbf{x}_2||_1\\
=\mathop {max}\limits_{\textbf{x}_1,\textbf{x}_2\in \Delta\textbf{X}}\left(\left| \textbf{x}_1[1]-\textbf{x}_2[1] \right|+\left|\textbf{x}_1[2]-\textbf{x}_2[2]\right|\right)
\end{align*}
Figure \ref{Figure-naive} shows an example. The true location $\textbf{x}_t^*$ is at the ``$\star$'' position and all the dots constitute the indistinguishable set $\Delta\textbf{X}$.
We can see that the maximum value of above equation is achieved when $\textbf{x}_1$ and $\textbf{x}_2$ are the left-top and right-bottom points. The sensitivity can be calculated as $\Delta_1+\Delta_2$ where $\Delta_1$ and $\Delta_2$ are the differences of $\textbf{x}_1$ and $\textbf{x}_2$ on the two axes respectively. By Laplace Mechanism, the answer $[\textbf{x}[1]+Lap((\Delta_1+\Delta_2)/\epsilon),\textbf{x}[2]+Lap((\Delta_1+\Delta_2)/\epsilon)]$ is $\epsilon$-differentially private.
\begin{figure}
\begin{subfigure}{0.23\textwidth}
\centering
\includegraphics[width=4.2cm]{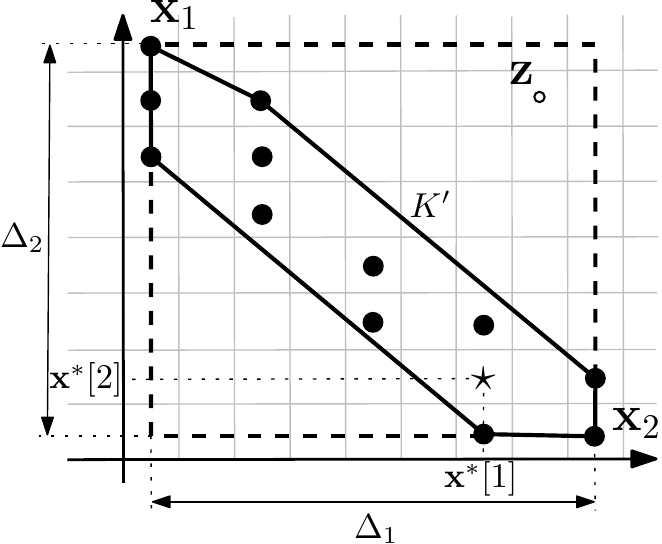}
\caption{}
\label{Figure-naive}
\end{subfigure}
\begin{subfigure}{0.23\textwidth}
\centering
\vspace{-5mm}
\includegraphics[width=4.3cm]{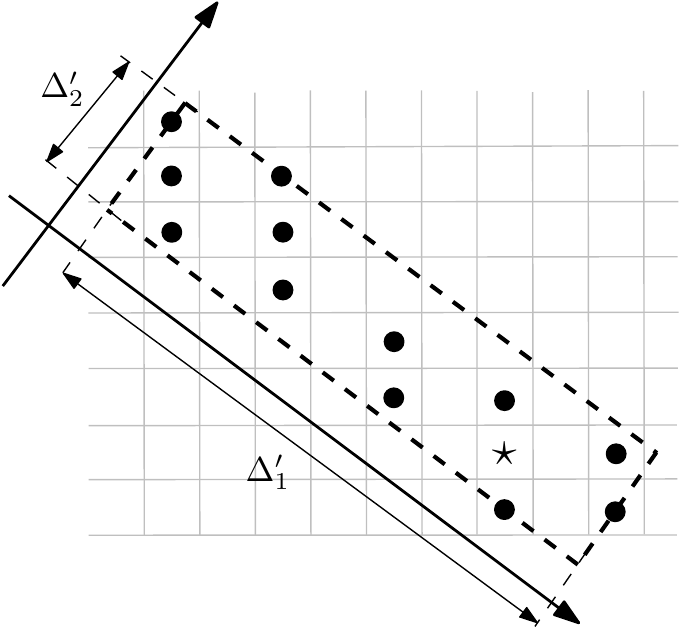}
\caption{}
\label{Figure-naive2}
\end{subfigure}
\caption{{\small (a). Laplace Mechanism.} {\small (b). a rotated Laplace Mechanism.}}
\end{figure}
%

We name the above method ``baseline Laplace Mechanism'' and summarize it in Algorithm \ref{alg-naive}.  We skip the proofs for the following theorems as they are straightforward.
\begin{algorithm}[htb]
\caption{Laplace Mechanism}
\begin{algorithmic}[1]
\Require{
$\epsilon$, $\Delta\textbf{X}$, $\textbf{x}^*$
}
\State{$\Delta\textbf{X}_1 \leftarrow \{\textbf{x}[1]|\textbf{x}\in\Delta\textbf{X}\}$;}
\Comment{project $\Delta\textbf{X}$ onto axis 1.}
\State{$\Delta\textbf{X}_2 \leftarrow \{\textbf{x}[2]|\textbf{x}\in\Delta\textbf{X}\}$;}
\Comment{project $\Delta\textbf{X}$ onto axis 2.}
\State{$\Delta_1 \leftarrow \mathop{max}\limits_{a,b\in\Delta\textbf{X}_1}|a-b|$;}
\Comment{sensitivity on axis 1.}
\State{$\Delta_2 \leftarrow \mathop{max}\limits_{a,b\in\Delta\textbf{X}_2}|a-b|$;}
\Comment{sensitivity on axis 2.}\\
\Return{$[\textbf{x}^*[1]+Lap(\frac{\Delta_1+\Delta_2}{\epsilon}),\textbf{x}^*[2]+Lap(\frac{\Delta_1+\Delta_2}{\epsilon})]^T$}
\end{algorithmic}
\label{alg-naive}
\end{algorithm}

\begin{theorem}
\label{theo-error-LM}
Algorithm \ref{alg-naive} is $\epsilon$-differentially private.
\end{theorem}

\begin{theorem}
\label{theo-time-LM}
Algorithm \ref{alg-naive} takes $O(n)$ time where $n$ is the number of points in indistinguishable set $\Delta\textbf{X}$.
\end{theorem}

\vspace{2mm}\noindent{\bf Utility Analysis} We can compute the utility of Laplace Mechanism with Equation \ref{equation-util}.
\begin{align*}
\textsc{Error}=\sqrt{2\mathbb{E}(Lap((\Delta_1+\Delta_2)/\epsilon))}
=\Theta(\frac{\Delta_1+\Delta_2}{\epsilon})
\end{align*}

Because Laplace Mechanism considers the two dimensions independently, utility of the perturbed location may be significantly reduced. For example, in \ref{Figure-naive}, it has a high probability to release point $\textbf{z}$ because $\textbf{z}[1]$ is near $\textbf{x}[1]$ (due to the Laplace mechanism). However, the position of $\textbf{z}$ is far from $\textbf{x}$. If we draw the convex hull of $\Delta\textbf{X}$ as $K'$, shown as the black polygon, $\textbf{z}$ is far from $K'$. Intuitively, we do not have to allocate a high perturbation probability at $\textbf{z}$. In general, if the released point is near $K'$, it has better utility because $\textbf{x}^*$ is in $K'$. This motivates us to consider the two dimensions together to obtain better utility.

One may wonder what happens if we rotate the map with a degree. For example, in Figure \ref{Figure-naive2}, if we rotate the map with the two axes, the new sensitivity  $(\Delta_1' + \Delta_2')$ in the rotated space is significantly reduced, hence the error is much smaller than that in Figure \ref{Figure-naive} because $\textsc{Error}$ is proportional to $(\Delta_1+\Delta_2)$.
Then the question is how to find an optimal rotation.
Or fundamentally  what is the optimal utility we can obtain?
To answer these questions, 
we have to go back to Section \ref{sec-sensitivityhull} to investigate the sensitivity hull and the error bound.